\documentclass[reqno,11pt]{amsart}
\usepackage[utf8]{inputenc}
\usepackage{graphicx}
\usepackage{slashed}
\usepackage{amscd}
\usepackage{amssymb}
\usepackage{esint}
\usepackage{enumitem}

\usepackage{pst-node}
\usepackage{tikz-cd}

\usepackage{comment}
\usepackage{amsmath}
\usepackage{dsfont}
\usepackage[mathscr]{eucal}
\ifx\makepics\undefined
\usepackage[linktocpage]{hyperref}
\else
\usepackage[pspdf=-dALLOWPSTRANSPARENCY]{auto-pst-pdf} % langsame Version
\fi
\usepackage{epsfig}
\usepackage{pstricks}
\usepackage{pst-all}

\numberwithin{equation}{section}
\allowdisplaybreaks[1]
\definecolor{labelkey}{gray}{.65}

\author[F.\ Finster]{Felix Finster}
\address{Fakult\"at f\"ur Mathematik \\ Universit\"at Regensburg \\ D-93040 Regensburg \\ Germany}
\email{finster@ur.de}

\author[M. van den Beld-Serrano]{Marco van den Beld-Serrano \\ \\ April 2025}
\address{Fakult\"at f\"ur Mathematik \\ Universit\"at Regensburg \\ D-93040 Regensburg \\ Germany}
\email{marco.van-den-beld-serrano@ur.de}

\newtheorem{Def}{Definition}[section]
\newtheorem{Thm}[Def]{Theorem}
\newtheorem{Prp}[Def]{Proposition}
\newtheorem{Lemma}[Def]{Lemma}
\newtheorem{Remark}[Def]{Remark}

\newtheorem{Corollary}[Def]{Corollary}

\newcommand{\Thanks}{\vspace*{.5em} \noindent \thanks}
\newcommand{\beq}{\begin{equation}}
\newcommand{\eeq}{\end{equation}}
\newcommand{\Proof}{\begin{proof}}
\newcommand{\QED}{\end{proof} \noindent}
\newcommand{\QEDrem}{\ \hfill~$\Diamond$}

\newcommand{\Sl}{\mbox{$\prec \!\!$ \nolinebreak}}
\newcommand{\Sr}{\mbox{\nolinebreak$\succ$}}

\newcommand{\C}{\mathbb{C}}
\newcommand{\R}{\mathbb{R}}

\newcommand{\N}{\mathbb{N}}

\renewcommand{\H}{\mathscr{H}}

\newcommand{\bep}{\begin{pmatrix}}
\newcommand{\enp}{\end{pmatrix}}

\newcommand{\scrL}{\myscr L}
\newcommand{\Tgamma}{\tilde{\gamma}}

\DeclareMathOperator*{\slim}{s-lim}

\newcommand{\pr}{\partial}
\newcommand{\gammaeta}{\gamma_\eta}

\DeclareFontFamily{OT1}{rsfso}{}
\DeclareFontShape{OT1}{rsfso}{m}{n}{ <-7> rsfso5 <7-10> rsfso7 <10-> rsfso10}{}
\DeclareMathAlphabet{\myscr}{OT1}{rsfso}{m}{n}

\renewcommand{\Tr}{\text{\rm{Tr}}}
\DeclareMathOperator{\tr}{tr}

\newcommand{\bitem}{\begin{itemize}[leftmargin=2em]}
\newcommand{\eitem}{\end{itemize}}

\newcommand{\coloneqq}{:=}
\newcommand{\gammag}{\gamma_{g}}

\title{Baryogenesis in Conformally Flat Spacetimes}

\begin{document}
\begin{abstract}
Based on a baryogenesis mechanism originating from the theory of causal fermion systems, we analyze its main geometric and analytic features in conformally flat spacetimes. An explicit formula is derived for the rate of baryogenesis in these spacetimes, which depends on the mass $m$ of the particles, the conformal factor $\Omega$ and a future directed timelike vector field $u$ (dubbed the regularizing vector field). Our analysis covers Friedmann-Lema{\^i}tre-Robertson-Walker, Milne and Milne-like spacetimes. It sets the ground for concrete, quantitative predictions for specific cosmological spacetimes.
\end{abstract}

\maketitle

\tableofcontents
\section{Introduction}
This paper is part of a series of papers devoted to the exploration of a novel mechanism of baryogenesis.
This mechanism was proposed in~\cite{baryogenesis} based on the structures of the dynamical equations
in the theory of causal fermion systems. In~\cite{baryomink}, this mechanism was worked out in more
detail in Minkowski space. In the present paper, we extend this analysis to conformally flat spacetimes.

We now introduce the setting and summarize our results (a general introduction to baryogenesis
and a discussion of our mechanism of baryogenesis can be found in~\cite{baryogenesis, baryomink}).
We consider a four-dimensional Lorentzian spacetime~$(M,g)$ with trivial topology~$M = \R \times \R^3$,
which is conformally flat. This means that, denoting the coordinates by~$t \in \R$ and~$(x,y,z) \in \R^3$,
the metric can be written as a conformal factor times the Minkowski metric,
\beq \label{confmink}
ds^2 = g_{jk}\, dx^j\, dx^k = \Omega^2(t,x,y,z)\: \big( dt^2 - dx^2-dy^2-dz^2 \big) \:,
\eeq
where~$\Omega$ is a strictly positive, smooth function in spacetime.
The class of conformally flat spacetimes includes many physically interesting examples like
Friedmann-Lema{\^i}tre-Robertson-Walker (FLRW), Milne and Milne-like spacetimes.
Our mechanism of baryogenesis is based on a modification of the Dirac dynamics.
In order to model this modification, we write the Dirac equation in the Hamiltonian form
\[ i \partial_t \psi = H_g \psi \:, \]
where~$H_g$ is the {\em{Dirac Hamiltonian}} in the Lorentzian metric~$g$
(for basics on Dirac spinors see the preliminaries in Section~\ref{sec:math_intro}).
Moreover, we introduce the so-called {\em{regularization vector field}}~$u$ as additional physical input. It is a timelike vector field whose time evolution obeys the  {\em{locally rigid dynamics}}
(see Definition~\ref{def:locally_rigid_dynamics}). We form the {\em{symmetrized Hamiltonian}} by
\[ A_t :=
\frac{1}{4}\{u^{0},H_{g}+H_{g}^{\ast}\}+\frac{i}{4} \sum_{\mu=1}^3
\{u^{\mu},\nabla^{s}_{\mu}-(\nabla^{s}_{\mu})^{\ast} \} \]
(where~$\nabla^s$ denotes the spinorial Levi-Civita connection in spacetime and~$\{ .,. \}$ is the anti-commutator);
it is an essentially selfadjoint operator on the Hilbert space~$\H_{t,g}$ of square-integrable Dirac spinors of the
Cauchy surface~$N_t$ (for details see again Section~\ref{sec:math_intro}).
Finally, we introduce the operator~$\tilde{H}_\eta$ as the Dirac Hamiltonian of Minkowski space,
unitarily transformed to the Hilbert space~$\H_{t,g}$ in curved spacetime
(for detail see Lemma~\ref{lem:time_independent}).
The {\em{rate of baryogenesis}}~$B_t$ can be expressed in terms of the trace of a product of
operator which can be obtained via the spectral calculus from~$A_t$ and~$\tilde{H}_\eta$ 
(see Definition~\ref{def:Baryo}). Our main result is to analyze this formula and to bring it into a more tractable form.
To this end, we perform a perturbation expansion in powers of the operator~$\Delta A(t) := A_t - \tilde{H}_{\eta}$, which is in general a first order differential operator (unless if~$u=\partial_{t}$; then~$\Delta A(t)$ is simply a multiplication operator).
This is our main result.

\begin{Thm}\label{theo:baryo_rate}
Let~$(M,g)$ be a conformally flat spacetime with time coordinate~$t$ as in~\eqref{confmink}.
Moreover, assume that the operator~$A_t$ has an absolutely continuous spectrum, that~$\Delta A(t)$ has smooth and compactly supported coefficients (in~$N_{t}$) and that for all~$\omega$ in the resolvent set of~$A_t$ it holds that
    \begin{equation}\label{eq:assumption}
        \| R_\omega(\tilde{H}_\eta)\Delta A(t))\|<1
    \end{equation}
(where~$R_{\omega}(\tilde{H}_\eta) := (\tilde{H}_\eta - \omega)^{-1}$ denotes the resolvent).
Then, the rate of baryogenesis~$B_t$ admits a power expansion in which the zeroth and first order contributions vanish, whereas the second order contribution is given by
    \begin{equation}\label{eq:Bt2_final}
        B_t^{(2)}=-\int \frac{d^{3}k}{(2\pi)^{3}}\int \frac{d^{3}k'}{(2\pi)^{3}}\frac{1}{4\omega_{k}\omega_{k'}}\frac{1}{(\omega_{k'}+\omega_{k})^{2}}\:G_{\Omega,m,u}(k,k')\;,
    \end{equation}
    where~$\omega_{k}:=\sqrt{|\vec{k}|^{2}+m^{2}}$,~$\omega_{k'}:=\sqrt{|\vec{k'}|^{2}+m^{2}}$ and~$G_{\Omega,m,u} : \R^3 \times \R^3 \rightarrow\R$ is a smooth function which depends on the conformal factor~$\Omega$, on the mass~$m$ and on the vector field~$u: M\rightarrow TM$.
\end{Thm}

We note that the condition~\eqref{eq:assumption} means that~$\Delta A$ must be sufficiently small.
This condition makes it possible to expand the resolvent in powers of~$\Delta A$.
For technical simplicity, we restrict attention to the case that~$\Delta A$ is compactly supported on~$N_t$.
In particular, it implies that, outside a compact region~$V\subset \R^{4}$, the considered spacetime must agree with Minkowski spacetime and~$u$ with~$\partial_t$. So, formula~\eqref{eq:Bt2_final} describes the rate of baryogenesis for spacetimes whose geometry deviates locally from Minkowski spacetime and vector fields deviating
locally from~$\partial_t$. More generally, the derived formula can be understood as describing the \emph{density} of the rate of baryogenesis in general conformally flat spacetimes. Finally, the assumption that~$A_t$ has an absolutely continuous spectrum does not seem too restrictive since~\eqref{eq:assumption} already entails that~$A_t$ is a small perturbation of~$\tilde{H}_\eta$ and the Dirac Hamiltonian in Minkowski spacetime is known to have an absolutely continuous spectrum (cf.~\cite[Theorem 1.1]{Thaller}).

This paper is devoted to the proof of the above theorem. However, this will not be done in one single step since it requires a careful analysis of the analytic and geometric tools related to the study of the baryogenesis mechanism  in conformally flat spacetimes. For this reason, the claims of the previous theorem will be obtained in separate individual steps. In the first place, in Proposition~\ref{prop:Bt_conf_flat} (using arguments from~\cite{baryogenesis} and~\cite{baryomink}) we show that, in general conformally flat spacetimes, the rate of baryogenesis 
allows a perturbative expansion, and that its zeroth and first order contribution vanish. Afterward, in Sections~\ref{sec:utrivial} and~\ref{sec:ugeneral} we derive formulas for the second order contribution to the rate of baryogenesis depending on the value of the vector field~$u : M \rightarrow TM$ and the mass~$m$. In particular, we obtain a series of formulas which can all be described by expression~\eqref{eq:Bt2_final}, where the function~$G_{\Omega,m,u}$ captures the different considered scenarios. The simplest example is when~$u=\partial_{t}$ and~$m\neq0$, for which we obtain in Section~\ref{sec:utrivial} that
\begin{equation*}
        G_{\Omega,m,u}(k,k')=2m^2(-\omega_{k}\omega_{k'}+m^{2}-k\cdot k')\big[\hat{\alpha}_{1}(k-k')\hat{\alpha}_{2}(k'-k)+\hat{\alpha}_{1}(k'-k)\hat{\alpha}_{2}(k-k')\big]\;,
\end{equation*}
where~$\alpha_{1}=\frac{d\Omega}{dt}$, ~$\alpha_{2}=(\Omega-1)$ and the hats denote its Fourier transform. Note that in this scenario, introducing the coordinates~$r=\frac{k+k'}{2}$ and~$q=k-k'$, and using the spherical symmetry of the integrand, expression \eqref{eq:Bt2_final} can be simplified further (see Corollary~\ref{cor:Bt_conf_flat_uparallel}). The case where~$u\neq\partial_{t}$ is more involved and is analyzed in Section~\ref{sec:ugeneral} (see Corollary~\ref{cor:Bt_conf_flat_ugeneral} and Remark~\ref{rem:Bt2_mixed_terms}).

The paper is organized as follows. After reviewing the necessary mathematical background
and introducing the setup (Section~\ref{sec:math_intro}), general results in conformally flat spacetimes
are derived (Section~\ref{secgenconf}). 
Then we specialize the setting to Minkowski spacetime (Section~\ref{sec:baryomink}),
also making contact to our previous paper~\cite{baryomink}.
We then work out two specific scenarios.
We begin with a trivial regularizing vector field (i.e. $u=\partial_{t}$), in which case the mass is crucial for
our mechanism to be effective  (Section~\ref{sec:utrivial}).
For a general regularizing vector field, baryogenesis is in general non-zero even in the massless
case (Section~\ref{sec:ugeneral}). The paper concludes with a discussion of our findings
and a brief outlook (Section~\ref{secdiscuss}). The appendix provides the detailed computation
of the symmetrized Hamiltonian.

\section{Preliminaries}\label{sec:math_intro}
\subsection{Geometric and spin geometric preliminaries}
In this paper the conventions and notations of~\cite{baryomink} are followed. In particular, all spacetimes~$(M,g)$ are assumed to be four-dimensional, smooth, oriented, time oriented and globally hyperbolic.
We denote by~$t$ a global time function and the associated global smooth foliation is~$(N_t)_{t\in \mathbb{R}}$. We use the convention~$(+,-,-,-)$ for the signature of the Lorentzian metric~$g$. Furthermore, small Latin indices~$j,k, \ldots$ denote spacetime coordinate indices, whereas small Greek indices label the spatial coordinates. Moreover, whenever a foliation~$(N_t)_{t\in\mathbb{R}}$ is fixed in the spacetime~$(M,g)$ and a mathematical object is given in local coordinates, we always choose coordinates~$(x^{j})_{j=0,\ldots, 3}$ such that~$x^0 = t$ coincides with the time function. Finally, as is customary, the Einstein summation convention is used throughout the paper.

We denote {\em{Clifford multiplication}} in~$(M,g)$ by~$\gamma_{g} : TM \otimes SM \rightarrow SM$ (where~$SM$ is the spinor bundle) and, given an orthonormal basis~$(e_{j})_{j=0,\ldots,3}$,~$\gamma_{gj}\coloneqq\gammag(e_{j})$. The associated fiber~$S_pM\cong\mathbb{C}^{4}$ at a spacetime point~$p \in M$ is a
four-dimensional complex vector space referred to as the {\em{spinor space}}. It is
endowed with an indefinite inner product~$\Sl .|. \Sr_{S_pM}$ of signature~$(2,2)$, referred to as
the {\em{spin inner product}}. Note that, Clifford multiplication is symmetric with respect to the inner product.
Sections in the spinor bundle are called {\em{spinor fields}}.

Furthermore, in local coordinates the Levi-Civita (or metric) spin connection~$\nabla^{s}$, the Dirac operator~$D_{g}:C^{\infty}(M,SM)\rightarrow C^{\infty}(M,SM)$ and the Dirac Hamiltonian~$H_{g}: C^{\infty}(N_t,SM)\rightarrow C^{\infty}(N_t,SM)$ are 
\begin{align}
    &\textit{Levi-Civita spin connection} & &\nabla^{s}_{j}=\partial_{j}-iE_{j}-ia_{j}\label{eq:spin_connection}\\
    & \textit{Dirac operator} &  &D_{g}=i\gamma_{g}^j \nabla^{s}_j \\
    & \textit{Dirac Hamiltonian} & &H_{g}=-(\gammag^{0})^{-1}\left(i\gammag^{\mu}\nabla^{s}_{\mu}-m\right) - E_{0}-a_{0}\label{eq:Hamiltonian}\;,
\end{align}
where~$E_{j},a_{j}$ are linear operators on the spinor space.
Given a mass parameter~$m >0$, the {\em{Dirac equation}}  (for smooth spinor fields) reads
\[  (D_g-m)\psi=0 \:. \]

For smooth and compactly supported spinor fields~$\psi,\phi\in C^{\infty}_{0}(N_t,SM)$ we define the scalar product 
\[ %\label{eq:innerprod_hypersurface}
    (\psi|\phi)_t\coloneqq \int_{N_t} \Sl\psi|\gammag(\nu)\phi\Sr_{S_{p}M} d\mu_{N_t} \:, \]
where~$\nu$ denotes the future-directed normal.
Let~$C^{\infty}_{\textrm{sc}}(M,SM)$ denote the space of smooth spinor fields with spatially compact support: i.e.\ $\psi \in C^{\infty}_{\textrm{sc}}(M,SM)$ provided for any foliation~$(N_t)_{t\in \mathbb{R}}$ and leaf~$N_{t'}$, it holds that~$\psi|_{N_{t'}}\in C^{\infty}_{0}(N_{t'},SM)$. Then, if~$\psi,\phi\in C^{\infty}_{\textrm{sc}}(M,SM)$ satisfy the Dirac equation, it can be proven that the scalar product~$(\psi|\phi)_t$ is independent of the considered Cauchy hypersurface (see~\cite[equation (2.6)]{finite} or~\cite[Corollaries~2.1.3 and~2.1.4]{treude}); this
is referred to as current conservation. The same holds, by construction, if the spinor fields follow the locally rigid dynamics (cf.~Definition~\ref{def:locally_rigid_dynamics_spinors}).

Moreover, we introduce the Hilbert space of square integrable spinor fields
\[ \H_{t,g}\coloneqq L^{2}(N_t,SM)\;, \]
with scalar product~$(\cdot|\cdot)_t$. If the underlying spacetime is Minkowski spacetime~$(\mathbb{R}^{1,3},\eta)$, we denote the scalar product~$(\cdot|\cdot)_t$ simply by~$(\cdot|\cdot)$ and the space of square integrable spinor fields by~$\H_{t,\eta}$.

In this paper we will focus on spacetimes~$(M,g)$ with~$M=\mathbb{R}\times \mathbb{R}^{3}$ which are \emph{conformally flat}. In other words, around every point~$p\in M$ there exists a neighborhood~$U$ such that the metric is
\begin{equation}\label{eq:metric_conf_flat}
    g_{p}=\Omega^{2}(p)\:\eta_{p}=\Omega^{2}(t,r,\theta,\varphi)\: \big(dt^{2}-dr^{2}-r^{2}d\theta^{2}-r^{2}\sin^{2}{\theta}d\varphi^{2} \big)\;,
\end{equation}
where the conformal factor~$\Omega : U\rightarrow (0,\infty)$ is a smooth function of all four coordinates,
and~$\eta$ denotes the metric of the Minkowski spacetime (in spherical coordinates). 
By the Weyl-Schouten Theorem, a spacetime (of dimension~$d\geq4$) is locally conformally flat if and only if its Weyl tensor vanishes. Many important cosmological spacetimes are locally conformally flat. For example, this is the case for FLRW spacetimes (see~\cite{confflatFLRW} or~\cite{gron}) with a conformal factor~$\Omega=\Omega(t,r)$. Note that in the rest of this paper we will always implicitly assume that~$M=\mathbb{R}\times\mathbb{R}^{3}$, which implies existence of a global chart such that the metric~$g$ is of the above form (so we will not distinguish anymore between local and global conformal flatness).

\par The prime (and, arguably, simplest) example of a FLRW spacetime is the flat FLRW spacetime, which is conformally flat and satisfies that~$\Omega=\Omega(t)$.
We will prove that the rate of baryogenesis in flat FLRW bears many similarities with the one in Minkowski spacetime and is even the same if~$m=0$ (Corollary~\ref{cor:At_FLRW}).

\subsection{Mathematical setup for the study of baryogenesis}
We start by introducing the space of spinor fields which we will be focussing on.
\begin{Def}\label{def:regularized_spinor_space}
    Let~$(M,g)$ be a globally hyperbolic spacetime with spinor bundle~$SM$ and consider a foliation~$(N_t)_{t\in\mathbb{R}}$. Then, given a hypersurface~$N_{t_{0}}\in(N_t)_{t\in\mathbb{R}}$, a subspace~$\H_{t_{0}}^{\varepsilon} \subset C^{\infty}_{0}(N_{t_{0}},SM)$ and an isometric
operator~$V_{t_{0}}^{t} : \H_{t_{0}}^{\varepsilon} \rightarrow C^{\infty}_{0}(N_t,SM)$, we define the {\bf{space of regularized spinor fields at a time~$\mathbf{t}$}} as
\[ \H^{\varepsilon}_t\coloneqq V_{t_{0}}^{t}(\H_{t_{0}}^{\varepsilon}) \]
\end{Def}
Intuitively the operator~$V_{t_{0}}^{t}$ describes the dynamics of the spinor fields in the spacetime~$(M,g)$. For example, if we assume that the spinor fields follow the Dirac dynamics, then~$V_{t_{0}}^{t}$ acts by restricting solutions to the Dirac equation to different Cauchy hypersurfaces. Alternatively, in Definition~\ref{def:locally_rigid_dynamics_spinors} we will introduce an explicit expression for~$V_{t_{0}}^{t}$ which describes a spinor dynamics deviating slightly from the Dirac dynamics.

In Minkowski spacetime~$(\mathbb{R}^{1,3},\eta)$, the Dirac Hamiltonian~$H_\eta$ is a selfadjoint operator on the Sobolev space~$H^{1}(\R^{3})$ with an absolutely continuous spectrum~$\sigma(H_\eta)=(-\infty,-m]\cup[m,\infty)$ (see \cite[Theorem 1.1]{Thaller}). More importantly, its eigenstates are associated to positive (and negative) eigenvalues of~$H_\eta$ are interpreted physically as \emph{particles} (respectively \emph{antiparticles}). Hence, the Dirac Hamiltonian~$H_\eta$ seems a suitable starting point in order to describe a process of particle creation. However, if we consider the Dirac Hamiltonian~$H_{g}$ in order to generalize the well understood operator~$H_\eta$ to a general globally hyperbolic spacetime~$(M,g)$, this faces severe problems:
\begin{itemize}[leftmargin=2em]
    \item In general (unless~$(M,g)$ is stationary, cf.~\cite[Remark 3.3]{baryomink}), the operator~$H_{g}$ is
    not symmetric on the Hilbert space~$\H_{t,g}$.
    \item Secondly, even if~$(M,g)$ is stationary, the Dirac Hamiltonian~$H_{g}$ only describes the Dirac dynamics of the spinor fields. I.e. it does not allow to study more general dynamics.
\end{itemize}
Hence, in a general globally hyperbolic spacetime~$(M,g)$, the Dirac Hamiltonian~$H_{g}$ is, for our purposes, not the right object to consider. In the following definition we introduce a symmetric operator which 
which will enter our description of a more general spinor dynamics.

\begin{Def} {\rm{(Symmetrized Hamiltonian)}}\label{def:symmetrized_Hamilt}
Let~$(M,g)$ be a globally hyperbolic spacetime and~$(N_t)_{t\in\mathbb{R}}$ a distinguished foliation. Moreover, 
consider a smooth global future directed timelike vector field~$u: M\rightarrow TM$, which will be referred to as the
{\bf{regularizing vector field}}. Then, we define the {\bf{symmetrized Hamiltonian}} at a time~$t$ as the operator 
\[ A_t: C^\infty_0(N_t,SM) \subset \H_{t,g}\rightarrow \H_{t,g}\;, \]
which in local coordinates is given by the following expression
    \begin{equation}\label{eq:At}
    A_t\coloneqq \frac{1}{4}\{u^{0},H_{g}+H_{g}^{\ast}\}+\frac{i}{4}\{u^{\mu},\nabla^{s}_{\mu}-(\nabla^{s}_{\mu})^{\ast} \}\;.
\end{equation} 
\end{Def}

\begin{Remark}\label{rem:chiAt}\hfill\em{
\begin{enumerate}[leftmargin=2em]
\item[{\rm{(i)}}] The symmetrized Hamiltonian~$A_t$ is essentially self-adjoint for any~$t\in\mathbb{R}$, see~\cite[Lemma 5.3]{baryomink}. Its unique self-adjoint extension with dense domain~$\mathcal{D}$
will be denoted with the same symbol, i.e.\
\begin{equation*}
A_t : \mathcal{D}\subset \H_{t,g}\rightarrow\H_{t,g} \:.
\end{equation*} 
\item[{\rm{(ii)}}] Choosing a bounded interval~$I\subset \mathbb{R}$ and denoting its
characteristic function by~$\chi_I$, by the spectral theorem for bounded Borel functions, the operator~$\chi_I(A_t) : \mathcal{D}\rightarrow \H_{t,g}$
is well-defined and bounded. Moreover, since,~$\mathcal{D}$ is a dense subset of~$\H_{t,g}$, there exists a unique bounded linear extension
\begin{equation*}
\chi_I(A_t): \H_{t,g}\rightarrow\H_{t,g}\;.    
\end{equation*}
In Proposition~\ref{prop:chiAt} we will show that the spectral projection operator~$\chi_I(A_t)$ maps even into the space of smooth spinor fields taking values on the Cauchy hypersurface~$N_{t}$.
\hfill\QEDrem    
\end{enumerate}
}
\end{Remark}

We now introduce the equations describing the dynamics of the regularization vector field.

\begin{Def}\label{def:DxL}
    Let~$(N_t)_{t\in\mathbb{R}}$ be a foliation of the  globally hyperbolic spacetime~$(M,g)$ and choose a Cauchy hypersurface~$N_{t_{0}}$. Furthermore, let~$u : N_{t_{0}}\rightarrow TM~$ be a smooth future directed timelike vector field. Then,~$\scrL$ is the set of maximally extended future directed null geodesics~$\gamma: I\subset \mathbb{R}\rightarrow M$ (together with the interval of parametrization~$I$) in~$(M,g)$ such that whenever~$\gamma(s)\in N_{t_{0}}$ it holds that
\[ %\label{eq:DxL}
        g_{\gamma(s)}(u_{\gamma(s)},\dot{\gamma}(s))=1 \]
    Furthermore, for an arbitrary point~$p\in M$ we define the hypersurface~$D_{p}\scrL$ of the null bundle as
\[ D_{p}\scrL\coloneqq \{\dot{\gamma}(s)\; | \;(I,\gamma)\in \scrL \;\textrm{and}\; \gamma(s)=p\} \:. \]
\end{Def}

\begin{Def}[Locally rigid dynamics of~$u$]\label{def:locally_rigid_dynamics}
     Let~$(N_t)_{t\in\mathbb{R}}$ be a foliation of the  globally hyperbolic spacetime~$(M,g)$ and choose a Cauchy hypersurface~$N_{t_{0}}$. Furthermore, let~$u : N_{t_{0}}\rightarrow TM~$ be a smooth future directed timelike vector field, and let~$\scrL$ and~$D_{q}\scrL$ (for an arbitrary point~$q\in M$) be as in Definition~\ref{def:DxL}. Consider a sufficiently small~$\Delta t$ such that for any~$q \in N_{t_{0}+\Delta t}$ there exists a normal neighborhood~$U\subset M$ of~$q$ with~$U \cap N_{t_{0}}\neq \emptyset$. Then we define the following timelike vector field at~$q$:
\[ \xi_{q}\coloneqq\frac{1}{\mu_{q}(D_{q}\scrL)}\int_{D_{q}\scrL}\dot{\gamma}(s)\:d\mu_{q}(\dot{\gamma}(s)) \:, \]
     where~$d\mu_{q}(\dot{\gamma}(s))$ is the induced volume measure on~$D_{q}\scrL$,$(I,\gamma)\in\scrL$ and~$\gamma(s)=q$. Using the vector field~$\xi_{q}$, we define the regularizing vector field at~$q$ by
\[ u_{q}\coloneqq\frac{1}{|\xi_{q}|_{g}^{2}}\xi_{q} \:. \]
    Proceeding in an analogous way at each~$q\in N_{t_{0}+\Delta t}$ the timelike vector field~$u$ is extended to~$N_{t_{0}+\Delta t}$. We refer to this process as the {\bf{locally rigid dynamics}} of~$u$.
\end{Def}
The locally rigid dynamics also describes how the regularization evolves in time. We now use it
to define a spinor dynamics deviating slightly from the Dirac dynamics:

\begin{Def}[Locally rigid operator]\label{def:locally_rigid_dynamics_spinors}
Let~$(M,g)$ be a globally hyperbolic spacetime and~$(N_t)_{t\in\mathbb{R}}$ a foliation. Furthermore, let~$u : M \rightarrow TM$ be the regularizing vector field satisfying the locally rigid dynamics,~$(A_t)_{t\in\mathbb{R}}$ the associated family of symmetrized Hamiltonians
and~$\H_{t_{0}}^{\varepsilon}\coloneqq \chi_I(A_{t_{0}})(\H_{t_{0},g})$. Then, the {\bf{locally rigid operator}}~$V_{t_{0}}^{t} : \H^{\varepsilon}_{t_{0}}\rightarrow \H_{t,g}$ is defined by
\[ %\label{eq:loc_rig_operator}
V_{t_{0}}^{t}\coloneqq \lim_{k_{\max}\to\infty} \chi_I(A_t) U^{t}_{t-\Delta t}\cdot\cdot\cdot \chi_I(A_{t_{0}+\Delta t})U^{t_{0}+\Delta t}_{t_{0}} \quad \text{with} \qquad \Delta t\coloneqq \frac{t-t_{0}}{k_{\max}} \:. \] 
where~$I=(-\frac{1}{\varepsilon},-m)$ and for any~$t_{k}<t_{k+1}$ the operator~$U_{t_{k}}^{t_{k+1}} : \H_{t_{k},g} \rightarrow \H_{t_{k+1},g}$ is the unitary operator that describes the Dirac evolution of the regularized spinor fields. The dynamics described by the time evolution operators~$V^t_{t_0}$ is referred to as the
{\bf{locally rigid spinor dynamics}}.
\end{Def}
\par Since the interval~$I$ is bounded, if in addition~$(M,g)$ agrees with Minkowski spacetime outside a compact set~$V\subset M$, we have that~$V_{t_{0}}^{t}(\H^{\varepsilon}_{t_{0}})\subset C^{\infty}(N_t,SM)$ by Proposition~\ref{prop:chiAt}.

\par By construction, the locally rigid operator~$V_{t_{0}}^{t}$ describes the locally rigid evolution of the regularized spinor fields. The adiabatic projections have the advantage of implementing deviations from the Dirac dynamics. Moreover, they guarantee that the locally rigid operator~$V_{t_{0}}^{t}$ is unitary and thus that the scalar product~$(.|.)_t$ is preserved in time.

\par Well-definedness of~$V_{t_{0}}^{t}$, in particular that its image is a subset of the space of smooth spinors fields, might not be clear at a first glance. It is a consequence of the following result. Note that, in order to use elliptic regularity theory and Sobolev embeddings, we will assume that the Cauchy hypersurfaces have a bounded geometry in the following way: we will impose that outside a compact subset~$V\subset M$, our spacetime is given by Minkowski spacetime.

\begin{Prp}\label{prop:chiAt}
   Let~$(M,g)$ be a globally hyperbolic spacetime and~$V\subset M$ a compact subset such that~$g_{p}=\eta_{p}$ for any~$p\in M\setminus V$. Consider the densely defined self-adjoint operator~$A_t: \mathcal{D}\subset\H_{t,g}\rightarrow\H_{t,g}$~ corresponding to the symmetrized Hamiltonian. If~$I\subset\mathbb{R}$ is a bounded interval, then
   \begin{equation*}
       \chi_I(A_t)(\H_{t,g})\subset C^{\infty}(N_t,SM)\;.
   \end{equation*}
\end{Prp}
\begin{proof}
    In the first place, we show that for any~$p\in\mathbb{N}$, the differential operator~$A_t^{p}$ is an elliptic operator. In a general globally hyperbolic spacetime~$(M,g)$ and a foliation~$(N_t)_{t\in\mathbb{R}}$, the first order differential operator~$A_t$ presents the following form (cf.\ expressions~\eqref{eq:spin_connection}, \eqref{eq:Hamiltonian} and~\eqref{eq:At})
\[ A_t=i(u^{\mu}\textrm{Id}_{\mathbb{C}^{4}}-u^{0}\gamma_{g 0}\gamma_{g}^{\mu})\partial_{\mu}+\textrm{(lower order terms)}\;. \]
    So, the principal symbol of the operator~$A_t$ is given by
    \begin{equation*}
        \sigma_{1}(A_t,\xi)=i(u^{\mu}\textrm{Id}_{\mathbb{C}^{4}}-u^{0}\gamma_{g 0}\gamma_{g}^{\mu})\xi_{\mu}=i(u^{\mu}\textrm{Id}_{\mathbb{C}^{4}}-u^{0}g^{\mu\nu}[\gamma_{g 0},\gamma_{g\nu}])\xi_{\mu} \neq0 \hspace{0.4cm} \textrm{for}\hspace{0.4cm} \xi_{\mu}\neq0,
    \end{equation*}
    where we used that the term in brackets can never vanish since~$u^{0}g^{\mu\nu}[\gamma_{g 0},\gamma_{g\nu}]$ is antisymmetric whereas~$u^{\mu}\textrm{Id}_{\mathbb{C}^{4}}$ is obviously symmetric. So~$A_t$ is a first order elliptic operator. Analogously, we see that for any~$p\in \mathbb{N}$
    \begin{equation*}
        \sigma_{p}(A_t^{p},\xi)=\big(\sigma_{1}(A_t,\xi)\big)^{p}=i^{p}(u^{\mu}\textrm{Id}_{\mathbb{C}^{4}}-u^{0}\gamma_{g 0}\gamma_{g}^{\mu})^{p}\xi_{\mu}^{p}\neq0 \hspace{0.5cm} \textrm{for}\hspace{0.5cm} \xi_{\mu}\neq0\;,
    \end{equation*}
    \par Secondly, we show that~$A_t^{p}\chi_I(A_t)$ is a bounded operator if~$I\subset \mathbb{R}$ is a bounded interval. Consider a point~$p\in N_t\cap V$ and a neighborhood~$U\subset N_t\cap V$. Then, using the properties of the functional calculus we have that 
    \begin{equation*}
        \|A_t^{p}\chi_I(A_t)\|_{L^{2}(U)}=\Big\|\int_I\lambda^{p}dE\Big\|_{L^{2}(U)}\leq |I|^{p}<\infty\;.
    \end{equation*}
    where~$E=\chi(A_t) : \mathcal{B}(\R)\rightarrow \mathcal{L}(\H_{t,g})$. Boundedness of~$A_t^{p}\chi_I(A_t)$ implies that for any~$p\in\mathbb{N}$ and any~$\psi\in L^{2}(U)$, it holds that
    \begin{equation*}
    A_t^{p}\chi_I(A_t)\psi\in L^{2}(U)    \;.
    \end{equation*}
    So, by the (interior) elliptic regularity theory, it follows that~$\chi_I(A_t)\psi\in H^{p}(U)$ for all~$p\in \mathbb{N}$. The Sobolev embedding theorem then implies that~$\chi_I(A_t)\psi\in C^{\infty}(U)$. Since we assume a bounded geometry, we actually have that~$\chi_I(A_t)\psi\in C^{\infty}(N_t,SM)$ by the following argument: as~$N_t\cap V$ is compact we can cover it with finitely many neighborhoods~$U$ to obtain that~$\chi_I(A_t)\psi\in C^{\infty}(N_t\cap V)$. Moreover, as by assumption~$g|_{N_t\setminus V}=\eta$, outside of~$V$ we can directly apply the previous local reasoning to conclude that~$\chi_I(A_t)\psi\in C^{\infty}(N_t,SM)$.
\end{proof} \noindent
We remark that the previous proposition does not hold if the interval~$I$ is unbounded. For example, for~$I=\mathbb{R}$, the operator~$\chi_I(A_t)= \textrm{id}_{\H_{t,g}}$ is the identity, which
clearly does not map to smooth spinors.

We now restrict our attention to the class of spacetimes which we will consider in the remainder of this paper, namely conformally flat spacetimes.

\begin{Lemma}\label{lem:time_independent}
    Let~$(M,g)$ be a conformally flat spacetime and consider an integral operator~$\tilde{Q} : \H^{\varepsilon}_t\subset \H_{t,g}\rightarrow\H_{t,g}$. Then:
    \begin{enumerate}[leftmargin=2em]
    \item[{\rm{(i)}}] There exists a unitary operator 
\[ \tilde{U}: \H_{t,g} \rightarrow \H_{t,\eta} \]
which satisfies~$\tilde{U}\psi=\Omega^{3/2}\psi$ for all~$\psi\in\H_{t,g}$.
      \item[{\rm{(ii)}}] Let~$\H_\eta^{\varepsilon}:=\tilde{U}(\H^{\varepsilon}_t)$ denote the image of the operator~$\tilde{U}$. Then, the kernel of the integral operator~$Q : \H_\eta^{\varepsilon} \rightarrow \H_\eta^{\varepsilon}$ with~$Q:=\tilde{U}\tilde{Q}\tilde{U}^{-1}$ satisfies
\[ Q(x,y)=\Omega^{3/2}(x)\tilde{Q}(x,y)\Omega^{3/2}(y) \:, \]
    where~$x,y\in N_{t}$.
     \item[{\rm{(iii)}}] Provided~$\tilde{Q}$ is trace-class, it holds that
\[ %\label{eq:t_derivative_trace}
    \frac{d}{dt}\tr_{\overline{\H^{\varepsilon}_t}}(\tilde{Q})=\tr_{\overline{\H^{\varepsilon}_t}}\Big(\frac{d}{dt}\tilde{Q}\Big)=\tr_{\overline{{\H}^{\varepsilon}_\eta}}\Big(\frac{d}{dt}Q\Big)\;, \]
where here and in the rest of the paper~$\overline{\H^{\varepsilon}_t}$ and~$\overline{\H^{\varepsilon}_\eta}$ denote the completion of~$\H^{\varepsilon}_t$ and~$\H^{\varepsilon}_\eta$ with respect to the norms~$\|\cdot\|_t$ (induced by~$(\cdot,\cdot)_t$) and~$\|\cdot\|$(induced by~$(\cdot,\cdot)$) respectively.
\end{enumerate}
\end{Lemma}
In other words, the operator~$Q$ is defined by the commutativity of the following diagram,
\[ \begin{tikzcd}
\H^{\varepsilon}_t \arrow{r}{\tilde{Q}} \arrow[swap]{d}{\tilde{U}} & \H^{\varepsilon}_t \arrow{d}{\tilde{U}} \\%
\H_\eta^{\varepsilon} \arrow{r}{Q}& \H_\eta^{\varepsilon}
\end{tikzcd}
\]
\begin{proof}[Proof of Lemma~\ref{lem:time_independent}]
    Existence of a mapping~$\tilde{U}$ into~$(L^{2}(\mathbb{R}^{3},\mathbb{C}^{4}),(\cdot|\cdot))$ (where~$(\cdot|\cdot)$ denotes the standard scalar product of spinor fields in Minkowski spacetime) follows from a simple computation in which we rewrite the scalar product~$(\cdot|\cdot)_t$ in terms of the scalar product~$(\cdot|\cdot)$. Consider~$\psi,\varphi\in\H^{\varepsilon}_t$ and recall that, by conformal flatness of~$(M,g)$, for a Cauchy hypersurface~$N_t$ with unit normal~$\nu$, it holds that~$d\mu_{N_t}=\Omega^{3}(x)d^{3}x$ and~$\gammag(\nu)=\gamma_\eta(\partial_t)$:
    \begin{align*}
        &(\psi|\phi)_t=\int_{N_t} \Sl\psi|\gammag(\nu)\phi\Sr_{S_{p}M} d\mu_{N_t}=\int_{\mathbb{R}^{3}} \Sl\psi|\gamma_\eta(\partial_t)\phi\Sr_{\mathbb{C}^{4}} \Omega^{3}(x)d^{3}x\\
        &=\int_{\mathbb{R}^{3}} \Sl (\Omega^{3/2}\psi)(x)|\gamma_\eta(\partial_t)(\Omega^{3/2}\phi)(x)\Sr_{\mathbb{C}^{4}} d^{3}x=(\Omega^{3/2}\psi|\Omega^{3/2}\varphi)=:(\tilde{U}\psi|\tilde{U}\varphi)
    \end{align*}
    Consider the space~$\H_\eta^{\varepsilon}\coloneqq\tilde{U}(\H^{\varepsilon}_t)$. The second claim then follows simply by demanding that the integral operator~$Q$ satisfies that~$Q(\tilde{U}\psi)=\tilde{U}(\tilde{Q}\psi)$ for all~$\psi\in\H^{\varepsilon}_t$ and recalling that~$\tilde{Q}$ is also an integral operator on~$\H^{\varepsilon}_t$ with kernel~$\tilde{Q}(x,y)$
    \begin{align*}
        &(Q(\tilde{U}\psi))(x)=(\tilde{U}(\tilde{Q}\psi))(x)=\Omega^{3/2}(x)(\tilde{Q}\psi)(x)=\Omega^{3/2}(x)\int_{\mathbb{R}^{3}}\tilde{Q}(x,y)\psi(y)\Omega^{3}(y)d^{3}y\\
        &=\int_{\mathbb{R}^{3}}Q(x,y)(\tilde{U}\psi)(y)d^{3}y \hspace{0.5cm}\textrm{with}\hspace{0.5cm}Q(x,y):=\Omega^{3/2}(x)\tilde{Q}(x,y)\Omega^{3/2}(y)\;,
    \end{align*}
    where $x\in N_{t}$. For the last claim, we consider an orthonormal basis~$(\varphi_{j})_{j\in\mathbb{N}}$  of~$\H^{\varepsilon}_t$ which yields the orthonormal basis~$(\Omega^{3/2}\varphi_{j})_{j\in\mathbb{N}}$ of~$\H_\eta^{\varepsilon}$. Note that the elements of this basis are independent of the time function~$t$ because they are elements of~$\H_{t,\eta}$ (for any~$t\in\mathbb{R}$), so~$\frac{d}{dt}(\Omega^{3/2}\varphi_{j})=0$ or, equivalently,
    \begin{equation*}
        \Omega^{3/2}\frac{d\varphi_{j}}{dt}+\frac{3}{2}\frac{d\Omega}{dt}\Omega^{1/2}\varphi_{j}=0\iff\frac{d\varphi_{j}}{dt}=-\frac{3}{2}\frac{d\Omega}{dt}\frac{1}{\Omega}\varphi_{j}\;.
    \end{equation*}
    For this reason, differentiating the trace (in~$\H^{\varepsilon}_t$) of~$\tilde{Q}$ yields the following simple expression
    \begin{align*}
        &\frac{d}{dt}\tr_{\overline{\H^{\varepsilon}_t}}(\tilde{Q})=\frac{d}{dt}\Big(\sum_{j=0}^{\infty}(\varphi_{j}|\tilde{Q}\varphi_{j})_t\Big)=\frac{d}{dt}\Big(\sum_{j=0}^{\infty}\int_{\mathbb{R}^{3}}\Sl\varphi_{j}|\gammaeta(\partial_t)\tilde{Q}\varphi_{j}\Sr_{\mathbb{C}^{4}}\Omega^{3}d^{3}x\Big)\\
        &=\sum_{j=0}^{\infty}\Big[\int_{\mathbb{R}^{3}}\Sl\frac{d\varphi_{j}}{dt}|\gammaeta(\partial_t)\tilde{Q}\varphi_{j}\Sr_{\mathbb{C}^{4}}\Omega^{3}d^{3}x+\int_{\mathbb{R}^{3}}\Sl\varphi_{j}|\gammaeta(\partial_t)\frac{d}{dt}(\tilde{Q})\varphi_{j}\Sr_{\mathbb{C}^{4}}\Omega^{3}d^{3}x\\
        &+\int_{\mathbb{R}^{3}}\Sl\varphi_{j}|\gammaeta(\partial_t)\tilde{Q}\frac{d\varphi_{j}}{dt}\Sr_{\mathbb{C}^{4}}\Omega^{3}d^{3}x+3\int_{\mathbb{R}^{3}}\Sl\varphi_{j}|\gammaeta(\partial_t)\tilde{Q}\varphi_{j}\Sr_{\mathbb{C}^{4}}\frac{d\Omega}{dt}\Omega^{2}d^{3}x\Big]\\
        &=\int_{\mathbb{R}^{3}}\Sl\varphi_{j}|\gammaeta(\partial_t)\frac{d}{dt}(\tilde{Q})\varphi_{j}\Sr_{\mathbb{C}^{4}}\Omega^{3}d^{3}x=\tr_{\overline{\H^{\varepsilon}_t}}\Big(\frac{d}{dt}\tilde{Q}\Big)\;,
    \end{align*}
    Finally, the last claim follows from invariance of the trace under unitary transformations,
    \begin{align*}
        \frac{d}{dt}\tr_{\overline{\H^{\varepsilon}_t}}(\tilde{Q})=\frac{d}{dt}\tr_{\overline{\H^{\varepsilon}_t}}(\Tilde{U}^{-1}Q\Tilde{U})=\frac{d}{dt}\tr_{\overline{\H^{\varepsilon}_{\eta}}}(Q)=\tr_{\overline{\H^{\varepsilon}_{\eta}}}\Big(\frac{d}{dt}Q\Big)
    \end{align*}
\end{proof}

\begin{Def} \label{def:Baryo} {\rm{(Rate of baryogenesis)}}.
    Let~$(M,g)$ be a conformally flat spacetime,~$(N_t)_{t\in\mathbb{R}}$ a distinguished foliation with~$(A_t)_{t\in\mathbb{R}}$ the associated family of symmetrized Hamiltonians and~$\varepsilon,\Lambda>0$. Consider an initial subspace~$\H^{\varepsilon}_{t_{0}}\subset C^{\infty}(N_{t_{0}},SM)$ and an isometric operator~$V_{t_{0}}^{t}: \H_{t_{0}}^{\varepsilon}\rightarrow \H_{t,g}$, and define for each time~$t\in\mathbb{R}$ the space~$\H_t^{\varepsilon}:=V_{t_{0}}^{t}(\H_{t_{0}}^{\varepsilon})$. Then, the  {\bf{rate of baryogenesis}} is
    \begin{align}\label{eq:baryo_rate}
        B_t\coloneqq \frac{d}{dt}\text{\rm{tr}}_{\overline{\H^{\varepsilon}_t}} \big( \eta_{\Lambda}(\tilde{H}_\eta)\big(\chi_I(A_t)-\chi_{I}(\tilde{H}_{\eta})\big)\big)\;,
    \end{align}
    where~$\eta_{\Lambda}\in C^{\infty}_{0}((-\Lambda,\Lambda),[0,1])$ is a smooth cut-off operator,~$I\coloneqq (-1/\varepsilon, -m)$ and~$\tilde{H}_\eta:=\tilde{U}^{-1}H_\eta\tilde{U}$ (with $\Tilde{U} :\H_{t,g}\rightarrow \H_{t,\eta}$ the unitary operator introduced in Lemma~\ref{lem:time_independent}). 
\end{Def} \noindent
In the previous definition it is tacitly understood
that the operator product
\begin{align*}
\eta_{\Lambda}(\tilde{H}_\eta)\big(\chi_I(A_t)-\chi_{I}(\tilde{H}_{\eta})\big)    
\end{align*}
is trace-class for all~$t\in\mathbb{R}$. In  Proposition~\ref{prop:Bt_conf_flat} and Lemma~\ref{lem:trace_class} we will prove that it is trace-class if the assumptions of Theorem~\ref{theo:baryo_rate} are satisfied.
\begin{Remark}\label{rem:baryo_rate}
\em{
 In addition, for computational simplicity, we will assume that there exists a sufficiently small~$\delta>0$ such that~$\eta_{\Lambda}(\omega)=1$ for~$\omega\in(-\Lambda+\delta,\Lambda-\delta)$. Moreover, we will always assume that the orders of magnitude of the parameters~$m, \varepsilon$ and~$\Lambda$ satisfy that
\[ -\frac{1}{\varepsilon}\ll-\Lambda\ll-m\:. \]
    }
    \hfill\QEDrem
\end{Remark}
\noindent Note that, more explicitly, the operator~$\tilde{H}_\eta$ is given by
\begin{align*} %\label{eq:H_eta_symmetric}
\tilde{H}_\eta\psi&=(\tilde{U}^{-1}H_\eta\tilde{U})\psi=(\Omega^{-3/2}H_\eta\Omega^{3/2})\psi=\Omega^{-3/2}\big(\Omega^{3/2}H_\eta-\frac{3}{2}i\pr_{\mu}(\Omega)\Omega^{1/2}\gamma_{\eta 0}\gamma_\eta^{\mu}\big)\psi\nonumber\\
&=\Big(H_\eta-\frac{3}{2}i\frac{\pr_{\mu}(\Omega)}{\Omega}\gamma_{\eta 0}\gamma_\eta^{\mu}\Big)\psi\;,
\end{align*}
which is a self-adjoint operator on the Sobolev space $H^{1}(\R^{3})$ since it is unitarily equivalent to $H_{\eta}$ (so also their spectra agree, i.e.~$\sigma(\tilde{H}_\eta)=(-\infty,m]\cup[m,\infty)$). 

\section{General results for conformally flat spacetimes} \label{secgenconf}
\par It is a well known feature of Lorentzian geometry that conformal transformations preserve the causal structure of spacetime: a vector field~$X$ is timelike (or null) in~$(M,g)$ if and only if it is timelike (or null) in~$(M,\Omega^{2}g)$, where~$\Omega : M \rightarrow (0,\infty)$ is smooth. More remarkably is that a specific timelike vector field, namely the regularizing vector field, is a conformal invariant.

\begin{Prp}\label{prop:u_conformal_invariant}
    The locally rigid evolution of the regularizing vector field~$u$ is a conformal invariant.
\end{Prp}
    By this conclusion we mean the following: consider a foliation~$(N_t)_{t\in \mathbb{R}}$ of the smooth manifold~$M$, a Lorentzian metric~$g$ on~$M$ with~$[g]$ the equivalence class of metrics which are conformally equivalent to~$g$, and a timelike vector field~$u : N_{t_{0}}\rightarrow TM$ (with~$t_{0}\in\mathbb{R}$). For a specific~$h\in[g]$, we assume that the initial vector field~$u : N_{t_{0}}\rightarrow TM$ evolves according to the locally rigid dynamics in~$(M,h)$. Then, the statement of the previous proposition is that, starting from the same~$u|_{N_{t_{0}}}$ and considering any other metric~$g\in[g]$, the locally rigid dynamics in~$(M,g)$ yields the same vector field~$u : M \rightarrow TM$.
    \begin{proof}[Proof of Lemma~\ref{prop:u_conformal_invariant}.]
        Let~$g$ and~$h$ be two conformally equivalent metrics on the manifold~$M$ with conformal factor~$\Omega : M \rightarrow(0,\infty)$, i.e.\
    \begin{equation*}
        g=\Omega^{2}h \;.
    \end{equation*}
     Furthermore, consider a foliation~$(N_t)_{t\in\mathbb{R}}$ and a fixed initial regularizing vector field~$u : N_{t_{0}}\rightarrow TM$ (with~$t_{0}\in\mathbb{R}$). In the following, we will call a curve~$\gamma : I\rightarrow M$ a~$h$-geodesic (or a~$g$-geodesic) provided it is a geodesic with respect to the metric~$h$ (respectively, the metric~$g$). 
     \par Applying Definition~\ref{def:DxL} to~$(M,h)$,~$(I,\gamma)\in \scrL^{h}$ provided~$\gamma: I\subset \mathbb{R}\rightarrow M, s\mapsto \gamma(s)$ is null~$h$-geodesic which satisfies that if~$\gamma(s)\in N_{t_{0}}$, then
    \begin{equation*}
        h_{\gamma(s)}(u_{\gamma(s)},\dot{\gamma}(s))=1 \;
    \end{equation*}
    It is well known that null geodesics are, up to a reparametrization, conformal invariants. In particular, given a null~$h$-geodesic~$\gamma : I\rightarrow M$ we obtain a null~$g$-geodesic~$\tilde{\gamma} : J\rightarrow M$ through the reparametrization~$\psi : J\rightarrow I , \tilde{s}\mapsto s(\tilde{s})$ which satisfies (cf.~\cite[Section 2.3]{sanchez_candela} or \cite[Section 4.2]{dgc})
    \begin{equation*}
            \frac{ds}{d\tilde{s}}(\tilde{s})=\frac{1}{c\Omega^{2}(\Tgamma(\tilde{s}))}\;,
    \end{equation*}
    where~$c\in(0,\infty)$ is an arbitrary integration constant which we set equal to one. Then, by the chain rule we directly see that the (reparametrized) null~$g$-geodesic satisfies that
    \begin{equation*}
        \tilde{\gamma}\coloneqq\gamma\circ\psi : J\rightarrow M \;,\quad \dot{\tilde{\gamma}}\coloneqq\frac{d\tilde{\gamma}}{d\tilde{s}}=\Omega^{-2}\dot{\gamma} \;,
    \end{equation*}
    where~$\dot{\gamma}\coloneqq\frac{d\gamma}{ds}$. Consider a point~$p\in N_{t_{0}}$ and parameters~$s_{0}\in I$ ,~$\tilde{s}_{0}\in J$ such that~$\gamma(s_{0})=\Tgamma(\tilde{s}_{0})=p$. It then directly follows that if~$(I,\gamma)\in \scrL^{h}$, then~$\dot{\Tgamma}$ satisfies that
    \begin{equation*}
        g_{p}(u_{p},\dot{\Tgamma}(\tilde{s}_{0}))=\Omega^{2}(p)h_{p}(u_{p},\Omega^{-2}(p)\dot{\gamma}(s_{0}))=h_{p}(u_{p},\dot{\gamma}({s}_{0}))=1\;,
    \end{equation*}
    in other words~$(J,\tilde{\gamma})\in \scrL^{g}$, where~$\scrL^{g}$ is the set of maximally extended null~$g$-geodesics that satisfy Definition~\ref{def:DxL}. Moreover, for an arbitrary point~$q\in M$, the sets~$D_{q}\scrL^{h}$ and~$D_{q}\scrL^{g}$ are related by
\[ D_{q}\scrL^{g}=\{\Omega^{-2}(q)\dot{\gamma}(s_{0})\; : \;\dot{\gamma}(s_{0})\in D_{q}\scrL^{h}\} \:, \]
    where ~$s_{0}\in J$ is such that~$\gamma(s_{0})=q$. Consider now that the regularizing vector field~$u$ evolves through spacetime following the locally rigid dynamics (recall Definition~\ref{def:locally_rigid_dynamics}). With respect to the metric~$h$, at an arbitrary point~$q$ the regularizing vector field~$u_{q}$ is defined by
    \begin{equation*}
        \xi_{q}\coloneqq\frac{1}{\mu_{q}(D_{q}\scrL^{h})}\int_{D_{q}\scrL^{h}}\dot{\gamma}(s)\:d\mu_{q}(\dot{\gamma}(s)) \:, \qquad u_{q}\coloneqq\frac{1}{|\xi_{q}|_{h}^{2}}\xi_{q} \:.
    \end{equation*}
    Analogously we can define the vector field~$\tilde{\xi}$ and the regularizing vector field~$\tilde{u}$ at~$q$
    \begin{align*}
        &\tilde{\xi}_{q}\coloneqq\frac{1}{\mu_{q}(D_{q}\scrL^{g})}\int_{D_{q}\scrL^{g}}\dot{\Tgamma}(\tilde{s})\:d\mu_{q}(\dot{\Tgamma}(\tilde{s}))= \Omega^{-2}(q)\xi_{q}\\
        &\tilde{u}_{q}\coloneqq\frac{1}{|\tilde{\xi}_{q}|_{g}^{2}}\tilde{\xi}_{q}=\frac{1}{\Omega^{-2}(q)|\xi_{q}|_{h}^{2}}\Omega^{-2}(q)\xi_{q}=u_{q}\:,
    \end{align*}
    where we used that~$\mu_{q}(D_{q}\scrL^{g})=\Omega^{-2}(q)\mu_{q}(D_{q}\scrL^{h})$ (see also~\cite[Section 4.2]{dgc}). In other words, the locally rigid evolution of the regularizing vector field~$u$ is independent of the considered metric~$g$ in the conformal equivalence class~$[g]$.    
    \end{proof}

As a direct consequence of the previous proposition it suffices to determine the dynamics of the regularizing vector field in Minkowski spacetime to know it for any conformally flat spacetime. Moreover, since the operator~$\tilde{H}_\eta$ involves spatial derivatives of the conformal factor, for conformally flat spacetimes with~$\Omega=\Omega(t)$ we have that the coordinate expression of~$H_\eta$ and $\tilde{H}_\eta$ agree. This, together with the previous proposition, brings about the following Corollary.

\begin{Corollary}\label{cor:At_FLRW}
    In the massless case ($m=0$), the rate of baryogenesis in Minkowski and conformally flat spacetimes with~$\Omega=\Omega(t)$ agree.
\end{Corollary}
\begin{proof}
    \par The operator~$A_t$ for a conformally flat spacetime with~$\Omega=\Omega(t)$ and $m=0$ (cf.\ Lemma~\ref{lem:coefs_conformally_flat}) is
\[ A_t=\frac{1}{2}\{u^{t},-i\gamma_{gt}\gammag^{\alpha}\pr_{\alpha}\}+\frac{i}{2}\{ u^{\alpha},\pr_{\alpha}\}+\frac{i}{2r}u^{r}+\frac{i}{2}\frac{\cos{\theta}}{\sin{\theta}} u^{\theta} \;. \]
    Note that as~$\gamma_{gt}\gammag^{\alpha}=\gamma_{\eta 0}\gamma_\eta^{\alpha}$ and~$u$ is the same in~$(\R\times\R^{3},g)$ or~$(\R\times\R^{3},\eta)$ by Proposition~\ref{prop:u_conformal_invariant}, the coordinate expression of $\tilde{H}_{\eta}$ and $H_{\eta}$ agree. It follows that for $\Omega=\Omega(t)$ and $m=0$, the unitarily transformed (back to Minkowski spacetime) operator $\tilde{U}A_{t}\tilde{U}^{-1} : \H^{\varepsilon}_{\eta}\rightarrow \H^{\varepsilon}_{\eta}$ (with $\tilde{U}$ as in Lemma~\ref{lem:time_independent}) agrees with the symmetrized Hamiltonian $A_{t}^{\eta}$ in Minkowski spacetime, so $A_{t}=\tilde{U}^{-1}A^{\eta}_{t}\tilde{U}$ and the rates of baryogenesis agree,
    \begin{align*}
        &B_t\coloneqq \frac{d}{dt}\text{\rm{tr}}_{\overline{\H^{\varepsilon}_t}} \big( \eta_{\Lambda}(\tilde{H}_\eta)\big(\chi_I(A_t)-\chi_{I}(\tilde{H}_{\eta})\big)\big)\\
        &=\frac{d}{dt}\tr_{\overline{\H_{t}^{\varepsilon}}}\big( \Omega^{-3/2}\eta_{\Lambda}(H_\eta)\big(\chi_I(A^{\eta}_t)-\chi_{I}(H_{\eta})\big)\Omega^{3/2}\big) \\ &=\frac{d}{dt}\tr_{\overline{\H_{\eta}^{\varepsilon}}}\big( \eta_{\Lambda}(H_\eta)\big(\chi_I(A^{\eta}_t)-\chi_{I}(H_{\eta})\big)\big) = B_t^{\eta} \:.
    \end{align*}
This concludes the proof.
\end{proof}
The previous statement does not hold for a general conformally flat spacetime nor if~$m\neq0$ because of the~$m\Omega\gamma_{\eta 0}$ factor appearing inside~$H_{g}+H_{g}^{\ast}$.

Consider a metric~$g=\Omega^{2}h\in[h]$, where~$[h]$ denotes again the equivalence class of metrics which are conformally equivalent to~$h$, and the corresponding Dirac operators~$D_{g}$ and~$D_{h}$. It is generally known that harmonic spinors (i.e.\ spinors in the kernel of the Dirac operator) enjoy particularly nice properties:~$\psi\in\textrm{ker}(D_{h})$ if and only if~$\Omega^{-3/2}\psi\in\textrm{ker}(D_{g})$\footnote{It follows from the property that if~$M$ is a four-dimensional spin manifold and~$g=\Omega^{2}h$, then~$D_{g}\psi= \Omega^{-\frac{5}{2}} D_{h}(\Omega^{\frac{3}{2}}\psi)$ for all spinors~$\psi$ (where the isomorphism identifying spinor bundles was omitted), see~\cite[Prop. 1.3]{Hitchin}, \cite[Theorem~5.24]{lawson+michelsohn} or~\cite[Proposition~1.3.10]{Ginoux}) for the proof.}. In the following lemma we will show that specific harmonic spinors have some additional useful properties.

\begin{Lemma}\label{lem:At_diagonalizable}
    Let~$(M,g)$ be the conformally flat spacetime~\eqref{eq:metric_conf_flat}, consider the operator~$\tilde{H}_\eta:=\tilde{U}^{-1}H_\eta\tilde{U}$ and~$m=0$. If~$\psi\in\textrm{ker}(H_\eta-\lambda)$ with~$\lambda\in\sigma(H_\eta)$, then~$\Omega^{-3/2}\psi\in \textrm{ker}(\tilde{H}_\eta-\lambda)$. 
    If in addition~$u=\partial_t$, then the rate of baryogenesis vanishes identically.
\end{Lemma}
\begin{proof}
    Let~$\psi\in \textrm{ker}(H_\eta-\lambda)$ with~$\lambda\in\sigma(H_\eta)$ and~$m=0$. A direct computation yields
\begin{align*}
    \tilde{H}_\eta(\Omega^{-3/2}\psi)&=\Big(H_\eta-\frac{3}{2}\frac{\partial_{\mu}(\Omega)}{\Omega}i\gamma_{\eta0}\gamma^{\mu}_\eta\Big)\Omega^{-3/2}\psi\\
    &=\Big(\frac{3}{2}\frac{\partial_{\mu}(\Omega)}{\Omega^{5/2}}i\gamma_{\eta0}\gamma^{\mu}_\eta+\lambda\Omega^{-3/2}-\frac{3}{2}\frac{\partial_{\mu}(\Omega)}{\Omega^{5/2}}i\gamma_{\eta0}\gamma^{\mu}_\eta\Big)\psi=\lambda(\Omega^{-3/2}\psi)\;.
\end{align*}
where we used that~$H_\eta\Omega^{-3/2}=\frac{3}{2}i\partial_{\mu}(\Omega)\Omega^{-5/2}\gamma_{\eta0}\gamma^{\mu}_\eta$. 
\par If~$m=0$ and~$u=\partial_t$, the symmetrized Hamiltonian~$A_t$ and~$\tilde{H}_\eta$ agree (see Lemma~\ref{lem:coefs_conformally_flat} for the derivation of~$A_t$). Then, 
\begin{align*}
    B_t&:= \frac{d}{dt}\text{\rm{tr}}_{\overline{\H^{\varepsilon}_t}} \big( \eta_{\Lambda}(\tilde{H}_\eta)\big(\chi_I(A_t)-\chi_{I}(\tilde{H}_{\eta})\big)\big)=\frac{d}{dt}\text{\rm{tr}}_{\overline{\H^{\varepsilon}_t}} \big( \eta_{\Lambda}(\tilde{H}_\eta)\big(\chi_I(\tilde{H}_{\eta})-\chi_{I}(\tilde{H}_{\eta})\big)\big)
    =0
\end{align*}
\end{proof}

\par Starting from the following lemma, in the remainder of the paper we will study the rate of baryogenesis perturbatively. Note that Corollary~\ref{cor:At_FLRW} and Lemma~\ref{lem:At_diagonalizable} show under which conditions the rate of baryogenesis vanishes in a conformally flat spacetime and when it agrees with the one in Minkowski spacetime. Hence, in the perturbative analysis of baryogenesis, we will perform small perturbations around the following background scenario: conformally flat spacetimes with~$\Omega=\Omega(t)$,~$m=0$ and~$u=\partial_t$. In the next lemma we give sufficient conditions for well-definedness of the power expansion of the rate of baryogenesis, show that the zeroth and first order contributions vanish and discuss a general formula for the second order contribution.

\par Furthermore, in the rest of the paper we will assume that~$A_t$ has an absolutely continuous spectrum. As a consequence, given an interval~$I\coloneqq (-1/\varepsilon,\omega)$ with~$\omega\leq-m$, for every~$\psi\in \mathcal{D}$ (with~$\mathcal{D}\subset\H_{t,g}$ the domain of self-adjointness of~$A_t$, cf.\ Remark~\ref{rem:baryo_rate}) there exists a Lebesgue integrable function~$f_{\psi,\psi} :\mathbb{R}\rightarrow\mathbb{R}$ such that the spectral measure~$\mu_{\psi,\psi}^{A_t}: \mathcal{B}(\mathbb{R})\rightarrow[0,\infty)$ satisfies that 
\begin{equation*}
    \mu_{\psi,\psi}^{A_t}(I):=(\psi|\chi_I(A_t)\psi)_t=\int_If_{\psi,\psi}(\omega')\:d\omega'\;.
\end{equation*}
We use this to introduce an operator~$F_{\omega'}(A_t):\H_{t,g}\rightarrow\H_{t,g}$ with $f_{\psi,\psi}(\omega')=(\psi|F_{\omega'}(A_t)\psi)_t$.
Then, the spectral projection operator~$\chi_I(A_t)$ and~$F_\omega(A_t)$ are related as follows
\[ %\label{eq:pvmDiracHamiltonian}
    \chi_I(A_t)=\int_{-1/\varepsilon}^{\omega}F_{\omega'}(A_t)\:d\omega' \;, \]
i.e.~$F_\omega(A_t):=\frac{d}{d\omega}\chi_I(A_t)$. The operators~$F_\omega(H_\eta)$ and~$F_\omega(\tilde{H}_\eta)$ are defined analogously since~$H_\eta$ has an absolutely continuous spectrum and~$\tilde{H}_\eta$ is unitarily equivalent to~$H_\eta$.
\begin{Prp}\label{prop:Bt_conf_flat}
    Let~$(N_t)_{t\in\mathbb{R}}$ be the foliation of the conformally flat spacetime~\eqref{eq:metric_conf_flat} given by the level sets of the global time function~$t$. Furthermore, assume that the family of symmetrized Hamiltonians~$(A_t)_{t\in\mathbb{R}}$ have an absolutely continuous spectrum, that the differential operator~$\Delta A(t)\coloneqq A_t-\tilde{H}_\eta$ has smooth and compactly supported coefficients in~$N_t$ and that for all~$\omega\in\rho(\Tilde{H}_{\eta})$ it holds that
    \begin{equation}\label{eq:small_Delta_A}
    \|R_\omega(\tilde{H}_\eta)\Delta A(t))\|<1\;.
    \end{equation}
    Then, expanding the rate of baryogenesis in powers of~$R_\omega(\tilde{H}_\eta)\Delta A(t)$ yields 
        \begin{align}
    &B^{(0)}_t=0\nonumber\\
    &B^{(1)}_t=\frac{d}{dt}\tr_{\H_{}^{\varepsilon}}\Big(\Delta AF_{-m}(\tilde{H}_\eta)\Big)=0 \label{eq:Bt1_conf_flat}\\
    &B^{(2)}_t=-\int_{-\infty}^\infty d\omega \int_{-\infty}^\infty d\omega'\:
    \partial_\omega\Big(\eta_{\Lambda}(\omega)\frac{d}{dt}\tr_{\overline{\H_t^{\varepsilon}}}(\tilde{Q}(\omega,\omega'))\Big)\:\frac{g(\omega') - g(\omega)}{\omega'-\omega} \label{eq:Bt2_conf_flat}\:,
    \end{align}
    where~$\tilde{Q}(\omega,\omega')\coloneqq \Delta AF_\omega(\tilde{H}_\eta)\Delta AF_{\omega'}(\tilde{H}_\eta) : \H_{t,g}\rightarrow C_{0}^{\infty}(N_t,SM)$ and~$g$ is the characteristic function of the set~$(-1/\varepsilon,-m)$. 
\end{Prp}
\noindent If~$(M,g)$ agrees with Minkowski spacetime outside a compact set and~$I:=(-\frac{1}{\varepsilon},\omega)$ is a bounded interval, it holds that~$\chi_I(A_t) : \H_{t,g}\rightarrow C^{\infty}(N_t,SM)$ (by Proposition~\ref{prop:chiAt}). As a consequence, also the operators~$F_\omega(A_t)$,~$\chi_I(\tilde{H}_\eta)$ and~$F_\omega(\tilde{H}_\eta)$ map into the space of smooth spinor fields on~$N_{t}$ (for the last two operators, it suffices to note that~$A_t=\tilde{H}_\eta$ if~$m=0$ and~$u=\partial_t$) and the operator~$\tilde{Q}(\omega,\omega') : \H_{t,g}\rightarrow C_{0}^{\infty}(N_t,SM)$ is well-defined. Moreover, that the co-domain of the operator~$\Tilde{Q}(\omega,\omega')$ is the space of smooth and compactly supported spinor fields on~$N_{t}$ (instead of only smooth) follows from the compactness assumption on the coefficients of~$\Delta A$.

\begin{proof}[Proof of Proposition~\ref{prop:Bt_conf_flat}]
    Note that most of the arguments used in this proof are inspired by~\cite{baryogenesis} and~\cite{baryomink} and, in particular, we will use some of the results derived there.
    \par In the first place, for~$\omega$ in the resolvent set of~$\tilde{H}_\eta$ (the same as the one of~$H_\eta$) and~$A_t$, the resolvent operator admits the expansion 
    \begin{equation}\label{eq:Bt_expansion_conf_flat}
        R_\omega(A_t)=\sum_{p=0}^{\infty}(-R_\omega(\tilde{H}_\eta)\Delta A(t))^{p}R_\omega(\tilde{H}_\eta)=:\sum_{p=0}^{\infty}R_\omega^{(p)}(A_t)\:,
    \end{equation}
    where we made use of the Neumann series, which converges in the operator norm by the assumption on~$\Delta A(t)$ given by expression~\eqref{eq:small_Delta_A}. Using Stone's formula and that $A_{t}$ has an absolutely continuous spectrum leads to the following expansion for the spectral projection operator  
    \begin{align*}
     \chi_{I}(A_{t})&=\int_{I} F_{\omega}(A_{t})d\omega=\frac{1}{2\pi i}\slim_{\delta\to0^{+}}\int_{I}R_{\omega+is\delta}(A_{t})\big|_{s=-1}^{s=1}d\omega\\
     &=\sum_{p=0}^{\infty}\frac{1}{2\pi i}\slim_{\delta\to0^{+}}\int_{I}R^{(p)}_{\omega+is\delta}(A_{t})\big|_{s=-1}^{s=1}d\omega=:\sum_{p=0}^{\infty}\chi^{(p)}_{I}(A_{t})\;.
    \end{align*}
    Clearly, for the zero order contribution we have that $\chi_{I}^{(0)}(A_{t})=\chi_{I}(\tilde{H}_{\eta})$. Hence, the operator product appearing in the definition of the trace can be expanded as follows
\begin{align*}
    &\eta_{\Lambda}(\tilde{H}_\eta)\big(\chi_I(A_t)-\chi_{I}(\tilde{H}_{\eta})\big)=\sum_{p=0}^{\infty}\eta_{\Lambda}(\tilde{H}_\eta)\Big(\frac{1}{2\pi i}\slim_{\delta\to0^{+}}\int_{I}R^{(p)}_{\omega+is\delta}(A_{t})\big|_{s=-1}^{s=1}d\omega-\chi_I(\tilde{H}_{\eta})\Big)\\
    &=\sum_{p=1}^{\infty}\eta_{\Lambda}(\tilde{H}_\eta)\frac{1}{2\pi i}\slim_{\delta\to0^{+}}\int_{I}R^{(p)}_{\omega+is\delta}(A_{t})\big|_{s=-1}^{s=1}d\omega\\
    &=\sum_{p=1}^{\infty}\frac{1}{2\pi i}\slim_{\delta\to0^{+}}\int_{-\frac{1}{\varepsilon}}^{-m}d\omega\int_{-\infty}^{\infty}d\omega'\eta_{\Lambda}(\omega')F_{\omega'}(\tilde{H}_\eta)R^{(p)}_{\omega+is\delta}(A_t)\big|_{s=-1}^{s=1}\;.
    %&B_t\coloneqq \frac{d}{dt}\text{\rm{tr}}_{\overline{\H^{\varepsilon}_t}} \big(( \eta_{\Lambda}(\tilde{H}_\eta)\big(\chi_I(A_t)-\chi_{I}(\tilde{H}_{\eta})\big)\big)=\\
\end{align*}
    That this operator product is trace-class follows from Lemma~\ref{lem:trace_class} (since~$F_{\omega}(\tilde{H}_{\eta})$ is again an integral operator with a smooth kernel and~$\Delta A$ a differential operator with smooth compactly supported coefficients). Hence, the perturbative expansion of the rate of baryogenesis is
    %where the~$p^\text{th}$ term of this expansion has the form
    \begin{align*}
     &B_t=\sum_{p=0}^{\infty}B^{(p)}_t\\
     &B_t^{(p)}\coloneqq \frac{1}{2\pi i}\slim_{\delta\to0^{+}}\int_{-\frac{1}{\varepsilon}}^{-m}\!\!\!d\omega\int_{-\infty}^{\infty} \!\!\!d\omega'\eta_{\Lambda}(\omega')\frac{d}{dt}\Big(\tr_{\overline{\H_t^{\varepsilon}}}\big(F_{\omega'}(\tilde{H}_\eta)R^{(p)}_{\omega+is\delta}(A_t)\big)\big|_{s=-1}^{s=1}\Big)
    \end{align*}
    where $B_{t}^{(0)}=0$ as discussed above. Simpler formulas for the first and second order contribution to the rate of baryogenesis were derived in \cite[Theorem 7.6]{baryomink} for Minkowski spacetime and also hold for a general conformally flat spacetime (simply replacing~$H_\eta$ with~$\tilde{H}_\eta$), yielding expressions
    ~\eqref{eq:Bt1_conf_flat} and~\eqref{eq:Bt2_conf_flat}. The only difference in this case with respect to Minkowski spacetime is that for a general conformally flat spacetime  
    \begin{equation*}
        \frac{d}{dt}\tr_{\overline{\H_t^{\varepsilon}}}\big(\Delta AF_{-m}(\tilde{H}_\eta)\big)\neq \tr_{\overline{\H_t^{\varepsilon}}}\Big(\frac{d}{dt}(A_t)F_{-m}(\tilde{H}_\eta)\Big)\;.
    \end{equation*}
    Finally, also~$B^{(1)}_t$ vanishes in conformally flat spacetimes
    \begin{align*}
       B^{(1)}_t&=\frac{d}{dt}\tr_{\overline{\H_t^{\varepsilon}}}\Big(\Delta AF_{-m}(\tilde{H}_\eta)\Big)=\frac{d}{dt}\tr_{\overline{\H_t^{\varepsilon}}}\Big(\Delta A(t)\big(\Omega^{-3/2}F_{-m}(H_\eta)\Omega^{3/2}\big)\Big)\\
       &=\frac{d}{dt}\tr_{\overline{\H_{\eta}^{\varepsilon}}}\Big((\Omega^{3/2}\Delta A\Omega^{-3/2})F_{-m}(H_\eta)\Big)\\
       &=\frac{d}{dt}\int_{\mathbb{R}^{3}}\Tr_{\C^4}\big((\Omega^{3/2}\Delta A\Omega^{-3/2})F_{-m}(x,y)\big)\Big|_{y=x}d^{3}x=0\;,
    \end{align*}
    where~$F_{-m}(x,y)$ is the kernel of the integral operator~$F_{-m}(H_\eta)$ and in the final step we simply used that 
    \begin{equation*}
        F_{-m}(x,x)=(\partial_{\mu}F_{-m}(x,y))|_{y=x}=0 \;,   
    \end{equation*}
    see the appendix of~\cite{baryomink} for the proof.
\end{proof}

Recall that in Lemma~\ref{lem:At_diagonalizable} it was proven that if~$m=0$ and~$u=\partial_t$, the rate of baryogenesis vanishes. Note that this also follows from the previous proposition and expression~\eqref{eq:Bt_expansion_conf_flat} (which describes the perturbative power expansion of~$B_t$): if~$u=\partial_t$ and~$m=0$,~$A_t$ and~$\tilde{H}_\eta$ agree and~$\Delta A=0$. Hence, we can already see that the two main perturbative parameters which trigger baryogenesis are the mass~$m$ and the regularizing vector field~$u$. We will study the effects of~$m$ and~$u$ separately: in Section~\ref{sec:utrivial} we will assume that~$m\neq 0$ and~$u=\partial_t$, whereas in Section~\ref{sec:ugeneral} we will consider that~$m=0$ and~$u\neq \partial_t$. The most general case (i.e.~$m\neq0$ and~$u\neq \partial_t$) will follow easily from the analysis of the separate effects.
\par However, before analyzing the second order contribution to the rate of baryogenesis in the two aforementioned scenarios, we will recall the basic setup and some important results in Minkowski spacetime which will be used when considering more general conformally flat spacetimes.

\section{Setting the stage in Minkowski spacetime}\label{sec:baryomink}
We briefly recall the basic setup and some of the results of the study of baryogenesis in Minkowski spacetime in~\cite{baryomink} as this will play an important role when analyzing general conformally flat spacetimes. When considering Minkowski spacetime, the superscript~$\eta$ will be added to some of the mathematical objects (e.g.~$B_t^{\eta}, A_t^{\eta}$, etc).

In the first place, let~$(N_t)_{t\in\mathbb{R}}$ be the foliation of Minkowski spacetime given by the level sets of the
global time function~$t$ and given an initial time~$t_0$, consider a compact subset~$V\subset N_{t_{0}}$. Assume there exists a timelike vector field~$u$ which in the subset~$V$ is of the form
\begin{equation}
    u_{p}=(1+\lambda f_{p})\nu+\lambda X_{p} \qquad \text{for all~$p\in V\subset N_{t_{0}}$}\;,
\end{equation}
where~$f\in C^{\infty}(\mathbb{R}^{3},\mathbb{R}_{>0})$ is a positive and smooth function (with~$f_{p}=f (x,y,z)$) and~$X$ is a spacelike vector field~$X$. Outside of~$V$, we impose that~$u=\nu$. Assuming that the vector field~$u : N_{t_{0}}\rightarrow TM$ follows the locally rigid dynamics (recall Definition~\ref{def:locally_rigid_dynamics}), a global regularizing vector field~$u: M \rightarrow TM$ is obtained. The evolution equation which governs the dynamics of~$u$ is given (to first order in~$\lambda$) by (\cite[Lemma 7.1]{baryomink}).

\begin{align}\label{eq:dynamical_eq_u}
    \frac{du_{p}}{dt}&=-\textrm{grad}_{\delta}(\tilde{f}_{p}^{-1})+\frac{\lambda}{\tilde{f}_{p}^{3}}\Big(\frac{\tilde{f}_{p}}{3}\textrm{div}_{\delta}\Big(X_{p}\Big)+4X_{p}(\tilde{f}_{p})\Big)\nu+\mathcal{O}(\lambda^{2})\nonumber\\
    &=\lambda\Big[-\textrm{grad}_{\delta}(f_{p}^{-1})+\frac{1}{3}\textrm{div}_{\delta}\Big(X_{p}\Big)\nu\Big]+\mathcal{O}(\lambda^{2})\;,
\end{align}
with~$\tilde{f}_{p}=1+\lambda f_{p}$ and~$p\in V$. Moreover, by Proposition~\ref{prop:u_conformal_invariant}, it follows that~$u$ presents the same locally rigid evolution equation in any conformally flat spacetime. In particular, if initially~$u|_{N_{t_{0}}}=\partial_t$ (i.e.~$\lambda=0$), then at any later time~$t>t_{0}$ the regularizing vector field remains unchanged, i.e.~$u|_{N_t}=\partial_t$. Note that in this setup the conditions of Proposition~\ref{prop:Bt_conf_flat} are easily satisfied: ~$\Delta A(t)$ is linear in~$\lambda$, which can be chosen sufficiently small in order to fulfill condition \eqref{eq:small_Delta_A} and, by construction, $\Delta A$ has compactly supported coefficients (supported in $V$).

In the following lemma, we show that in Minkowski spacetime the operator product
\begin{align*}
    \eta_{\Lambda}(H_\eta)\big(\chi_I(A^{\eta}_t)-\chi_{I}(H_{\eta})\big) :\H_{t,\eta}\rightarrow\H_{t,\eta}    
\end{align*}
is trace-class if the conditions of Proposition~\ref{prop:Bt_conf_flat} are fulfilled. The proof simply relies on the fact that, in Minkowski spacetime, operator products of the form $F_{\omega'}(H_\eta)R^{(p)}_{\omega}(A^{\eta}_t)$ can be rewritten (for any $p\in\N$) as an integral operator with a smooth compactly supported kernel, which is trace-class by ~\cite[Lemma 7.4]{baryomink}). Moreover, also in a conformally flat spacetime satisfying the conditions of Proposition~\ref{prop:Bt_conf_flat}, operator products of the form
\begin{align*}
    F_{\omega'}(\tilde{H}_\eta)R^{(p)}_{\omega}(A_t)
\end{align*}
correspond to an integral operator with a smooth compactly supported kernel (since $F_{\omega}(\tilde{H}_{\eta})$ is again an integral operator with a smooth kernel). Thus, in particular, the following lemma implies that also in a general conformally flat spacetime satisfying the conditions of Proposition~\ref{prop:Bt_conf_flat}, the operator product
\begin{align*}
    \eta_{\Lambda}(\tilde{H}_\eta)\big(\chi_I(A_t)-\chi_{I}(\tilde{H}_{\eta})\big) :\H_{t,g}\rightarrow\H_{t,g}   
\end{align*}
is trace-class.
\begin{Lemma}\label{lem:trace_class}
    In Minkowski spacetime $(\R^{1,3},\eta)$, assume that the differential operator~$\Delta A=A_{t}^{\eta}-H_{\eta}$ has smooth and compactly supported coefficients, $A_{t}^{\eta}$ has an absolutely continuous spectrum and that for any~$\omega\in\rho(H_{\eta})$, it holds that 
    \begin{align*}
        \|R_\omega(H_\eta)\Delta A(t)\|<1
    \end{align*}
    Then, for any~$\omega\in\rho(H_{\eta})$ and~$p\in\N$ the following operator is trace-class
    \begin{equation}\label{eq:trace_class}
        \int_{-\infty}^{\infty} d\omega'F_{\omega'}(H_{\eta})R^{(p)}_{\omega}(A^{\eta}_t)\;,
    \end{equation}
    where~$R^{(p)}_{\omega}(A^{\eta}_t)$ correspond to the $p^{\text{\rm{th}}}$-element of the perturbative expansion of the resolvent operator (cf. expression \eqref{eq:Bt_expansion_conf_flat}). Moreover, the operator product
    \begin{align} \label{opprod}
        \eta_{\Lambda}(H_{\eta})\big(\chi_{I}(A^{\eta}_{t})-\chi_{I}(H_{\eta})\big):\H_{t,\eta}\rightarrow\H_{t,\eta}\;,
    \end{align}
    where $I=(-\frac{1}{\varepsilon},-m)$, is trace-class.
\end{Lemma}
\begin{proof}
Let~$p\geq1$. Using the functional calculus of~$H_\eta$
\begin{align}\label{eq:perturbative_exp_operator_prod}
    &\int_{-\infty}^{\infty} d\omega'F_{\omega'}(H_{\eta})R^{(p)}_{\omega}(A^{\eta}_t)=\int_{-\infty}^{\infty} d\omega'F_{\omega'}(H_{\eta})(-R_\omega(H_\eta)\Delta A(t))^{p}R_\omega(H_\eta)\nonumber\\
    &=(-1)^{p}\int_{-\infty}^{\infty} d\omega'\int_{-\infty}^{\infty} d\omega''\frac{1}{\omega''-\omega}F_{\omega'}(H_{\eta})(R_\omega(H_\eta)\Delta A(t))^{p}F_{\omega''}(H_\eta)\nonumber\\
    &=(-1)^{p}\int_{-\infty}^{\infty} d\omega'\int_{-\infty}^{\infty} \frac{d\omega''}{\omega''-\omega}\Big(\prod_{l=1}^{p}\int_{-\infty}^{\infty}\frac{ d\omega_{l}}{\omega_{l}-\omega} \Big) F_{\omega'}(H_{\eta})\Big(\prod_{l=1}^{p}F_{\omega_{l}}(H_{\eta})\Delta A(t)\Big)F_{\omega''}(H_\eta)\nonumber\\
    &=(-1)^{p}\int_{-\infty}^{\infty} \frac{d\omega''}{\omega''-\omega}\Big(\prod_{l=1}^{p}\int_{-\infty}^{\infty}\frac{ d\omega_{l}}{\omega_{l}-\omega} \Big)\Big(\prod_{l=1}^{p}F_{\omega_{l}}(H_{\eta})\Delta A(t)\Big)F_{\omega''}(H_\eta)
\end{align}
where in the last line we simply used that~$F_{\omega'}(H_{\eta})F_{\omega_{1}}(H_\eta)=F_{\omega'}(H_\eta)\delta(\omega'-\omega_{1})$. Note that an operator product of the form 
\begin{align}\label{eq:prod_op}
    \prod_{l=1}^{p}F_{\omega_{l}}(H_{\eta})\Delta A(t)
\end{align}
is trace-class by Lemma 7.4 in \cite{baryomink} since it corresponds to an integral operator with a smooth and compactly supported kernel (in \cite{baryomink} it is shown for~$p=1$ and~$p=2$; however the proof, which relies on Mercer's theorem, holds for an arbitrary~$p\in\N$). By the same argument, composing the trace-class operator product \eqref{eq:prod_op} with the operator $F_{\omega''}(H_{\eta})$ in order to obtain the one appearing in expression \eqref{eq:perturbative_exp_operator_prod} yields again a trace-class operator. Moreover, the map $\omega_{l}\in\sigma(H_{\eta})\mapsto\frac{1}{\omega_{l}-\omega}$ is bounded for any $\omega_{l}\in\{\omega_{1},\ldots,\omega_{p},\omega_{p+1}\}$ (where $\omega_{p+1}:=\omega''$) since $\omega$ is in the resolvent set $\rho(H_{\eta})\coloneqq\mathbb{C}\setminus\sigma(H_\eta)$. Hence, the trace-norm (which we denote by $\|\cdot\|_{1}$) of the operator \eqref{eq:perturbative_exp_operator_prod} is finite
\begin{align*}
    &\Big\| \int_{-\infty}^{\infty} d\omega'F_{\omega'}(H_{\eta})R^{(p)}_{\omega}(A^{\eta}_t)\Big\|_{1}\\
    &\leq \prod_{l=1}^{p+1}\int_{-\infty}^{\infty} d\omega_{l}\frac{1}{|\omega_{l}-\omega|} 
    \Big\|\Big(\prod_{l=1}^{p}F_{\omega_{l}}(H_{\eta})\Delta A(t)\Big)F_{\omega''}(H_\eta)\Big\|_{1}<\infty\;,
\end{align*}
so the operator \eqref{eq:trace_class} is trace-class.
By Proposition~\ref{prop:Bt_conf_flat}, it follows that the operator product in~\eqref{opprod} is trace-class.
\end{proof}

\section{Scenario 1: a trivial regularizing vector field}\label{sec:utrivial}
In the following lemma we choose an arbitrary subset~$\H^{\varepsilon}_\eta$ of the space of smooth spinors on a slice~$N_{t}$ in Minkowski spacetime. The reason that here we decided to label this subspace with~$\H^{\varepsilon}_\eta$ is in preparation of Corollary~\ref{cor:Bt_conf_flat_uparallel}: there we will apply the following lemma with~$\H^{\varepsilon}_\eta$ the image of the unitary operator~$\tilde{U}$ (acting on a subspace $\H^{\varepsilon}_{t}\subset\H_{t,g}$).

\begin{Lemma}\label{lem:Bt2_multiplication_op}
Let~$T_{j}: \H^{\varepsilon}_\eta\subset C^{\infty}(N_t,SM)\rightarrow C_{0}^{\infty}(N_t,SM)$ with~$j\in\{1,2\}$ denote two smooth and compactly supported multiplication operators which satisfy that at every point~$x\in N_t$ and for every~$\psi\in \H^{\varepsilon}_\eta$,~$ (T_{j}\psi)(x)=\alpha_{j}(x)\gamma_{\eta 0}\psi(x)$ with~$\alpha_{j}(x)\in\mathbb{C}$. Then, for any~$\omega,\omega'\in \sigma(H_\eta)$ it holds that
\begin{align}\label{eq:Bt2_multiplication_op}
I_{T_{1}T_{2}}&:=\int_{-\infty}^{\infty}d\omega'\int_{-\infty}^{\infty}d\omega\;\partial_\omega\big(\eta_{\Lambda}(\omega)\text{\rm{tr}}_{\overline{\H^{\varepsilon}_\eta}}(K_{T_{1}T_{2}})\big)\frac{g(\omega') - g(\omega)}{\omega'-\omega}\nonumber\\
&=-2\int_{0}^{\infty} \frac{d\rho}{(2\pi)^{4}}\hat{\alpha}_{1}(\rho)\hat{\alpha}_{2}(-\rho)K(\rho)\;,
\end{align}
where~$K_{T_{1}T_{2}}=T_{1}F_\omega(H_\eta)T_{2}F_{\omega'}(H_\eta)$ is assumed to be trace-class,~$g$ is the characteristic function of the interval~$(-\frac{1}{\varepsilon},-m)$ and the kernel~$K$ is given by expression~\eqref{eq:kernel_Q}.
\end{Lemma}
\begin{proof}
In this proof, all tangent vectors are in~$\mathbb{R}^{3}$. The (Euclidean) scalar product will be denoted by ``$\cdot$", i.e.~$k\cdot x:=\delta_{\mu\nu}k^{\mu}x^{\nu}$ (where~$\delta$ is the Euclidean metric) and analogously we define~$\gamma_\eta\cdot k:=\delta_{\mu\nu}\gamma_\eta^{\mu}x^{\nu}$ and~$|k|:=\sqrt{k\cdot k}$.

In the first place, note that since~$g$ is the characteristic function of the interval~$(-1/\varepsilon,-m)$, the only contribution to the integral~\eqref{eq:Bt2_multiplication_op} is when~$\omega'\in(-\infty,-m)$ and~$\omega\in(m,\infty)$ or when~$\omega\in(-\infty,-m)$ and~$\omega'\in(m,\infty)$. In other words,
\[ I_{T_{1}T_{2}}=J_{+}+J_{-}\;, \]
where,~$J_{+}$ and~$J_{-}$ are defined by
\begin{align*}
    &J_{+}:=\int_{-\infty}^{-m}d\omega'\int_{m}^{\infty}d\omega\;\partial_\omega\big(\eta_{\Lambda}(\omega)\text{\rm{tr}}_{\overline{\H^{\varepsilon}_\eta}}(K_{T_{1}T_{2}})\big)\frac{1}{\omega'-\omega}\\
    &J_{-}:=-\int_{m}^{\infty}d\omega'\int_{-\infty}^{-m}d\omega\;\partial_\omega\big(\eta_{\Lambda}(\omega)\text{\rm{tr}}_{\overline{\H^{\varepsilon}_\eta}}(K_{T_{1}T_{2}})\big)\frac{1}{\omega'-\omega}\;.
\end{align*}

Secondly, given~$x,y\in \mathbb{R}^{3}$ and~$\omega\in(-\infty,-m)$, an explicit expression for the kernel~$F_\omega(x,y)$ of the integral operator~$F_\omega(H_\eta)$ was derived in~\cite[Appendix]{baryomink}
\begin{align}
    &F_\omega(x,y)=\int \frac{d^{3}k}{(2\pi)^{3}}\Hat{F}_\omega(k)\delta(\omega^{2}-\omega_{k}^{2})e^{ik\cdot(x-y)}\label{eq:kernelFomega}\\
    &\textrm{with}\hspace{0.5cm}\Hat{F}_\omega(k)\coloneqq-(\gamma_{\eta0}\omega-\gamma_\eta\cdot k+m)\gamma_{\eta0}\Theta(1+\varepsilon \omega)\label{eq:Fomega_k}\;,
\end{align}
where~$k\in\mathbb{R}^{3}$ and~$\omega_{k}\coloneqq(|k|^{2}+m^{2})^{1/2}$. In the case that~$\omega'\in(m,\infty)$,~$\Theta(1+\varepsilon \omega)$ has to be replaced with~$\Theta(1-\varepsilon \omega')$ in~\eqref{eq:Fomega_k}. 
\par With the help of the Fourier transformation~$\mathcal{F}: L^{1}(\mathbb{R}^{3})\rightarrow L^{1}(\mathbb{R}^{3})$, we define the Fourier conjugated operators~$\Hat{T}_{j}:=\mathcal{F}T_{j}\mathcal{F}^{-1}$,~$\Hat{F}_\omega(H_\eta):=\mathcal{F}F_\omega(H_\eta)\mathcal{F}^{-1}$ and~$\Hat{Q}_{j,\omega}\coloneqq\Hat{T}_{j}\Hat{F}_\omega(H_\eta)$, where~$j\in\{1,2\}$. In particular, they are integral operators
\begin{align}
    & (\Hat{T}_{j}\Hat{\psi})(p)=\int d^{3}k\hat{\alpha}_{j}(p-k)\gamma_{\eta0}\Hat{\psi}(k) & & \textrm{with} \hspace{0.4cm}\hat{\alpha}_{j}(p-k)=\frac{1}{(2\pi)^{3}}(\mathcal{F}\alpha_{j})(p-k) \label{eq:FourierT1}\\
    & (\Hat{F}_\omega(H_\eta)\Hat{\psi})(p)=\int d^{3}k \Hat{F}_\omega(p,k)\Hat{\psi}(k)& & \textrm{with}\hspace{0.4cm}\Hat{F}_\omega(p,k)\!=\!\!\int\!\!\frac{d^{3}x}{(2\pi)^{3}}\!\!\int\!\! d^{3}ye^{-ip\cdot x}F_\omega(x,y)e^{ik\cdot y}\nonumber\\
    & (\Hat{Q}_{j,\omega}\Hat{\psi})(p)=\int d^{3}k \Hat{Q}_{j,\omega}(p,k)\Hat{\psi}(k)& & \textrm{with}\hspace{0.4cm}\Hat{Q}_{j,\omega}(p,k)\!=\!\!\int\!\! d^{3}k'\hat{\alpha}_{j}(p-k')\gamma_{\eta0}\Hat{F}_\omega(k',k)\nonumber
\end{align}
where in all the previous expressions integration is over~$\mathbb{R}^{3}$ (for the sake of conciseness we will continue omitting the domain~$\mathbb{R}^{3}$ over which the integration is performed). It then follows that the trace of the Fourier transformed operator $\Hat{K}_{T_{1}T_{2}}\coloneqq\Hat{T}_{1}\Hat{F}_\omega(H_\eta)\Hat{T}_{2}\Hat{F}_{\omega'}(H_\eta)=\Hat{Q}_{1,\omega}\Hat{Q}_{2,\omega'}=\mathcal{F}K_{T_{1}T_{2}}\mathcal{F}^{-1}$ is
\begin{align*} %\label{eq:trace1}
    (\Hat{K}_{T_{1}T_{2}}\hat{\psi})(p)\!& =\!\int\!\! d^{3}k \Hat{K}_{T_{1}T_{2}}(p,k)\hat{\psi}(k) \hspace{0.4cm} \textrm{with}\hspace{0.4cm}\Hat{K}_{T_{1}T_{2}}(p,k)\!=\!\!\int \!\!d^{3}k'\Hat{Q}_{1,\omega}(p,k')\Hat{Q}_{2,\omega'}(k',k)\nonumber\\
    \tr_{\overline{\H^{\varepsilon}_\eta}}(K_{T_{1}T_{2}})&=\tr_{\overline{\H^{\varepsilon}_\eta}}(\Hat{K}_{T_{1}T_{2}})=\int d^{3}k\Tr_{\mathbb{C}^{4}}{(\Hat{K}_{T_{1}T_{2}}(k,k))}\\
    &=\int d^{3}k\int d^{3}k'\Tr_{\mathbb{C}^{4}}{(\Hat{Q}_{1,\omega}(k,k')\Hat{Q}_{2,\omega'}(k',k))}\;,
\end{align*}
where we used Lemma~\ref{lem:time_independent} and Mercer's Theorem. Given our explicit expression~\eqref{eq:kernelFomega} for the kernel~$F_\omega(x,y)$, it is easy to show that the kernel~$\Hat{F}_\omega(p,k)$ (and thus also $\Hat{Q}_{j,\omega}(p,k)$) presents the following form:
    \begin{align*}%\label{eq:Fomegapq}
        \Hat{F}_\omega(p,k)&=\int \frac{d^{3}x}{(2\pi)^{3}}\int d^{3}y e^{-ip\cdot x}F_\omega(x,y)e^{ik\cdot y}\nonumber\\
        &=\int \frac{d^{3}x}{(2\pi)^{3}}\int d^{3}y \int \frac{d^{3}k'}{(2\pi)^{3}}e^{-ip\cdot x}\Hat{F}_\omega(k')\delta(\omega^{2}-\omega_{k'}^{2})e^{ik'\cdot(x-y)}e^{ik\cdot y}\nonumber\\
        &=\int d^{3}k'\Hat{F}_\omega(k')\delta(\omega^{2}-\omega_{k'}^{2})\int \frac{d^{3}x}{(2\pi)^{3}} e^{ix\cdot(k'-p)}\int \frac{d^{3}y}{(2\pi)^{3}} e^{iy\cdot(k-k')}\nonumber\\
        &=\int d^{3}k'\Hat{F}_\omega(k')\delta(\omega^{2}-\omega_{k'}^{2})\delta^{(3)}(k'-p)\delta^{(3)}(k-k')\\
        &=\Hat{F}_\omega(p)\delta(\omega^{2}-\omega_{p}^{2})\delta^{(3)}(k-p)
    \end{align*}
This simplifies the kernels~$\Hat{Q}_{j,\omega}(p,k)$ and the trace of~$\Hat{K}_{12}$ considerably
% trace in~\eqref{eq:trace1} considerably (in particular, the distributions~$\delta^{(3)}(q-k)$ take care of six out of the twelve integrals appearing in~\eqref{eq:trace1})
\begin{align*}
     &\Hat{Q}_{j,\omega}(p,k)=\int d^{3}k'\hat{\alpha}_{j}(p-k')\gamma_{\eta0} \Hat{F}_\omega(k')\delta(\omega^{2}-\omega_{k'}^{2})\delta^{(3)}(k-k')\\
        &\;\quad\quad\quad\quad=\hat{\alpha}_{j}(p-k)\gamma_{\eta0}\Hat{F}_\omega(k)\delta(\omega^{2}-\omega_{k}^{2})\\
    &\implies\tr_{\overline{\H^{\varepsilon}_\eta}}(K_{T_{1}T_{2}})=\int d^{3}k\int d^{3}k'\Tr_{\mathbb{C}^{4}}{(\Hat{Q}_{1,\omega}(k,k')\Hat{Q}_{2,\omega'}(k',k))}\\
    &=\!\!\int \!\!\frac{d^{3}k}{(2\pi)^{3}}\!\!\int\!\! \frac{d^{3}k'}{(2\pi)^{3}}\hat{\alpha}_{1}(k-k')\hat{\alpha}_{2}(k'-k)\Tr_{\mathbb{C}^{4}}{[\gamma_{\eta0}\Hat{F}_\omega(k')\gamma_{\eta0}\Hat{F}_{\omega'}(k)]}\delta(\omega^{2}-\omega_{k}^{2})\delta(\omega'^{2}-\omega_{k'}^{2})
\end{align*}

The computation of the trace gives
\begin{align}\label{eq:trace_mult_ops}
    &\chi(\omega,\omega')\coloneqq\Tr_{\mathbb{C}^{4}}{[\gamma_{\eta0}\Hat{F}_\omega(k')\gamma_{\eta0}\Hat{F}_{\omega'}(k)]}=4(\omega'\omega+m^{2}-k\cdot k')\;,
\end{align}
where we used that the trace of an odd number of gamma matrices vanishes as well as the terms proportional to~$\gamma_\eta^{0}\gamma_\eta^{\mu}$ (since~$\Tr_{\mathbb{C}^{4}}{(\gamma_\eta^{0}\gamma_\eta^{\mu})}=4g^{\mu0}=0$). So, we have that
\begin{align*}
    &J_{+}=\int \frac{d^{3}k}{(2\pi)^{3}}\int \frac{d^{3}k'}{(2\pi)^{3}}\hat{\alpha}_{1}(k-k')\hat{\alpha}_{2}(k'-k)\int_{-\infty}^{\infty}d\omega\int_{-\infty}^{\infty}d\omega'\\
    &\;\times\frac{1}{\omega'-\omega}\partial_\omega\big(r(\omega,\omega')\chi(\omega,\omega')\delta(\omega^{2}-\omega_{k}^{2})\delta(\omega'^{2}-\omega_{k'}^{2} )\big) \Theta(\omega-m)\Theta(-\omega'-m)\;,
\end{align*}
where~$r(\omega,\omega'):=\eta_{\Lambda}(\omega)\Theta(1+\varepsilon\omega)\Theta(1-\varepsilon\omega')$. In order to integrate over~$\omega$ and~$\omega'$ we use that~$\delta(\omega^{2}-\omega_{k}^{2})=\frac{1}{2\omega_{k}}(\delta(\omega-\omega_{k})+\delta(\omega+\omega_{k}))$ and~$\delta(\omega'^{2}-\omega_{k'}^{2})=\frac{1}{2\omega_{k'}}(\delta(\omega'-\omega_{k'})+\delta(\omega'+\omega_{k'}))$. Then
\begin{align}
    &\int_{-\infty}^{\infty}d\omega\int_{-\infty}^{\infty}d\omega'\frac{1}{\omega'-\omega}\partial_\omega\big(r(\omega,\omega')\chi(\omega,\omega')\delta(\omega^{2}-\omega_{k}^{2})\delta(\omega'^{2}-\omega_{k'}^{2} )\big)\nonumber\\
    &\quad\quad\quad\quad\quad\quad\quad\quad\quad\quad\quad\times\Theta(\omega-m)\Theta(-\omega'-m)\nonumber\\
    &=\int_{m}^{\infty}d\omega\int_{-\infty}^{-m}d\omega'\frac{1}{\omega'-\omega}\frac{1}{4\omega_{k}\omega_{k'}}\partial_\omega\big(r(\omega,\omega')\chi(\omega,\omega')\delta(\omega-\omega_{k})\delta(\omega'+\omega_{k'})\big)\nonumber\\
    &=\int_{m}^{\infty}d\omega\frac{1}{-\omega_{k'}-\omega}\frac{1}{4\omega_{k}\omega_{k'}}\Big(\partial_\omega\big(r(\omega,-\omega_{k'})\chi(\omega,-\omega_{k'})\big)\delta(\omega-\omega_{k})\nonumber\\
    &\quad\quad+\partial_\omega\big(\delta(\omega-\omega_{k})\big)r(\omega,-\omega_{k'})\chi(\omega,-\omega_{k'})\Big)\nonumber\\
    &=\frac{1}{-\omega_{k'}-\omega_{k}}\frac{1}{4\omega_{k}\omega_{k'}}\partial_\omega\big(r(\omega,-\omega_{k'})\chi(\omega,-\omega_{k'})\big)\big|_{\omega=\omega_{k}}\nonumber\\
    &\quad-\frac{1}{4\omega_{k}\omega_{k'}}\partial_\omega\Big(\frac{r(\omega,-\omega_{k'})}{-\omega_{k'}-\omega}\chi(\omega,-\omega_{k'})\Big)\Big|_{\omega=\omega_{k}}\nonumber\\
    &=-\frac{1}{(\omega_{k'}+\omega_{k})^{2}}\frac{1}{4\omega_{k}\omega_{k'}}r(\omega_{k},-\omega_{k'})\chi(\omega_{k},-\omega_{k'})\;,\label{eq:omega_integral}\\
    &J_{+}=-\int \frac{d^{3}k}{(2\pi)^{3}}\int \frac{d^{3}k'}{(2\pi)^{3}}\hat{\alpha}_{1}(k-k')\hat{\alpha}_{2}(k'-k)\frac{r(\omega_{k},-\omega_{k'})}{(\omega_{k'}+\omega_{k})^{2}}\frac{1}{4\omega_{k}\omega_{k'}}\chi(\omega_{k},-\omega_{k'})\label{eq:traceBt2_intermediate}
\end{align}
That the multiplication operators are smooth and have a compact support implies that~$\hat{\alpha}_{1}$ and~$\hat{\alpha}_{2}$ are Schwartz functions and thus that they have a rapid decay in~$k-k'$ (they decay faster than any inverse polynomial). Therefore, the main contribution of the integrand to the~$k,k'$ integrals in~$J_{+}$ is for~$k$ close to~$k'$. Since expression~\eqref{eq:omega_integral} corresponds to an inverse polynomial in~$k$ and~$k'$, that~$k\approx k'$ implies that the main contribution to the integrals in~\eqref{eq:traceBt2_intermediate} is, in particular, for~$k$ and~$k'$ close to zero, i.e.\ when~$\omega_{k}\approx\omega_{k'}\approx m$. Moreover, for~$\omega_{k}\approx\omega_{k'}\approx m$, clearly~$\Theta(1+\varepsilon\omega_{k})=\Theta(1-\varepsilon\omega_{k'})=1$ and~$\eta_{\Lambda}(\omega)=1$ (as~$\Lambda\gg m$ and there exists a sufficiently small~$\delta>0$ such that~$\eta_{\Lambda}(\omega)=1$ for~$\omega\in(-\Lambda+\delta,\Lambda-\delta)$, cf.\ Remark~\ref{rem:baryo_rate}). For this reason, without loss of generality we can simply set~$r(\omega_{k},-\omega_{k'})$ equal to~$1$ in expression~\eqref{eq:traceBt2_intermediate}.

Proceeding analogously for the integral~$J_{-}$ (and omitting~$r(\omega,\omega')$ by the previous argument) yields
\begin{align}
    &\int_{-\infty}^{\infty}d\omega\int_{-\infty}^{\infty}d\omega'\frac{1}{\omega'-\omega}\partial_\omega\big(\chi(\omega,\omega')\delta(\omega^{2}-\omega_{k}^{2})\delta(\omega'^{2}-\omega_{k'}^{2} )\big)\Theta(-\omega-m)\Theta(\omega'-m)\nonumber\\
    &=\frac{1}{(\omega_{k'}+\omega_{k})^{2}}\frac{1}{4\omega_{k}\omega_{k'}}\chi(-\omega_{k},\omega_{k'})\label{eq:integral2_omega}\;,\\
    & J_{-}=-\int \frac{d^{3}k}{(2\pi)^{3}}\int \frac{d^{3}k'}{(2\pi)^{3}}\hat{\alpha}_{1}(k-k')\hat{\alpha}_{2}(k'-k)\frac{1}{(\omega_{k'}+\omega_{k})^{2}}\frac{1}{4\omega_{k}\omega_{k'}}\chi(-\omega_{k},\omega_{k'})\;,
\end{align}
where the relative sign difference between~\eqref{eq:omega_integral} and~\eqref{eq:integral2_omega} stems from the factor~$\frac{1}{\omega'-\omega}$ appearing in both integrals and the replacement~$\omega=\omega_{k}$ and~$ \omega'=-\omega_{k'}$ in the first one and~$\omega=-\omega_{k}$ and~$ \omega'=\omega_{k'}$ in the second one. Moreover, since~$\chi(\omega_{k},-\omega_{k'})=\chi(-\omega_{k},\omega_{k'})$, the two integrals agree ($J_{+}=J_{-}$). Note that for more general multiplication operators~$T_{j}$,~$\chi(\omega,-\omega')\neq\chi(-\omega,\omega')$ (more on this in Remark~\ref{rem:Bt2_general_multiplication}~(ii)).

We introduce the function~$\Gamma : \mathbb{R}\times\mathbb{R}\to\mathbb{R}$ by
\[ %\label{eq:g_function}
    \Gamma(k^{\mu},k'^{\mu}):=\frac{1}{(\omega_{k'}+\omega_{k})^{2}}\frac{1}{4\omega_{k}\omega_{k'}}\chi(\pm\omega_{k},\mp\omega_{k'})\;, \]
and the integration variables~$q=k-k'$ and~$r=\frac{k+k'}{2}$,
\begin{align}\label{eq:integral1}
    &I_{T_{1}T_{2}}=J_{+}+J_{-}=-2\int \frac{d^{3}q}{(2\pi)^{3}}\int \frac{d^{3}r}{(2\pi)^{3}}\hat{\alpha}_{1}(q)\hat{\alpha}_{2}(-q)\Gamma(r,q)\;,
\end{align}
where we used that the Jacobian determinant of the change of integration variables is~$-1$. The explicit expression of~$\Gamma$ in terms of~$k$ and~$k'$, and~$r$ and~$q$ is
\begin{align*}%\label{eq:Gamma_function}
    &\Gamma(k^{\mu},k'^{\mu})=\frac{-\omega_{k'}\omega_{k}+m^{2}-k\cdot k'}{(\omega_{k'}+\omega_{k})^{2}\omega_{k}\omega_{k'}}\implies \Gamma(r,q)=\frac{a-2|r|^{2}-\sqrt{a^{2}-b^{2}}}{2(a+\sqrt{a^{2}-b^{2}})\sqrt{a^{2}-b^{2}}}\;,
\end{align*}
where~$a:=|r|^{2}+\frac{1}{4}|q|^{2}+m^{2}$ and~$b\coloneqq q\cdot r$, so~$\omega_{k}=\sqrt{a+b}$ and~$\omega_{k'}=\sqrt{a-b}$. The spherical symmetry of~$\Gamma(r,q)$ (it is only a function of the norm of~$q$ and~$r$ and of the scalar product of~$q$ and~$r$), simplifies the computation of~$J_{\pm}$: introducing spherical coordinates~$\rho,\rho'\in(0,\infty)$,~$\theta,\theta'\in(0,\pi)$, and~$\varphi,\varphi'\in(0,2\pi)$ we can assume that
\begin{align*}
    q=(0,0,\rho') \hspace{0.5cm}\textrm{and}\hspace{0.5cm}r=(\rho\sin{\theta}\cos{\varphi},\rho\sin{\theta}\sin{\varphi},\rho\cos{\theta})\;.
\end{align*}
By this ansatz for~$q$ and~$r$, three out of the six integrals in~\eqref{eq:integral1} become trivial,
\[ I_{T_{1}T_{2}}=-2\int_{0}^{\infty} \frac{d\rho}{(2\pi)^{4}}\hat{\alpha}_{1}(\rho)\hat{\alpha}_{2}(-\rho)K(\rho)\;, \]
where we introduced the kernel~$K$
\begin{equation}\label{eq:kernel_Q}
    K(\rho)\coloneqq 2\int_{0}^{\infty} d\rho\int_{0}^{\pi} d\theta \rho^{2}\sin{\theta}\Gamma(r,q)\;.
\end{equation}
\end{proof}

\begin{Remark}\label{rem:Bt2_general_multiplication}\em{
    The previous lemma can be easily extended to more general multiplication operators, which will be required for Corollary~\ref{cor:Bt_conf_flat_ugeneral} and Remark~\ref{rem:Bt2_mixed_terms}.
    \bitem
\item[{\rm{(i)}}] If the multiplication operators do not involve gamma matrices (e.g.~$(T_{j}\psi)(x)=\alpha_{j}(x)\psi(x)$ with~$\psi\in\H_\eta^{\varepsilon}$,~$x\in N_t$ and~$\alpha_{j}(x)\in\mathbb{C}$) all of the derivations and the final result of the previous lemma go through with a small sign difference: the function~$\chi(\omega,\omega')$ has to be replaced with
\begin{equation}
    \chi(\omega,\omega')=4(\omega'\omega+m^{2}+k\cdot k')\;.
\end{equation}
\item[{\rm{(ii)}}] If the multiplication operators involve an arbitrary number of gamma matrices the function~$\chi(\omega,\omega')$ has to be replaced again. Let~$(T_{j}\psi)(x)=\alpha_{j}A_{j}(\gamma_\eta)\psi(x)$ with~$\psi\in\H_\eta^{\varepsilon}$,~$x\in N_t$,~$\alpha_{j}(x)\in\mathbb{C}$ and~$A_{j}(\gamma_\eta)\in\mathbb{C}^{4\times4}$ an arbitrary product of gamma matrices. Then~$\chi(\omega,\omega')$ has to be replaced with the less simplified expression
\begin{equation*}
    \chi(\omega,\omega'):=\Tr_{\mathbb{C}^{4}}{[A_{1}(\gamma_\eta)\Hat{F}_\omega(k')A_{2}(\gammaeta)\Hat{F}_{\omega'}(k)}]\;.
\end{equation*}
However, in this case,~$\chi(\omega,-\omega')\neq\chi(-\omega,\omega')$ since the~$\omega$ and~$\omega'$ dependence in previous expression is, in general, not of the form~$\omega\omega'$ anymore (which was the case in the previous lemma, cf.~\eqref{eq:trace_mult_ops}): for example, if~$A_{1}=\textrm{Id}_{\mathbb{C}^{4}}$ and~$A_{2}=\gamma_{\eta0}$ we have terms proportional to~$m\omega+m\omega'$. So, the function~$\Gamma$ is replaced with
\begin{equation*}
    \Gamma_{\pm}(k^{\mu},k'^{\mu}):=\frac{1}{(\omega_{k'}+\omega_{k})^{2}}\frac{1}{4\omega_{k}\omega_{k'}}\chi(\pm\omega_{k},\mp\omega_{k'})\;.
\end{equation*}
and~$J_{+}\neq J_{-}$. Moreover, a priori, it is not clear anymore if the integrand in~\eqref{eq:integral1} presents spherical symmetry. Hence, in this scenario
\begin{align}\label{eq:general_IT1T2}
   &I_{T_{1}T_{2}}=-\int \frac{d^{3}q}{(2\pi)^{3}}\int \frac{d^{3}r}{(2\pi)^{3}}\hat{\alpha}_{1}(q)\hat{\alpha}_{2}(-q)(\Gamma_+{}(r,q)+\Gamma_{-}(r,q))\;.
\end{align}

\eitem
\hfill\QEDrem
    }
\end{Remark}

\begin{Corollary}\label{cor:Bt_conf_flat_uparallel}
    Let~$(N_t)_{t\in\mathbb{R}}$ be the foliation of the conformally flat spacetime~\eqref{eq:metric_conf_flat} given by the level sets of the global time function~$t$. Moreover, assume that the conformal factor satisfies that~$\Omega=1$ outside a compact region~$V\subset M$, that~$m\neq0$ and~$u=\partial_t$. Furthermore, given a hypersurface~$N_t$ with~$N_t\cap V\neq\emptyset$, we assume that~$A_t$ has an absolutely continuous spectrum and that condition~\eqref{eq:small_Delta_A} is satisfied. Then for this hypersurface~$N_t$ it holds that
\[ B_t^{(2)}=2m^{2}\int_{0}^{\infty} \frac{d\rho}{(2\pi)^{4}}\big(\hat{\alpha}_{1}(\rho)\hat{\alpha}_{2}(-\rho)+\hat{\alpha}_{2}(\rho)\hat{\alpha}_{1}(-\rho)\big)K(\rho)\:, \]
    where~$\alpha_{1}=\dot{\Omega}$, ~$\alpha_{2}=(\Omega-1)$ and~$K$ is the kernel~\eqref{eq:kernel_Q}.
\end{Corollary}
\begin{proof}
        In a conformally flat spacetime with~$u=\partial_t$ the symmetrized Hamiltonian is
        \begin{equation*}
        A_t=-i\gamma_{\eta 0}\gamma_\eta^{\mu}\pr_{\mu}-\frac{3}{2}i\frac{\pr_{\mu}(\Omega)}{\Omega}\gamma_{\eta 0}\gamma_\eta^{\mu}+m\gamma_{gt}=\tilde{H}_\eta+(\Omega-1)m\gamma_{\eta 0} \;, 
        \end{equation*}
    where in the last step we used that~$\gamma_{gt}=\Omega\gamma_{\eta 0}$. Hence, we see that 
        \begin{align*} %\label{eq:DeltaA_conf_flat_u_parallel}
        &\Delta A(t)\coloneqq A_t-\tilde{H}_\eta=(\Omega-1)m\gamma_{\eta 0}\;,
    \end{align*}
    which corresponds to a smooth multiplication operator compactly supported in~$V\cap N_t$ (since~$\Omega=1$ in~$N_t\setminus V$). As in Proposition~\ref{prop:Bt_conf_flat}, we introduce the operator $\tilde{Q}(\omega,\omega'):=\Delta AF_\omega(\tilde{H}_\eta)\Delta AF_{\omega'}(\tilde{H}_\eta) : \H_{t,g}\rightarrow C_{0}^{\infty}(N_t,SM)$. Using that~$F_\omega(\tilde{H}_\eta)=\Omega^{-3/2}F_\omega(H_\eta)\Omega^{3/2}$ we obtain
    \begin{align*}
        &\frac{d}{dt}\tr_{\overline{\H_t^{\varepsilon}}}({\tilde{Q}}(\omega,\omega'))=\frac{d}{dt}\tr_{\overline{\H_t^{\varepsilon}}}\big(\Delta AF_\omega(\tilde{H}_\eta)\Delta AF_{\omega'}(\tilde{H}_\eta)\big)\\
        &=\frac{d}{dt}\tr_{\overline{\H_{\eta}^{\varepsilon}}}\big((\Omega^{3/2}\Delta A(t)\Omega^{-3/2})F_\omega(H_\eta)(\Omega^{3/2}\Delta A(t)\Omega^{-3/2} )F_{\omega'}(H_\eta)\big)\\
        &= m^{2}\Big(\tr_{\overline{\H_{\eta}^{\varepsilon}}}(\dot{\Omega}\gamma_{\eta 0}F_\omega(H_\eta)(\Omega-1)\gamma_{\eta 0}F_{\omega'}(H_\eta))+\tr_{\overline{\H_t^{\varepsilon}}}((\Omega-1)\gamma_{\eta 0}F_\omega(H_\eta)\dot{\Omega}\gamma_{\eta 0}F_{\omega'}(H_\eta))\Big)\\
        &=:m^{2}(\tr_{\overline{\H_\eta^{\varepsilon}}}(K_{T_{1}T_{2}})+\tr_{\overline{\H_\eta^{\varepsilon}}}(K_{T_{2}T_{1}}))\;,
    \end{align*}
    where~$\H_\eta^{\varepsilon}=\tilde{U}(\H_t^{\varepsilon})$ (cf. Lemma~\ref{lem:time_independent}), we used that since $\Delta A$ is a multiplication operator, the unitarily transformed operator is simply $\Omega^{3/2}\Delta A(t)\Omega^{-3/2}=(\Omega-1)m\gamma_{\eta 0} : \H^{\varepsilon}_{\eta}\rightarrow \H^{\varepsilon}_{\eta}$ and introduced the operators~$T_{1}:=\dot{\Omega}\gamma_{\eta0}$,~$T_{2}:=(\Omega-1)\gamma_{\eta0}$. Moreover, the operator product $K_{T_{i}T_{j}}$ is defined by
    \begin{align*}
    K_{T_iT_{j}}:=T_iF_\omega(H_\eta)T_{j}F_{\omega'}(H_\eta) : \H^{\varepsilon}_\eta  \rightarrow \H^{\varepsilon}_\eta \hspace{0.3cm}
    \end{align*}
    Then, Proposition~\ref{prop:Bt_conf_flat} and Lemma~\ref{lem:Bt2_multiplication_op} yield the stated result
    \begin{align*}
        B^{(2)}_t&=-\int_{-\infty}^\infty d\omega \int_{-\infty}^\infty d\omega'\:
    \partial_\omega\Big(\eta_{\Lambda}(\omega)\frac{d}{dt}\tr_{\overline{\H_t^{\varepsilon}}}(\tilde{Q}(\omega,\omega'))\Big)\:\frac{g(\omega') - g(\omega)}{\omega'-\omega}\\
    & =-m^{2}\int_{-\infty}^\infty d\omega \int_{-\infty}^\infty d\omega'\:
    \partial_\omega\big(\eta_{\Lambda}(\omega)(\tr_{\overline{\H_\eta^{\varepsilon}}}(K_{T_{1}T_{2}})+\tr_{\overline{\H_\eta^{\varepsilon}}}(K_{T_{2}T_{1}}))\big)\:\frac{g(\omega') - g(\omega)}{\omega'-\omega}\\
    &=2m^{2}\int_{0}^{\infty} \frac{d\rho}{(2\pi)^{4}}\big(\hat{\alpha}_{1}(\rho)\hat{\alpha}_{2}(-\rho)+\hat{\alpha}_{2}(\rho)\hat{\alpha}_{1}(-\rho)\big)K(\rho)\;,
    \end{align*}
    where~$\alpha_{1}=\dot{\Omega}$, ~$\alpha_{2}=(\Omega-1)$ and~$K$ is the kernel~\eqref{eq:kernel_Q}.

\end{proof}

\section{Scenario 2: a general regularizing vector field} \label{sec:ugeneral}

Recall that by Lemma~\ref{lem:At_diagonalizable} the rate of baryogenesis is identically zero in the case that~$m=0$ and~$u=\partial_t$. Hence, in the following corollary we analyze the case in which the regularizing vector field deviates slightly from~$\partial_t$. Note that we start again with an arbitrary subset~$\H^{\varepsilon}_\eta$ of the space of smooth spinors on a Cauchy hypersurface~$N_{t}$ in Minkowski spacetime.

\begin{Lemma}\label{lem:Bt2_diff_ops}
    Let~$L_{j}: \H_\eta^{\varepsilon}\subset C^{\infty}(N_t,SM)\rightarrow C_{0}^{\infty}(N_t,SM)$ with~$j\in\{1,2\}$ denote two smooth and compactly supported first order differential operators. Furthermore,~$T_{1}: \H_\eta^{\varepsilon}\rightarrow C_{0}^{\infty}(N_t,SM)$ is a smooth and compactly supported multiplication operator which at every point~$x\in N_t$ and for every~$\psi\in \H_\eta^{\varepsilon}$ satisfies that~$ (T_{1}\psi)(x)=\alpha_{1}(x)\gamma_{\eta 0}\psi(x)$ with~$\alpha_{1}(x)\in\mathbb{C}$. Then, for any~$\omega,\omega'\in \sigma(H_\eta)$ it holds that
\begin{align}\label{eq:Bt2_general_ops}
I_{AB}:=&\int_{-\infty}^{\infty}d\omega'\int_{-\infty}^{\infty}d\omega\;\partial_\omega\big(\eta_{\Lambda}(\omega)\text{\rm{tr}}_{\overline{\H^{\varepsilon}_\eta}}(K_{AB})\big)\frac{g(\omega') - g(\omega)}{\omega'-\omega}\nonumber\\
=&-\int \frac{d^{3}k}{(2\pi)^{3}}\int \frac{d^{3}k'}{(2\pi)^{3}}(\Gamma^{AB}_{+}(k,k')+\Gamma^{AB}_{-}(k,k'))\;,
\end{align}
where either~$A=L_{1}$ and~$ B=L_{2}$, or~$A=T_{1}$ and~$B=L_{2}$. Moreover,~$K_{AB}=AF_\omega(H_\eta)BF_{\omega'}(H_\eta)$ is assumed to be trace-class,~$g$ is the characteristic function of the interval~$(-\frac{1}{\varepsilon},-m)$ and~$\Gamma^{AB}_{\pm}$ is a smooth function defined by equations~\eqref{eq:gamma_1} and~\eqref{eq:gamma_2}.
\end{Lemma}
\begin{proof}
    The proof of this lemma is very similar to the one of Lemma~\ref{lem:Bt2_multiplication_op}. In the first place, consider the case in which~$A=L_{1}$ and~$B=L_{2}$. At every point~$x\in N_t$ the first order differential operators~$L_{j}$ adopt the form
    \begin{align*}
        &L_{j}(x)=a_{j}(x)\cdot \partial+b_{j}(x) \;,
    \end{align*}
    where~$a_{j}(x)\cdot \partial:=\delta^{\nu\mu} a_{\nu}(x)\partial_{\mu}$ and the coefficients~$a_{j}^{\mu}$ and~$b_{j}$ can be real, complex or matrix valued. As before, we define again the Fourier conjugated operators~$\Hat{L}_{j}=\mathcal{F}L_{j}\mathcal{F}^{-1}$,~$\Hat{F}_\omega(H_\eta)=\mathcal{F}F_\omega(H_\eta)\mathcal{F}^{-1}$ and~$\Hat{Q}_{j,\omega}\coloneqq\Hat{L}_{j}\Hat{F}_\omega(H_\eta)$, where~$j\in\{1,2\}$. For the explicit expression of~$\Hat{F}_\omega(H_\eta)$ and its kernel~$\Hat{F}_\omega(k,k')$ see the proof of Lemma~\ref{lem:Bt2_multiplication_op}; the operators~$\Hat{L}_{j}$ and~$\Hat{Q}_{j,\omega}$ satisfy that
    \begin{align*}
        &(\Hat{L}_{j}\hat{\psi})(p)=\int d^{3}k \Hat{L}_{j}(p,k)\hat{\psi}(k) & &\textrm{with} \hspace{0.4cm} \Hat{L}_{j}(p,k)=ik\cdot\hat{a}_{j}(p-k)+\hat{b}_{j}(p-k) \\ %\label{eq:FourierL2}\\
        &(\Hat{Q}_{j,\omega}\hat{\psi})(p)=\int d^{3}k \Hat{Q}_{j,\omega}(p,k)\hat{\psi}(k) & &\textrm{with} \hspace{0.4cm} \Hat{Q}_{j,\omega}(p,k)=\int d^{3}k' \Hat{L}_{j}(p,k')\Hat{F}_\omega(k',k)\;.\nonumber
    \end{align*}
    Using that~$\Hat{F}_\omega(p,k)=\Hat{F}_\omega(p)\delta(\omega^{2}-\omega_{p}^{2})\delta^{(3)}(k-p)$ (cf.\ expression~\eqref{eq:Fomega_k}) the kernel~$\Hat{Q}_{j,\omega}$ simplifies and the trace of~$\Hat{K}_{L_{1}L_{2}}\coloneqq \Hat{L}_{1}\Hat{F}_\omega(H_\eta)\Hat{L}_{2}\Hat{F}_{\omega'}(H_\eta)$ are
    \begin{align*}
    \Hat{Q}_{j,\omega}(p,k)&=\int d^{3}k'\Hat{L}_{j}(p,k')\Hat{F}_\omega(k')\delta(\omega^{2}-\omega_{k'}^{2})\delta^{(3)}(k-k')\\
        &=\Hat{L}_{j}(p,k)\Hat{F}_\omega(k)\delta(\omega^{2}-\omega_{k}^{2})\\
        \implies\tr_{\overline{\H^{\varepsilon}_\eta}}(\Hat{K}_{L_{1}L_{2}})&=\!\int\!d^{3}k\Tr_{\mathbb{C}^{4}}{(\Hat{K}_{L_{1}L_{2}}(k,k))}\!=\!\int\! \!d^{3}k\!\int \!\!d^{3}k'\Tr_{\mathbb{C}^{4}}{(\Hat{Q}_{1,\omega}(k,k')\Hat{Q}_{2,\omega'}(k',k))}\\
     &=\int d^{3}k\int d^{3}k' \chi(\omega,\omega')\delta(\omega^{2}-\omega_{k}^{2})\delta(\omega^{2}-\omega_{k'}^{2})\;,
    \end{align*}
    where the function~$\chi : \mathbb{R}\times \mathbb{R}\rightarrow \mathbb{R}$ is 
    \begin{align*}%\label{eq:f_differential_op}
        &\chi(\omega,\omega')\!:=\!\Tr_{\mathbb{C}^{4}}\!{\Big[\!\big(ik'\!\cdot\!\hat{a}_{1}(k-k')+\hat{b}_{1}(k-k')\big)\Hat{F}_\omega(k')\big(ik\!\cdot\!\hat{a}_{2}(k'-k)+\hat{b}_{2}(k'-k)\big)\Hat{F}_{\omega'}(k)\!\Big]}
    \end{align*}
    By the discussion in Remark~\ref{rem:Bt2_general_multiplication}~(ii),~$\chi(\omega,-\omega')\neq\chi(-\omega,\omega')$ and using the derivation of Lemma~\ref{lem:Bt2_multiplication_op} it follows that
    \begin{align}\label{eq:gamma_1}
        &J_{\pm}=-\int \frac{d^{3}k}{(2\pi)^{3}}\int \frac{d^{3}k'}{(2\pi)^{3}}\Gamma^{L_{1}L_{2}}_{\pm}(k,k')\nonumber\\
        &\textrm{with}\hspace{0.5cm}\Gamma_{\pm}^{L_{1}L_{2}}(k,k')=\frac{1}{(\omega_{k'}+\omega_{k})^{2}}\frac{1}{4\omega_{k}\omega_{k'}}\chi(\mp\omega_{k},\pm\omega_{k'})\;.
    \end{align}
    In conclusion,
    \begin{align*}
        &I_{L_{1}L_{2}}=-\int \frac{d^{3}k}{(2\pi)^{3}}\int \frac{d^{3}k'}{(2\pi)^{3}}(\Gamma^{L_{1}L_{2}}_{+}(k,k')+\Gamma^{L_{1}L_{2}}_{-}(k,k'))\;.
    \end{align*}
    Consider now that~$A=T_{1}$ and~$B=L_{2}$. Recall that the Fourier conjugated operator~$\Hat{T}_{1}:=\mathcal{F}T_{1}\mathcal{F}^{-1}$ is given by~\eqref{eq:FourierT1}. The trace of~$\Hat{K}_{L_{1}L_{2}}=\Hat{T}_{1}\Hat{F}_\omega(H_\eta)\Hat{L}_{2}\Hat{F}_{\omega'}(H_\eta)$ is
    \begin{align*}
    \tr_{\overline{\H^{\varepsilon}_\eta}}(\Hat{K}_{T_{1}L_{2}})
    &=\int \frac{d^{3}k}{(2\pi)^{3}}\int \frac{d^{3}k'}{(2\pi)^{3}}\hat{\alpha}_{1}(k-k')\chi(\omega,\omega')\;,
\end{align*}
where, in this case, the function~$\chi : \mathbb{R}\times \mathbb{R}\rightarrow\mathbb{R}$ is
\begin{equation*}
    \chi(\omega,\omega')=\Tr_{\mathbb{C}^{4}}{\big[\gamma_{\eta0}\Hat{F}_\omega(k')(k\cdot\hat{a}_{2}(k'-k)+\hat{b}_{2}(k'-k))\Hat{F}_{\omega'}(k)\big]}\;.
\end{equation*}
The desired integral is
\[ I_{T_{1}L_{2}}=-\int \frac{d^{3}k}{(2\pi)^{3}}\int \frac{d^{3}k'}{(2\pi)^{3}}(\Gamma^{T_{1}L_{2}}_{+}(k,k')+\Gamma^{T_{1}L_{2}}_{-}(k,k'))\;, \]
where~$\Gamma^{T_{1}L_{2}}_{\pm}$ is 
\begin{equation}\label{eq:gamma_2}
    \Gamma^{T_{1}L_{2}}_{\pm}(k,k')=\frac{1}{(\omega_{k'}+\omega_{k})^{2}}\frac{1}{4\omega_{k}\omega_{k'}}\hat{\alpha}_{1}(k-k')\chi(\mp\omega_{k},\pm\omega_{k'})\;.
\end{equation}
\end{proof}

\begin{Corollary}\label{cor:Bt_conf_flat_ugeneral}
     Let~$(N_t)_{t\in\mathbb{R}}$ be the foliation of the conformally flat spacetime~\eqref{eq:metric_conf_flat} given by the level sets of the global time function~$t$. Assume that~$u$ evolves according to the locally rigid equation~\eqref{eq:dynamical_eq_u} subject to the initial condition
        \begin{align}\label{eq:u_general}
    & u_{p}=
    \begin{cases}
    (1+\lambda f_{p})\nu+\lambda X_{p}  &\textrm{for}\hspace{0.5cm} p\in V\\
    \partial_t &\textrm{for}\hspace{0.5cm} p\in N_{t_{0}} \setminus V \;,
    \end{cases}
\end{align}
where~$f\in C^{\infty}(\mathbb{R}^{3},\mathbb{R}_{>0})$ is a positive and smooth function,~$X$ is a spacelike vector field~$X$ and~$V\subset N_{t_{0}}$ is compact. 
\\Moreover, assume that~$m=0$,~$A_t$ has an absolutely continuous spectrum and that for all~$\omega\in\rho(A_{t})$ it holds that
\[ %\label{eq:assumption99}
        \| R_\omega(\tilde{H}_\eta)\Delta A(t))\|<1\;. \]
Then, the rate of baryogenesis is 
\begin{align*}
B_t^{(2)}=\lambda^{2}\big[I_{T_{2}T_{1}}+I_{T_{2}L_{1}}+I_{L_{2}T_{1}}+I_{L_{2}L_{1}}+I_{T_{1}T_{2}}+I_{T_{1}L_{2}}+I_{L_{1}T_{2}}+I_{L_{1}L_{2}}\big]\;,
\end{align*}
where~$I_{AB}$ (with $A,B\in\{L_{1},L_{2},T_{1},T_{2}\}$) stand for the double integrals \eqref{eq:general_IT1T2} and ~\eqref{eq:Bt2_general_ops}, with ~$L_{1}$ and~$L_{2}$ first order differential operators (see expression ~\eqref{eq:L1L2})~and $T_{1}$ and~$T_{2}$ multiplication operators (see expression ~\eqref{eq:T1T2}).
\end{Corollary}
\begin{proof}
    Lemma~\ref{lem:coefs_conformally_flat} (with~$m=0$) with the ansatz~\eqref{eq:u_general} for~$u$ gives
    \begin{align}\label{eq:DeltaAt_ugeneral}
        &A_t
    =\tilde{H}_\eta+\frac{\lambda}{2}\Big[\big\{ f,\tilde{H}_\eta\big\}+i\Big\{X^{\mu},\partial_{\mu}+\frac{\partial_{\mu}(\sqrt{|\text{\rm{det}}(g|_{N_t})|})}{2\sqrt{|\text{\rm{det}}(g|_{N_t})|}}+\frac{1}{4}\frac{\partial_{\nu}(\Omega)}{\Omega}\eta^{\nu\rho}[\gamma_{\eta\mu},\gamma_{\eta \rho}]\Big\}\Big]\nonumber\\
        &\implies\Delta A(t)\coloneqq A_t-\tilde{H}_\eta=\lambda\Big(\frac{1}{2}\big\{ f,\tilde{H}_\eta\big\}+\frac{i}{2}\big\{X^{\mu},K_{\mu}\big\}\Big)\;,
    \end{align}
    where we introduced the first order antisymmetric differential operator 
    \begin{align*}
        K_{\mu}\coloneqq \partial_{\mu}+\frac{\partial_{\mu}(\sqrt{|\text{\rm{det}}(g|_{N_t})|})}{2\sqrt{|\text{\rm{det}}(g|_{N_t})|}}+\frac{1}{4}\frac{\partial_{\nu}(\Omega)}{\Omega}\eta^{\nu\rho}[\gamma_{\eta\mu},\gamma_{\eta \rho}]\;.
    \end{align*}
    In order to determine the differential operators~$\Delta A$ and~$\frac{d}{dt}(\Delta A)$ explicitly we rewrite~$\tilde{H}_\eta$ and~$K_{\mu}$ as follows
    \begin{align*}
        &\tilde{H}_\eta=\gamma_{\eta 0}\gamma_\eta^{\mu}(a\partial_{\mu}+b_{\mu}) & &\textrm{with}\hspace{0.3cm}a:=-i,\: b_{\mu}:=-\frac{3}{2}i\frac{\pr_{\mu}(\Omega)}{\Omega}\\
        &K_{\mu}=\partial_{\mu}+c_{\mu}+d^{\rho}[\gamma_{\eta\mu},\gamma_{\eta\rho}]& &\textrm{with}\hspace{0.3cm}c_{\mu}:=ib_{\mu}+\frac{\partial_{\mu}(\sqrt{|\text{\rm{det}}(\eta)|})}{2\sqrt{|\text{\rm{det}}(\eta)|}}, d^{\rho}:=\frac{1}{4}\frac{\partial_{\nu}(\Omega)}{\Omega}\eta^{\nu\rho}\;.
    \end{align*}
    An explicit formula for~$\Delta A$ is obtained by combining the previous ones with~\eqref{eq:DeltaAt_ugeneral}
    \begin{equation}\label{eq:DeltaAt_ugeneral_explicit}
        \Delta A= \lambda\big((\alpha_{1}\gamma_{\eta 0}\gamma_\eta^{\mu}+\alpha_{2}^{\mu})\partial_{\mu}+\alpha_{3\mu}\gamma_{\eta 0}\gamma_\eta^{\mu}+\alpha_{4}^{\mu\rho}[\gamma_{\eta\mu},\gamma_{\eta\rho}]+\alpha_{5}\big)\;,
    \end{equation}
    where the coefficients are
    \begin{align*}
        &\alpha_{1}=af\;, & &\alpha_{2}^{\mu}=iX^{\mu}\;, & &\alpha_{3\mu}=fb_{\mu}+\frac{a}{2}\partial_{\mu}(f)\;,\nonumber\\
        &\alpha_{4}^{\mu\rho}=iX^{\mu}d^{\rho}\;, & &\alpha_{5}=iX^{\mu}c_{\mu}+\frac{i}{2}\partial_{\mu}(X^{\mu})\;.
    \end{align*}
    Moreover, the coefficients~$a,b_{\mu},c_{\mu}$ and~$d^{\rho}$ satisfy that~$\dot{a}=0$,~$\dot{c}_{\mu}=i\dot{b}_{\mu}$ and
    \begin{align*}
        & \dot{b}_{\mu}=-\frac{3}{2}i\frac{1}{\Omega}\Big(\partial_{\mu}\dot{\Omega}-\frac{\dot{\Omega}\partial_{\mu}\Omega}{\Omega}\Big);\hspace{0.5cm}\dot{d}^{\rho}=\frac{1}{4}\frac{\eta^{\nu\rho}}{\Omega}\Big(\partial_{\nu}\dot{\Omega}-\frac{\dot{\Omega}\partial_{\nu}\Omega}{\Omega}\Big).\hspace{0.5cm}
    \end{align*}
    Note that from the dynamical equation of~$u$~\eqref{eq:dynamical_eq_u} we have that to linear order in~$\lambda$
    \begin{align*}
        \lambda\frac{df}{dt}=\frac{d(1+\lambda f)}{dt}=\frac{du^{t}}{dt}=\frac{\lambda}{3}\textrm{div}_{\delta}(X) \hspace{0.5cm}\textrm{and}\hspace{0.5cm} \lambda\frac{dX^{\mu}}{dt}=\frac{du^{\mu}}{dt}=-\lambda\delta^{\mu\nu}\partial_\nu{(f^{-1})}\;.
    \end{align*}
Differentiating~\eqref{eq:DeltaAt_ugeneral} with respect to time yields
\[ %\label{eq:derivative_DeltaAt}
    \frac{d}{dt}\Delta A=\lambda\big((\beta_{1}\gamma_{\eta 0}\gamma_\eta^{\mu}+\beta_{2}^{\mu})\partial_{\mu}+\beta_{3\mu}\gamma_{\eta 0}\gamma_\eta^{\mu}+\beta_{4}^{\mu\rho}[\gamma_{\eta\mu},\gamma_{\eta\rho}]+\beta_{5}\big)\;, \]
where the coefficients
\begin{align*}
    &\beta_{1}=\frac{1}{3}\textrm{div}_{\delta}(X)a & &\beta_{2}^{\mu}=-i\delta^{\mu\nu}\partial_\nu{(f^{-1})}\nonumber\\
    &\beta_{3\mu}=\frac{1}{3}\textrm{div}_{\delta}(X)b_{\mu}+f\dot{b}_{\mu}+\frac{a}{6}\partial_{\mu}(\textrm{div}_{\delta}(X)) & &\beta_{4}^{\mu\rho}=-id^{\rho}\delta^{\mu\nu}\partial_{\nu}(f^{-1})+iX^{\mu}\dot{d}^{\rho}\nonumber\\
    &\beta_{5}=-ic_{\mu}\delta^{\mu\nu}\partial_\nu{(f^{-1})}-\frac{i}{2}\delta^{\mu\nu}\partial_{\mu}\partial_\nu{(f^{-1})}+iX^{\mu}\dot{c}_{\mu}\;,
\end{align*}
are real and/or complex valued. We introduce now the differential operators~$L_{j}$ by setting
\begin{align}\label{eq:L1L2}
    \Delta A=\lambda L_{1} \hspace{0.5cm} \textrm{and}\hspace{0.5cm}\frac{d}{dt}\Delta A=\lambda L_{2}\;.    
\end{align}
The coefficients of the differential operators~$\Delta A$ and~$\frac{d}{dt}\Delta A$ are related to the coefficients~$a_{j}^{\mu}$ and~$b_{j}$  of the differential operator~$L_{j}$ appearing in Lemma~\ref{lem:Bt2_diff_ops} by the following identifications
\begin{align*}
        &a_{1}^{\mu}=\alpha_{1}\gamma_{\eta 0}\gamma_\eta^{\mu}+\alpha_{2}^{\mu}; & &b_{1}=\alpha_{3\mu}\gamma_{\eta 0}\gamma_\eta^{\mu}+\alpha_{4}^{\mu\rho}[\gamma_{\eta\mu},\gamma_{\eta\rho}]+\alpha_{5} \\ %\label{eq:L1_coeffs}\\
        &a_{2}^{\mu}=\beta_{1}\gamma_{\eta 0}\gamma_\eta^{\mu}+\beta_{2}^{\mu}; & & b_{2}=\beta_{3\mu}\gamma_{\eta 0}\gamma_\eta^{\mu}+\beta_{4}^{\mu\rho}[\gamma_{\eta\mu},\gamma_{\eta\rho}]+\beta_{5} %\label{eq:L2_coeffs}
    \end{align*}
Let~$\psi\in\H_{t,g}$. Then, the operator product~$\Omega^{3/2}\Delta A \Omega^{-3/2}$ can be rewritten as the sum of a differential and a multiplication operator, i.e.\
\begin{align*}
    \Omega^{3/2}\Delta A \Omega^{-3/2}=\Big[-\frac{\lambda}{2}\frac{3}{2}\frac{\partial_{\mu}(\Omega)}{\Omega}(\alpha_{1}\gamma_{\eta 0}\gamma_\eta^{\mu}+\alpha_{2}^{\mu})+\Delta A\Big]\psi = \lambda[T_{1}+L_{1}]\psi\;,
\end{align*}
where the multiplication operators~$T_{1}$ and~$T_{2}$ are
\begin{align}\label{eq:T1T2}
    T_{1}:=\frac{3}{4}\frac{\partial_{\mu}(\Omega)}{\Omega}(\alpha_{1}\gamma_{\eta 0}\gamma_\eta^{\mu}+\alpha_{2}^{\mu})\hspace{0.5cm}\textrm{and}\hspace{0.5cm}T_{2}:=\frac{d}{dt}(T_{1})\;.
\end{align}
Then, differentiating with respect to time the trace of the product operator~$\tilde{Q}(\omega,\omega'):=\Delta AF_\omega(\tilde{H}_\eta)\Delta AF_{\omega'}(\tilde{H}_\eta)$ gives
    \begin{align*}
        &\frac{d}{dt}\tr_{\overline{\H_t^{\varepsilon}}}(\tilde{Q}(\omega,\omega'))=\frac{d}{dt}\tr_{\overline{\H_t^{\varepsilon}}}\big(\Delta AF_\omega(\tilde{H}_\eta)\Delta AF_{\omega'}(\tilde{H}_\eta)\big)\\
        &=\frac{d}{dt}\tr_{\overline{\H_{\eta}^{\varepsilon}}}\big((\Omega^{3/2}\Delta A(t)\Omega^{-3/2} )F_\omega(H_\eta)(\Omega^{3/2}\Delta A(t)\Omega^{-3/2})F_{\omega'}(H_\eta)\big)\\
        &=\lambda^{2}\frac{d}{dt}\tr_{\overline{\H_{\eta}^{\varepsilon}}}([T_{1}+L_{1}]F_\omega(H_\eta)[T_{1}+L_{1}]F_{\omega'}(H_\eta))\\
        &=\lambda^{2}\big(\tr_{\overline{\H_{\eta}^{\varepsilon}}}([T_{2}+L_{2}]F_\omega(H_\eta)[T_{1}+L_{1}]F_{\omega'}(H_\eta))\\
        &\quad\quad\quad+\tr_{\overline{\H_t^{\varepsilon}}}([T_{1}+L_{1}]F_\omega(H_\eta)[T_{2}+L_{2}]F_{\omega'}(H_\eta))\big)
    \end{align*}
where~$\H_\eta^{\varepsilon}=\tilde{U}(\H_t^{\varepsilon})$ (cf. Lemma~\ref{lem:time_independent}). By Proposition~\ref{prop:Bt_conf_flat} it directly follows that~$B_t^{(2)}$ can be split in eight different integrals
\begin{align*}
    B_t^{(2)}&=-\int_{-\infty}^\infty d\omega \int_{-\infty}^\infty d\omega'\:
    \partial_\omega\Big(\eta_{\Lambda}(\omega)\frac{d}{dt}\tr_{\overline{\H_t^{\varepsilon}}}(\tilde{Q}(\omega,\omega'))\Big)\:\frac{g(\omega') - g(\omega)}{\omega'-\omega}\\
    &=\lambda^{2}\big[I_{T_{2}T_{1}}+I_{T_{2}L_{1}}+I_{L_{2}T_{1}}+I_{L_{2}L_{1}}+I_{T_{1}T_{2}}+I_{T_{1}L_{2}}+I_{L_{1}T_{2}}+I_{L_{1}L_{2}}\big]\;.
\end{align*}
Each one of the eight integrals have been determined in Lemma~\ref{lem:Bt2_multiplication_op} (for its extended version to more general multiplication operators see Remark~\ref{rem:Bt2_general_multiplication}) and Lemma~\ref{lem:Bt2_diff_ops}.
\end{proof}
\begin{Remark}\label{rem:Bt2_mixed_terms}\em{
The previous lemma can be easily extended to the case with~$m\neq0$ if we assume again that~$\Omega=1$ outside a compact subset~$V\subset M$. In this scenario,~$\Delta A$ corresponds to the sum of a multiplication operator~$T_{3}$ and a differential operator~$L_{1}$
\begin{align*}
    \Delta A(t)&\coloneqq A_t-\tilde{H}_\eta=(\Omega-1)m\gamma_{\eta 0}+\frac{\lambda}{4}\Big[\{f_{p},H_{g}+H_{g}^{\ast}\}+ i\{X_{p}^{\mu},\nabla^{s}_{\mu}-(\nabla^{s}_{\mu})^{\ast}\}\Big]\\
    &= m T_{3}+ \lambda L_{1}\;,
\end{align*}
where~$T_{3}:=(\Omega-1)\gamma_{\eta 0}$. So,
\begin{align*}
    &\Omega^{-3/2}\Delta A \Omega^{3/2}=m T_{3}+\lambda (T_{1}+L_{1})\\
    &\frac{d}{dt}(\Omega^{-3/2}\Delta A \Omega^{3/2})=m T_{4}+\lambda (T_{2}+L_{2})\;,
\end{align*}
where~$L_{1}$ and~$L_{2}$ correspond to the differential operators appearing in~\eqref{eq:L1L2},~$T_{1}$ and~$T_{2}$ are given by~\eqref{eq:T1T2} and~$T_{4}:=\frac{d}{dt}(T_{3})$. Moreover, differentiating the trace of the operator~$\tilde{Q}(\omega,\omega')$ yields
\begin{align*}
    &\frac{d}{dt}\tr_{\overline{\H_t^{\varepsilon}}}(\tilde{Q}(\omega,\omega'))=\frac{d}{dt}\tr_{\overline{\H_{\eta}^{\varepsilon}}}\big((\Omega^{3/2}\Delta A(t)\Omega^{-3/2} )F_\omega(H_\eta)(\Omega^{3/2}\Delta A(t)\Omega^{-3/2})F_{\omega'}(H_\eta)\big)\\
    &=\frac{d}{dt}\tr_{\overline{\H_{\eta}^{\varepsilon}}}\big([m T_{3}+\lambda (T_{1}+L_{1})]F_\omega(H_\eta)[m T_{3}+\lambda (T_{1}+L_{1})]F_{\omega'}(H_\eta)\big)\\
    &=\tr_{\overline{\H_{\eta}^{\varepsilon}}}([m T_{4}+\lambda (T_{2}+L_{2})]F_\omega(H_\eta)[m T_{3}+\lambda (T_{1}+L_{1})]F_{\omega'}(H_\eta))\\
    &+\tr_{\overline{\H_{\eta}^{\varepsilon}}}([m T_{3}+\lambda (T_{1}+L_{1})]F_\omega(H_\eta)[m T_{4}+\lambda (T_{2}+L_{2})]F_{\omega'}(H_\eta))\\
    &=m^{2}\tr_{\overline{\H_{\eta}^{\varepsilon}}}\Big(\frac{d}{dt}\tilde{Q}(\omega,\omega')\Big)\Big|_{m\neq0,\lambda=0}+\lambda^{2}\tr_{\overline{\H_{\eta}^{\varepsilon}}}\Big(\frac{d}{dt}\tilde{Q}(\omega,\omega')\Big)\Big|_{m=0,\lambda\neq0}\\
    &+m\lambda\Big(\tr_{\overline{\H_{\eta}^{\varepsilon}}}\big[T_{4}F_\omega(H_\eta)(T_{1}+L_{1})F_{\omega'}(H_\eta)\big]+\tr_{\overline{\H_{\eta}^{\varepsilon}}}\big[(T_{2}+L_{2})F_\omega(H_\eta)T_{3}F_{\omega'}(H_\eta)\big]\\
    &+\tr_{\overline{\H_{\eta}^{\varepsilon}}}\big[T_{3}F_\omega(H_\eta)(T_{2}+L_{2})F_{\omega'}(H_\eta)\big]+\tr_{\overline{\H_{\eta}^{\varepsilon}}}\big[(T_{1}+L_{1})F_\omega(H_\eta)T_{4}F_{\omega'}(H_\eta)\big]\Big)\;.
\end{align*}
Using again Lemma~\ref{lem:time_independent} to evaluate the trace on the space~$\H^{\varepsilon}_\eta$, the second order contribution to the rate of baryogenesis is
\begin{align*}
B_t^{(2)}&=B_t^{(2)}\big|_{m\neq0,\lambda=0}+B_t^{(2)}\big|_{m=0,\lambda\neq0}\nonumber\\
    &+m\lambda \big[I_{T_{4}T_{1}}+I_{T_{4}L_{1}}+I_{T_{2}T_{3}}+I_{L_{2}T_{3}}+I_{T_{3}T_{2}}+I_{T_{3}L_{2}}+I_{T_{1}T_{4}}+I_{L_{1}T_{4}}\big]\;,
\end{align*}
where~$B_t^{(2)}\big|_{m\neq0,\lambda=0}$ and~$B_t^{(2)}\big|_{m=0,\lambda\neq0}$ correspond, respectively, to the rate of baryogenesis computed in Corollaries~\ref{cor:Bt_conf_flat_uparallel} and~\ref{cor:Bt_conf_flat_ugeneral}. Furthermore, the eight integrals~$I_{AB}$ (where~$A$ and~$B$ are multiplication or first order differential operators) are of the form of those computed in Lemmas~\ref{lem:Bt2_multiplication_op} and~\ref{lem:Bt2_diff_ops}.}
\hfill\QEDrem
\end{Remark}

\section{Discussion} \label{secdiscuss}

The first contribution of this paper is to study the main analytic and geometric features of the baryogenesis mechanism presented in \cite{baryogenesis} in general conformally flat spacetimes. We emphasize that Proposition~\ref{prop:chiAt} sheds light onto the interpretation of the locally rigid spinor dynamics already introduced in \cite{baryogenesis} as a realization, through adiabatic projections, of a spinor dynamics which deviates slightly from Dirac. Proposition~\ref{prop:chiAt} implies that, under suitable assumptions, the spectral projection operator~$\chi_{I}(A_{t})$ is a \emph{regularization operator} (in the sense that it maps into the smooth spinor fields on~$N_{t}$). So, in the locally rigid spinor dynamics the spinors evolve following the Dirac dynamics and, after arbitrarily small time-steps, they are regularized. As a consequence, by construction, the spinor dynamics depends crucially on the dynamics of~$\chi_{I}(A_{t})$ and thus, by the definition of~$A_{t}$, on the dynamics of the vector field~$u$. By Proposition~\ref{prop:u_conformal_invariant}, the locally rigid dynamics of the regularizing vector field is a conformal invariant, so the dynamical equation derived in \cite[Lemma 7.1]{baryomink} also applies to general conformally flat spacetimes. 

Moreover, the main result of this paper is Theorem~\ref{theo:baryo_rate}, which gives a concise formula for the leading order contribution to the rate of baryogenesis in conformally flat spacetimes depending on the value of the conformal factor~$\Omega$, the mass~$m$ and the regularizing vector field~$u$. A first implication is that, in our setting, a process of baryogenesis is only triggered if the mass is non-zero and/or if the regularizing vector field deviates from~$\partial_{t}$ (with~$t$ the time coordinate, which we also use to construct the foliation). Moreover, another interesting consequence is that when~$u=\partial_{t}$ the rate of baryogenesis vanishes identically in Minkowski spacetime, but, however, it is non-zero (yet very small; it scales as~$m^{2}$) for general conformally flat spacetimes.

Our derived formula for the rate of baryogenesis paves the way for concrete, quantitative predictions for cosmological spacetimes. Spacetimes modeling the early Universe are, arguably, those of most interest, since it is then that cosmologists believe that the matter/antimatter asymmetry originated. A particular example of an interesting conformally flat spacetime which is anisotropic is the one presented in \cite{LingPiubello}. Hence, working out concrete predictions for spacetimes describing the early Universe will be the next natural step of our analysis. Finally, we would like to address the question of whether the predictions of our baryogenesis mechanism 
match the observed matter/antimatter asymmetry.

\Thanks{{{\em{Acknowledgments:}} We would like to thank Eric Ling, Claudio F.\ Paganini and Gabriel Schmid for helpful discussions. The second author gratefully acknowledges support by the Studienstiftung des deutschen Volkes.

\appendix
\section{Computation of the symmetrized Hamiltonian}\label{sec:At_twisted_prod}

In the following lemma we compute the spin connection, Dirac operator, Dirac Hamiltonian and symmetrized Hamiltonian for the general conformally flat spacetime with metric~\eqref{eq:metric_conf_flat}.
\begin{Lemma}\label{lem:coefs_conformally_flat}
    In a conformally flat spacetime with metric~\eqref{eq:metric_conf_flat} the spin coefficients, the Dirac operator, the Dirac Hamiltonian and its corresponding adjoint operators are
    \begin{align*}
        &\nabla^{s}_{r}=\pr_{r}+\frac{1}{2}\gamma_{gr}\Big(\frac{\dot{\Omega}}{\Omega}\gamma_{g}^{t}+\frac{\partial_{\theta}(\Omega)}{\Omega}\gamma_{g}^{\theta}+\frac{\partial_{\varphi}(\Omega)}{\Omega}\gamma_{g}^{\varphi}\Big)\\
        &\nabla^{s}_{\theta}=\pr_{\theta}+\frac{1}{2}\gamma_{g\theta}\Big(\frac{\dot{\Omega}}{\Omega}\gamma_{g}^{t}+\frac{\Omega'}{\Omega}\gamma_{g}^{r}+\frac{\partial_{\varphi}(\Omega)}{\Omega}\gamma_{g}^{\varphi}\Big)\\
        &\nabla^{s}_{\varphi}=\pr_{\varphi}+\frac{1}{2}\gamma_{g\varphi}\Big(\frac{\dot{\Omega}}{\Omega}\gamma_{g}^{t}+\frac{\Omega'}{\Omega}\gamma_{g}^{r}+\frac{\partial_{\theta}(\Omega)}{\Omega}\gamma_{g}^{\theta}\Big)\\
        &D_{g}=i\gammag^{j}\pr_{j}+\frac{3}{2}i\gammag^{j}\frac{\partial_{j}(\Omega)}{\Omega}\\ 
        &H_{g}=\tilde{H}_\eta+(\Omega-1) m\gamma_{\eta 0}-\frac{3}{2}i\frac{\dot{\Omega}}{\Omega}\\
        &(\nabla^{s}_{r})^{\ast}=-\partial_{r}-\frac{2}{r}-3\frac{\Omega'}{\Omega}+\frac{1}{2}\gamma_{gr}\Big(\frac{\dot{\Omega}}{\Omega}\gamma_{g}^{t}-\frac{\partial_{\theta}(\Omega)}{\Omega}\gamma_{g}^{\theta}-\frac{\partial_{\varphi}(\Omega)}{\Omega}\gamma_{g}^{\varphi}\Big)\\
        &(\nabla^{s}_{\theta})^{\ast}=-\partial_{\theta}-\frac{\cos{\theta}}{\sin{\theta}}-3\frac{\partial_{\theta}(\Omega)}{\Omega}+\frac{1}{2}\gamma_{g\theta}\Big(\frac{\dot{\Omega}}{\Omega}\gamma_{g}^{t}-\frac{\Omega'}{\Omega}\gamma_{g}^{r}-\frac{\partial_{\varphi}(\Omega)}{\Omega}\gamma_{g}^{\varphi}\Big)\\
        &(\nabla^{s}_{\varphi})^{\ast}=-\partial_{\varphi}-3\frac{\partial_{\varphi}(\Omega)}{\Omega}+\frac{1}{2}\gamma_{g\varphi}\Big(\frac{\dot{\Omega}}{\Omega}\gamma_{g}^{t}-\frac{\Omega'}{\Omega}\gamma_{g}^{r}-\frac{\partial_{\theta}(\Omega)}{\Omega}\gamma_{g}^{\theta}\Big)\\
        &H_{g}^{\ast}= \tilde{H}_\eta+(\Omega-1) m\gamma_{\eta 0}+\frac{3}{2}i\frac{\dot{\Omega}}{\Omega}\;.
    \end{align*}
    Furthermore, let~$u : M\rightarrow TM$ be a global future-directed timelike vector field. Then, 
    % the symmetrized Hamiltonian~$A_t$ is
    \begin{align*}
        A_t=&\frac{1}{2}\{u^{t},\tilde{H}_\eta+(\Omega-1) m\gamma_{\eta 0}\}+\frac{i}{2}\Big\{u^{\mu},\partial_{\mu}+\frac{\partial_{\mu}(\sqrt{|\text{\rm{det}}(g|_{N_t})|})}{2\sqrt{|\text{\rm{det}}(g|_{N_t})|}}+\frac{1}{4}\frac{\partial_{\nu}(\Omega)}{\Omega}\eta^{\nu\rho}[\gamma_{\eta\mu},\gamma_{\eta\rho}]\Big\}
    \end{align*}
\end{Lemma}
\begin{proof}
In the first place, let~$\gamma_\eta$ denote Clifford multiplication in the Minkowski spacetime in spherical coordinates (see~\cite[eq.~(2.1)]{moritz}):
    \begin{align*}%\label{eq:Minkowski_gammma_matrices}
        &\gamma_\eta^{0}=\gamma^{0} \\
        &\gammaeta^{r}=\cos{\theta} \gammaeta^{3}+\sin{\theta}\cos{\varphi}\gamma^{1}+\sin{\theta}\sin{\varphi}\gamma^{2} \\
        &\gammaeta^{\theta}=\frac{1}{r}(-\sin{\theta}\gamma^{3}+\cos{\theta}\cos{\varphi}\gamma^{1}+\cos{\theta}\sin{\varphi}\gamma^{2}) \\
        &\gammaeta^{\varphi}=\frac{1}{r\sin{\theta}}(-\sin{\varphi}\gamma^{1}+\cos{\varphi}\gamma^{2}) 
    \end{align*}
    where~$\gamma^{0},...,\gamma^{3}$ denote Clifford multiplication in~$(\mathbb{R}^{1,3},\eta)$ in Cartesian coordinates and in the Dirac representation. Note~$\gamma^{0},...,\gamma^{3}$ satisfy that with respect to the usual inner product~$\langle\cdot|\cdot\rangle_{\mathbb{C}^{4}}$ on~$\mathbb{C}^{4}$,~$(\gamma^{0})^{\ast}=\gamma^{0}$ and~$(\gamma^{\mu})^{\ast}=-\gamma^{\mu}$ for~$\mu\in\{1,2,3\}$. In the conformally flat spacetime~$(M,g)$ Clifford multiplication is
\begin{align*}
    & \gammag^{j}=\frac{1}{\Omega}\gammaeta^{j}\;.
\end{align*}
  
  Moreover, the non-vanishing Christoffel symbols are:
    \begin{align*}
        &\Gamma^{t}_{tt}=\Gamma^{t}_{rr}=\Gamma^{r}_{tr}=\Gamma^{\theta}_{t\theta}=\Gamma^{\varphi}_{t\varphi}=\frac{\dot{\Omega}}{\Omega}& &\Gamma^{t}_{\theta\theta}=r^{2}\frac{\dot{\Omega}}{\Omega}\\
        &\Gamma^{t}_{\varphi\varphi}=r^{2}\sin^{2}{\theta}\frac{\dot{\Omega}}{\Omega}& &\Gamma^{t}_{rt}=\Gamma^{r}_{tt}=\Gamma^{r}_{rr}=\frac{\Omega'}{\Omega} \\
        &\Gamma^{r}_{\theta\theta}=-r-r^{2}\frac{\Omega'}{\Omega} & &\Gamma^{r}_{\varphi\varphi}=-\sin^{2}{\theta}\Big(r+r^{2}\frac{\Omega'}{\Omega}\Big)\\
        &\Gamma^{\theta}_{\varphi\varphi}=-\sin{\theta}\cos{\theta}-\sin^{2}{\theta}\frac{\partial_{\theta}(\Omega)}{\Omega} & &\Gamma^{\theta}_{r\theta}=\Gamma^{\varphi}_{r\varphi}=\frac{1}{r}+\frac{\Omega'}{\Omega}\\
        &\Gamma^{\varphi}_{\theta\varphi}=\frac{\cos{\theta}}{\sin{\theta}} +\frac{\partial_{\theta}(\Omega)}{\Omega} & &\Gamma^{t}_{t\theta}=\Gamma^{r}_{r\theta}=\Gamma^{\theta}_{\theta\theta}=\frac{\partial_{\theta}(\Omega)}{\Omega}\\
        &\Gamma^{t}_{t\varphi}=\Gamma^{r}_{r\varphi}=\Gamma^{\theta}_{\theta\varphi}=\Gamma^{\varphi}_{\varphi\varphi}=\frac{\partial_{\varphi}(\Omega)}{\Omega} & &\Gamma^{\theta}_{tt}=\frac{\partial_{\theta}(\Omega)}{r^{2}\Omega}\\
        &\Gamma^{\theta}_{rr}=-\frac{\partial_{\theta}(\Omega)}{r^{2}\Omega}& &\Gamma^{\varphi}_{\theta\theta}=-\frac{1}{\sin^{2}{\theta}}\frac{\partial_{\varphi}(\Omega)}{\Omega}\\
        &\Gamma^{\varphi}_{tt}=\frac{1}{r^{2}\sin^{2}{\theta}}\frac{\partial_{\varphi}(\Omega)}{\Omega}& &\Gamma^{\varphi}_{rr}=-\frac{1}{r^{2}\sin^{2}{\theta}}\frac{\partial_{\varphi}(\Omega)}{\Omega}\\
    \end{align*}
    In the second place, we now proceed to compute the coefficients of the spin connection. Given an arbitrary spacetime~$(M,g)$ with local frame~$\{x^{j}\}_{j=0,1,2,3}$, the coefficients of its spin connection along the coordinate vector field~$\pr_{j}$ are (see~\cite[Appendix A]{intro} for an explicit derivation):
    \begin{equation}\label{eq:spin_coefs}
        \nabla^{s}_{j}=\pr_{j}+\frac{1}{2}\rho\pr_{j}(\rho)-\frac{1}{16}\textrm{tr}(\gammag^{m}\nabla_{j}\gammag^{n})\gamma_{gm}\gamma_{gn}+\frac{1}{8}\textrm{tr}(\rho\gamma_{gj}\nabla_{n}\gammag^{n})\rho
    \end{equation}
    where~$\rho\coloneqq\frac{i}{4!}\sqrt{|g|}\epsilon_{jklm}\gammag^{j}\gammag^{k}\gammag^{l}\gammag^{m}$ and~$\nabla_{j}$ is the Levi-Civita connection along~$\pr_{j}$. A computation similar to the one of the proof of~\cite[Lemma 2.1]{moritz} shows that the first and the third term in equation~\eqref{eq:spin_coefs} vanish so that the expression for the spin coefficients simplifies to
    \begin{equation}\label{eq:conf_flat_spin_coefs}
        \nabla^{s}_{j}=\pr_{j}-\frac{1}{16}\textrm{tr}(\gammag^{m}\nabla_{j}\gammag^{n})\gamma_{gm}\gamma_{gn}
    \end{equation}
    We introduce the coefficients~$h_{jk}^{n}$ which encode the partial derivatives of the gamma matrices, i.e.\ 
    \begin{equation*}
        \pr_{j}\gammag^{n}=-\frac{\pr_{j}(\Omega)}{\Omega}\gammag^{n}+\frac{1}{\Omega}\pr_{j}\gammaeta^{n}=:h^{n}_{jk}\gammag^{k}\;,
    \end{equation*}
    and such that expression~\eqref{eq:conf_flat_spin_coefs} can be rewritten as follows:
    \begin{align}\label{eq:spin_coefs_simpler}
    \nabla^{s}_{j}&=\pr_{j}-\frac{1}{16}\Big(\textrm{tr}(\gammag^{m}\pr_{j}\gammag^{n})\gamma_{gm}\gamma_{gn}+4\Gamma^{n}_{jk}\gamma_{g}^{k}\gamma_{gn}\Big)\nonumber\\
    &=\pr_{j}-\frac{1}{4}\Big(h^{n}_{jk}\gammag^ {k}\gamma_{gn}+\Gamma^{n}_{jk}\gamma_{g}^{k}\gamma_{gn}\Big)
    \end{align}
    where we used that~$\nabla_{j}\gammag^{n}=\pr_{j}\gammag^{n}+\Gamma^{n}_{jk}\gammag^{k}$ and~$\textrm{tr}(\gammag^{m}\gammag^{k})=4g^{mk}$. The only non-zero coefficients~$h^{n}_{jk}$ are:
\begin{align*}%\label{eq:coefs_gamma_derivatives_twisted_prod}
    &h^{t}_{tt}=h^{r}_{tr}=h^{\theta}_{t\theta}=h^{\varphi}_{t\varphi}=-\frac{\dot{\Omega}}{\Omega}\nonumber & &h^{t}_{rt}=h^{r}_{rr}=-\frac{\Omega'}{\Omega}\\ 
    &h^{\theta}_{r\theta}=h^{\varphi}_{r\varphi}=-\frac{\Omega'}{\Omega}-\frac{1}{r} & &h^{\varphi}_{\varphi r}=h^{\theta}_{\theta r}=-\frac{1}{r}\nonumber\\
    &h^{r}_{\theta \theta}=r & & h^{\theta}_{\varphi\varphi}=\sin{\theta}\cos{
    \theta}\nonumber\\
    &h^{\varphi}_{\varphi \theta}=-\frac{\cos{\theta}}{\sin{\theta}} & &h^{r}_{\varphi \varphi}=r\sin^{2}{\theta}\nonumber\\
    &h^{t}_{\theta t}=h^{r}_{\theta r}=h^{\theta}_{\theta \theta}=-\frac{\partial_{\theta}(\Omega)}{\Omega} & &h^{\varphi}_{\theta \varphi}=-\frac{\partial_{\theta}(\Omega)}{\Omega}-\frac{\cos{\theta}}{\sin{\theta}}\\
    &h^{t}_{\varphi t}=h^{r}_{\varphi r}=h^{\theta}_{\varphi \theta}=h^{\varphi}_{\varphi\varphi}=-\frac{\partial_{\varphi}(\Omega)}{\Omega}
\end{align*}
We can now proceed to compute the spin coefficients, the Dirac operator and the Dirac Hamiltonian using expression~\eqref{eq:spin_coefs_simpler}, the spin coefficients and the Christoffel symbols.
    \item[{\rm{(i)}}]\underline{$\nabla^{s}_t$:}
    \\Since~$\gammag^{0}\gamma_{g 0}=\gammag^{r}\gamma_{g r}=\gammag^{\theta}\gamma_{g\theta}=\gammag^{\varphi}\gamma_{g\varphi}=1$, and~$\gamma_{gj}$ and~$\gamma_{gk}$ anti-commute for~$j\neq k$
    \begin{align*}
        & h^{n}_{tk}\gammag^{k}\gamma_{gn}= h^{t}_{tt}\gammag^{t}\gamma_{g t}+h^{r}_{tr}\gammag^{r}\gamma_{g r}+h^{\theta}_{t\theta}\gammag^{\theta}\gamma_{g\theta}+h^{\varphi}_{t\varphi}\gammag^{\varphi}\gamma_{g\varphi}=-4 \frac{\dot{\Omega}}{\Omega}\\
        & \Gamma^{n}_{tk}\gamma_{g}^{k}\gamma_{gn}= \Gamma^{t}_{tt}\gamma_{g}^{t}\gamma_{gt}+\Gamma^{r}_{tr}\gamma_{g}^{r}\gamma_{gr}+\Gamma^{\theta}_{t\theta}\gamma_{g}^{\theta}\gamma_{g\theta}+\Gamma^{\varphi}_{t\varphi}\gamma_{g}^{\varphi}\gamma_{g\varphi}+\Gamma^{t}_{tr}\gamma_{g}^{r}\gamma_{gt}+\Gamma^{r}_{tt}\gamma_{g}^{t}\gamma_{gr}\\
        &+\Gamma^{t}_{t\theta}\gamma_{g}^{\theta}\gamma_{gt}+\Gamma^{\theta}_{tt}\gamma_{g}^{t}\gamma_{g\theta}+\Gamma^{t}_{t\varphi}\gamma_{g}^{\varphi}\gamma_{gt}+\Gamma^{\varphi}_{tt}\gamma_{g}^{t}\gamma_{g\varphi}\\
        &=4\frac{\dot{\Omega}}{\Omega}+\frac{\Omega'}{\Omega}(\gamma_{g}^{r}\gamma_{gt}+\gamma_{g}^{t}\gamma_{gr})+\frac{\partial_{\theta}(\Omega)}{\Omega}(\gamma_{g}^{\theta}\gamma_{gt}+\frac{1}{r^{2}}\gamma_{g}^{t}\gamma_{g\theta})+\frac{\partial_{\varphi}(\Omega)}{\Omega}(\gamma_{g}^{\varphi}\gamma_{gt}+\frac{1}{r^{2}\sin^{2}{\theta}}\gamma_{g}^{t}\gamma_{g\varphi})\\
        &=4\frac{\dot{\Omega}}{\Omega}+2\Big(\frac{\Omega'}{\Omega}\gamma_{g}^{t}\gamma_{gr}+\frac{\partial_{\theta}(\Omega)}{\Omega}\frac{1}{r^{2}}\gamma_{g}^{t}\gamma_{g\theta}+\frac{\partial_{\varphi}(\Omega)}{\Omega}\frac{1}{r^{2}\sin^{2}{\theta}}\gamma_{g}^{t}\gamma_{g\varphi}\Big)\;,
    \end{align*}
    which yields the following spin coefficient
    \begin{align*}
       \nabla^{s}_t&=\pr_t-\frac{1}{4}\Big(h^{n}_{tk}\gammag^ {k}\gamma_{gn}+\Gamma^{n}_{tk}\gamma_{g}^{k}\gamma_{gn}\Big)\\
       &=\pr_t-\frac{1}{2}\gamma_{gt}\Omega^{-2}\Big(\frac{\Omega'}{\Omega}\gamma_{gr}+\frac{\partial_{\theta}(\Omega)}{\Omega}\frac{1}{r^{2}}\gamma_{g\theta}+\frac{\partial_{\varphi}(\Omega)}{\Omega}\frac{1}{r^{2}\sin^{2}{\theta}}\gamma_{g\varphi}\Big)\\
       &=\pr_t+\frac{1}{2}\gamma_{gt}\Big(\frac{\Omega'}{\Omega}\gamma_{g}^{r}+\frac{\partial_{\theta}(\Omega)}{\Omega}\gamma_{g}^{\theta}+\frac{\partial_{\varphi}(\Omega)}{\Omega}\gamma_{g}^{\varphi}\Big)
    \end{align*}
    \item[{\rm{(ii)}}] \underline{$\nabla^{s}_{r}$:}
    \begin{align*}
        & h^{n}_{rk}\gammag^{k}\gamma_{gn}= h^{t}_{rt}\gammag^{t}\gamma_{gt} +h^{r}_{rr}\gammag^{r}\gamma_{gr}+h^{\theta}_{r\theta}\gammag^{\theta}\gamma_{g\theta}+h^{\varphi}_{r\varphi}\gammag^{\varphi}\gamma_{g\varphi}=-4\frac{\Omega'}{\Omega}-\frac{2}{r}\\
        &\Gamma^{n}_{rk}\gamma_{g}^{k}\gamma_{gn} = \Gamma^{r}_{rr}\gamma_{g}^{r}\gamma_{gr}+\Gamma^{\theta}_{r\theta}\gamma_{g}^{\theta}\gamma_{g\theta}+\Gamma^{\varphi}_{r\varphi}\gamma_{g}^{\varphi}\gamma_{g\varphi}+\Gamma^{t}_{rr}\gamma_{g}^{r}\gamma_{gt}+\Gamma^{r}_{tr}\gamma_{g}^{t}\gamma_{gr}+\Gamma^{t}_{rt}\gamma_{g}^{t}\gamma_{gt}\\
        &+\Gamma^{r}_{r\theta}\gamma_{g}^{\theta}\gamma_{gr}+\Gamma^{r}_{r\varphi}\gamma_{g}^{\varphi}\gamma_{gr}+\Gamma^{\theta}_{rr}\gamma_{g}^{r}\gamma_{g\theta}+\Gamma^{\varphi}_{rr}\gamma_{g}^{r}\gamma_{g\varphi}=2\frac{\Omega'}{\Omega}+2\Big(\frac{1}{r}+\frac{\Omega'}{\Omega} \Big)\\
        &+\frac{\dot{\Omega}}{\Omega}(\gammag^{r}\gamma_{gt}+\gammag^{t}\gamma_{gr})+\frac{\partial_{\theta}(\Omega)}{\Omega}\Big(-\frac{1}{r^{2}}\gammag^{r}\gamma_{g\theta}+\gammag^{\theta}\gamma_{gr}\Big)+\frac{\partial_{\varphi}(\Omega)}{\Omega}(-\frac{1}{r^{2}\sin^{2}{\theta}}\gammag^{r}\gamma_{g\varphi}+\gammag^{\varphi}\gamma_{gr})\\
        &=4\frac{\Omega'}{\Omega}+\frac{2}{r}+2\frac{\dot{\Omega}}{\Omega}\gammag^{r}\gamma_{gt}-2\frac{1}{r^{2}}\frac{\partial_{\theta}(\Omega)}{\Omega}\gammag^{r}\gamma_{g\theta}-2\frac{1}{r^{2}\sin^{2}{\theta}}\frac{\partial_{\varphi}(\Omega)}{\Omega}\gammag^{r}\gamma_{g\varphi}
    \end{align*}
    Thus, the corresponding spin coefficient is:
    \begin{align*}
        \nabla^{s}_{r}&=\pr_{r}-\frac{1}{4}\Big(h^{n}_{rk}\gammag^ {k}\gamma_{gn}+\Gamma^{n}_{rk}\gamma_{g}^{k}\gamma_{gn}\Big)\\
        &=\pr_{r}-\frac{1}{2}\gammag^{r}\Big(\frac{\dot{\Omega}}{\Omega}\gamma_{gt}-\frac{1}{r^{2}}\frac{\partial_{\theta}(\Omega)}{\Omega}\gamma_{g\theta}-\frac{1}{r^{2}\sin^{2}{\theta}}\frac{\partial_{\varphi}(\Omega)}{\Omega}\gamma_{g\varphi}\Big)\\
        &=\pr_{r}+\frac{1}{2}\gamma_{gr}\Omega^{-2}\Big(\frac{\dot{\Omega}}{\Omega}\gamma_{gt}-\frac{1}{r^{2}}\frac{\partial_{\theta}(\Omega)}{\Omega}\gamma_{g\theta}-\frac{1}{r^{2}\sin^{2}{\theta}}\frac{\partial_{\varphi}(\Omega)}{\Omega}\gamma_{g\varphi}\Big)\\
        &=\pr_{r}+\frac{1}{2}\gamma_{gr}\Big(\frac{\dot{\Omega}}{\Omega}\gamma_{g}^{t}+\frac{\partial_{\theta}(\Omega)}{\Omega}\gamma_{g}^{\theta}+\frac{\partial_{\varphi}(\Omega)}{\Omega}\gamma_{g}^{\varphi}\Big)
    \end{align*}
\item[{\rm{(iii)}}] \underline{$\nabla^{s}_{\theta}$:}
 \begin{align*}
        & h^{n}_{\theta k}\gammag^ {k}\gamma_{gn}= h^{\theta}_{\theta r}\gammag^ {r}\gamma_{g\theta}+h^{r}_{\theta \theta}\gammag^{\theta}\gamma_{gr}+h^{\varphi}_{\theta\varphi}\gammag^ {\varphi}\gamma_{g\varphi}+h^{t}_{\theta t}\gammag^ {t}\gamma_{gt}+h^{r}_{\theta r}\gammag^ {r}\gamma_{gr}+h^{\theta}_{\theta \theta}\gammag^{\theta}\gamma_{g\theta}\\
        &=-4\frac{\partial_{\theta}(\Omega)}{\Omega}+r\gammag^{\theta}\gamma_{gr}-\frac{1}{r}\gamma_{g}^{r}\gamma_{g\theta}-\frac{\cos{\theta}}{\sin{\theta}}=-4\frac{\partial_{\theta}(\Omega)}{\Omega}+2r\gammag^{\theta}\gamma_{gr}-\frac{\cos{\theta}}{\sin{\theta}}\\ 
        &\Gamma^{n}_{\theta k}\gamma_{g}^{k}\gamma_{gn}=\Gamma^{t}_{\theta\theta}\gamma_{g}^{\theta}\gamma_{gt}+\Gamma^{\theta}_{t \theta}\gamma_{g}^{t}\gamma_{g\theta}+\Gamma^{\theta}_{\theta r}\gamma_{g}^{r}\gamma_{g\theta}+\Gamma^{r}_{\theta\theta}\gamma_{g}^{\theta}\gamma_{gr}+\Gamma^{\varphi}_{\theta\varphi}\gamma_{g}^{\varphi}\gamma_{g\varphi}\\
        &+\Gamma^{t}_{\theta t}\gamma_{g}^{t}\gamma_{gt}+\Gamma^{r}_{\theta r}\gamma_{g}^{r}\gamma_{gr}+\Gamma^{\theta}_{\theta\theta}\gamma_{g}^{\theta}\gamma_{g\theta}+\Gamma^{\theta}_{\theta\varphi}\gamma_{g}^{\varphi}\gamma_{g\theta}+\Gamma^{\varphi}_{\theta\theta}\gamma_{g}^{\theta}\gamma_{g\varphi}\\
        &=4\frac{\partial_{\theta}(\Omega)}{\Omega}+\frac{\cos{\theta}}{\sin{\theta}}+\frac{\dot{\Omega}}{\Omega}\Big(r^{2}\gamma_{g}^{\theta}\gamma_{gt}+\gamma_{g}^{t}\gamma_{g\theta}\Big)\\
        &+\Big(\frac{\Omega'}{\Omega}+\frac{1}{r}\Big)\Big(-r^{2}\gamma_{g}^{\theta}\gamma_{gr}+\gamma_{g}^{r}\gamma_{g\theta}\Big)+\frac{\partial_{\varphi}(\Omega)}{\Omega}\Big(-\frac{1}{\sin^{2}{\theta}}\gamma_{g}^{\theta}\gamma_{g\varphi}+\gamma_{g}^{\varphi}\gamma_{g\theta}\Big)\\
        &=4\frac{\partial_{\theta}(\Omega)}{\Omega}+\frac{\cos{\theta}}{\sin{\theta}}+2\frac{\dot{\Omega}}{\Omega}r^{2}\gamma_{g}^{\theta}\gamma_{gt}-2\Big(\frac{\Omega'}{\Omega}+\frac{1}{r}\Big)r^{2}\gamma_{g}^{\theta}\gamma_{gr}-2\frac{\partial_{\varphi}(\Omega)}{\Omega}\frac{1}{\sin^{2}{\theta}}\gamma_{g}^{\theta}\gamma_{g\varphi}
    \end{align*}
    Hence, it follows that:
    \begin{align*}
        \nabla^{s}_{\theta}&=\pr_{\theta}-\frac{1}{4}\Big(h^{n}_{\theta k}\gammag^ {k}\gamma_{gn}+\Gamma^{n}_{\theta k}\gamma_{g}^{k}\gamma_{gn}\Big) \\
        &=\pr_{\theta}-\frac{1}{2}r^{2}\gamma_{g}^{\theta}\Big(\frac{\dot{\Omega}}{\Omega}\gamma_{gt}-\frac{\Omega'}{\Omega}\gamma_{gr}-\frac{\partial_{\varphi}(\Omega)}{\Omega}\frac{1}{r^{2}\sin^{2}{\theta}}\gamma_{g\varphi}\Big)\\
        &=\pr_{\theta}+\frac{1}{2}\gamma_{g\theta}\Big(\frac{\dot{\Omega}}{\Omega}\gamma_{g}^{t}+\frac{\Omega'}{\Omega}\gamma_{g}^{r}+\frac{\partial_{\varphi}(\Omega)}{\Omega}\gamma_{g}^{\varphi}\Big)
    \end{align*}
\item[{\rm{(iv)}}] \underline{$\nabla^{s}_{\varphi}$:}
\begin{align*}
    & h^{n}_{\varphi k}\gammag^ {k}\gamma_{gn}=h^{\varphi}_{\varphi r}\gammag^ {r}\gamma_{g\varphi}+h^{\theta}_{\varphi\varphi}\gammag^ {\varphi}\gamma_{g\theta}+h^{\varphi}_{\varphi\theta}\gammag^ {\theta}\gamma_{g\varphi}+h^{r}_{\varphi \varphi}\gammag^ {\varphi}\gamma_{gr}+h^{t}_{\varphi t}\gammag^ {t}\gamma_{gt}\\
    &+h^{r}_{\varphi r}\gammag^ {r}\gamma_{gr}+h^{\theta}_{\varphi\theta}\gammag^ {\theta}\gamma_{g\theta}+h^{\varphi}_{\varphi\varphi}\gammag^ {\varphi}\gamma_{g\varphi} \\
    &=-4\frac{\partial_{\varphi}(\Omega)}{\Omega}-\frac{1}{r}\gammag^ {r}\gamma_{g\varphi}+\sin{\theta}\cos{\theta}\gammag^ {\varphi}\gamma_{g\theta}-\frac{\cos{\theta}}{\sin{\theta}}\gammag^ {\theta}\gamma_{g\varphi}+r\sin^{2}{\theta}\gammag^ {\varphi}\gamma_{gr} \\
    &=-4\frac{\partial_{\varphi}(\Omega)}{\Omega}+2\sin{\theta}\cos{\theta}\gammag^ {\varphi}\gamma_{g\theta}+2r\sin^{2}{\theta}\gammag^ {\varphi}\gamma_{gr} \\
    &\Gamma^{n}_{\varphi k}\gamma_{g}^{k}\gamma_{gn}=\Gamma^{\varphi}_{\varphi t}\gamma_{g}^{t}\gamma_{g\varphi}+\Gamma^{t}_{\varphi \varphi}\gamma_{g}^{\varphi}\gamma_{gt}+\Gamma^{r}_{\varphi \varphi}\gamma_{g}^{\varphi}\gamma_{gr}+\Gamma^{\theta}_{\varphi \varphi}\gamma_{g}^{\varphi}\gamma_{g\theta}+\Gamma^{\varphi}_{\varphi r}\gamma_{g}^{r}\gamma_{g\varphi}\\
    &\quad\: +\Gamma^{\varphi}_{\varphi \theta}\gamma_{g}^{\theta}\gamma_{g\varphi}+\Gamma^{t}_{\varphi t}\gamma_{g}^{t}\gamma_{gt}+\Gamma^{r}_{\varphi r}\gamma_{g}^{r}\gamma_{gr}+\Gamma^{\theta}_{\varphi \theta}\gamma_{g}^{\theta}\gamma_{g\theta}+\Gamma^{\varphi}_{\varphi \varphi}\gamma_{g}^{\varphi}\gamma_{g\varphi}\\
    &=4\frac{\partial_{\varphi}(\Omega)}{\Omega}+\frac{\dot{\Omega}}{\Omega}\Big(r^{2}\sin^{2}{\theta}\gamma_{g}^{\varphi}\gamma_{gt}+\gamma_{g}^{t}\gamma_{g\varphi}\Big) \\
    &\quad\: +\Big(\frac{1}{r}+\frac{\Omega'}{\Omega}\Big)\Big(-r^{2}\sin^{2}{\theta}\gamma_{g}^{\varphi}\gamma_{gr}+\gamma_{g}^{r}\gamma_{g\varphi}\Big)\\
&\quad\: +\Big(-\sin{\theta}\cos{\theta}-\sin^{2}{\theta}\frac{\partial_{\theta}(\Omega)}{\Omega}\Big)\Big(\gamma_{g}^{\varphi}\gamma_{g\theta}-\frac{1}{\sin^{2}{\theta}}\gamma_{g}^{\theta}\gamma_{g\varphi}\Big)\\
&=4\frac{\partial_{\varphi}(\Omega)}{\Omega}+2\frac{\dot{\Omega}}{\Omega}r^{2}\sin^{2}{\theta}\gamma_{g}^{\varphi}\gamma_{gt}-2\Big(\frac{1}{r}+\frac{\Omega'}{\Omega}\Big)r^{2}\sin^{2}{\theta}\gamma_{g}^{\varphi}\gamma_{gr}\\
&\quad\: +2\Big(-\sin{\theta}\cos{\theta}-\sin^{2}{\theta}\frac{\partial_{\theta}(\Omega)}{\Omega}\Big)\gamma_{g}^{\varphi}\gamma_{g\theta}
\end{align*}
Thus,
\begin{align*}
    \nabla^{s}_{\varphi}&=\pr_{\varphi}-\frac{1}{4}\Big(h^{n}_{\varphi k}\gammag^ {k}\gamma_{gn}+\Gamma^{n}_{\varphi k}\gamma_{g}^{k}\gamma_{gn}\Big) \\
    &=\pr_{\varphi}-\frac{1}{2}r^{2}\sin^{2}{\theta}\gammag^{\varphi}\Big(\frac{\dot{\Omega}}{\Omega}\gamma_{gt}-\frac{\Omega'}{\Omega}\gamma_{gr}-\frac{1}{r^{2}}\frac{\partial_{\theta}(\Omega)}{\Omega}\gamma_{g\theta}\Big)\\
    &=\pr_{\varphi}+\frac{1}{2}\gamma_{g\varphi}\Big(\frac{\dot{\Omega}}{\Omega}\gamma_{g}^{t}+\frac{\Omega'}{\Omega}\gamma_{g}^{r}+\frac{\partial_{\theta}(\Omega)}{\Omega}\gamma_{g}^{\theta}\Big)
\end{align*}
\item[{\rm{(v)}}] \underline{$D_{g}$:}
Using the previous computations the Dirac operator follows directly
\begin{align*}
        D_{g}&=i\gammag^{j}\nabla^{s}_{j}=i\gammag^{j}\pr_{j}+\frac{3}{2}i\gammag^{j}\frac{\partial_{j}(\Omega)}{\Omega}\;.
\end{align*}
\item[{\rm{(vi)}}] \underline{$H_{g}$:}
Reordering terms in the Dirac equation so that~$H_{g}=i\partial_t$ yields
\begin{align*}
    H_{g}&=-i\gamma_{gt}\Big(\gammag^{\mu}\pr_{\mu}+\frac{3}{2}\gammag^{j}\frac{\partial_{j}(\Omega)}{\Omega}\Big)+m\gamma_{gt}\\
    &=-i\gamma_{\eta 0}\Big(\gamma_\eta^{\mu}\pr_{\mu}+\frac{3}{2}\frac{\pr_{\mu}(\Omega)}{\Omega}\gamma_\eta^{\mu}+\frac{3}{2}\gamma_\eta^{0}\frac{\dot{\Omega}}{\Omega}\Big)+\Omega m\gamma_{\eta 0}=\tilde{H}_\eta+(\Omega-1) m\gamma_{\eta 0}-\frac{3}{2}i\frac{\dot{\Omega}}{\Omega}
\end{align*}
In the rest of the proof we compute the adjoint operators.
    \item[{\rm{(i)}}]\underline{$(\nabla^{s}_{r})^{\ast}$:} We integrate by parts and use that~$d\mu_{N_t}=\sqrt{|\textrm{det}(g|_{N_t})|}d^{3}x=\Omega^{3}r^{2}\sin{\theta} drd\theta d\varphi$
    \begin{align*}
(\psi|\nabla^{s}_{r}\phi)_t &= \int_{N_t}\langle\psi|\gamma_{\eta 0}\partial_{r}\phi\rangle d\mu_{N_t}+\frac{1}{2}\int_{N_t}\langle\psi|\gamma_{\eta 0}\gamma_{gr}\Big(\frac{\dot{\Omega}}{\Omega}\gamma_{g}^{t}+\frac{\partial_{\theta}(\Omega)}{\Omega}\gamma_{g}^{\theta}+\frac{\partial_{\varphi}(\Omega)}{\Omega}\gamma_{g}^{\varphi}\Big)\phi\rangle d\mu_{N_t}\\
&=\int_{N_t}\langle\Big(-\partial_{r}-\frac{2}{r}-3\frac{\Omega'}{\Omega}\Big)\psi|\gamma_{\eta 0}\phi\rangle d\mu_{N_t} \\
&\quad\: -\frac{1}{2}\int_{N_t}\langle\psi|\gamma_{gr}\Big(\frac{\dot{\Omega}}{\Omega}\gamma_{g}^{t}-\frac{\partial_{\theta}(\Omega)}{\Omega}\gamma_{g}^{\theta}-\frac{\partial_{\varphi}(\Omega)}{\Omega}\gamma_{g}^{\varphi}\Big)\gamma_{\eta 0}\phi\rangle d\mu_{N_t}\\
&=\int_{N_t}\langle\Big(-\partial_{r}-\frac{2}{r}-3\frac{\Omega'}{\Omega}\Big)\psi|\gamma_{\eta 0}\phi\rangle d\mu_{N_t} \\
&\quad\: -\frac{1}{2}\int_{N_t}\langle\Big(\frac{\dot{\Omega}}{\Omega}\gamma_{g}^{t}-\frac{\partial_{\theta}(\Omega)}{\Omega}\gamma_{g}^{\theta}-\frac{\partial_{\varphi}(\Omega)}{\Omega}\gamma_{g}^{\varphi}\Big)\gamma_{gr}\psi|\gamma_{\eta 0}\phi\rangle d\mu_{N_t}\\
&=\int_{N_t}\langle\Big(-\partial_{r}-\frac{2}{r}-3\frac{\Omega'}{\Omega}+\frac{1}{2}\gamma_{gr}\Big(\frac{\dot{\Omega}}{\Omega}\gamma_{g}^{t}-\frac{\partial_{\theta}(\Omega)}{\Omega}\gamma_{g}^{\theta}-\frac{\partial_{\varphi}(\Omega)}{\Omega}\gamma_{g}^{\varphi}\Big)\Big)\psi|\gamma_{\eta 0}\phi\rangle d\mu_{N_t}
    \end{align*}
    \item[{\rm{(ii)}}]\underline{$(\nabla^{s}_{\theta})^{\ast}$:}
        \begin{align*}
        &(\psi|\nabla^{s}_{\theta}\phi)_t=\int_{N_t}\langle\psi| \gamma_{\eta 0}\partial_{\theta}\phi\rangle d\mu_{N_t}+\frac{1}{2}\int_{N_t}\langle\psi|\gamma_{\eta 0}\gamma_{g\theta}\Big(\frac{\dot{\Omega}}{\Omega}\gamma_{g}^{t}+\frac{\Omega'}{\Omega}\gamma_{g}^{r}+\frac{\partial_{\varphi}(\Omega)}{\Omega}\gamma_{g}^{\varphi}\Big)\phi\rangle d\mu_{N_t}\\
        &=\int_{N_t}\langle\Big(-\partial_{\theta}-\frac{\cos{\theta}}{\sin{\theta}}-3\frac{\partial_{\theta}(\Omega)}{\Omega}-\frac{1}{2}\Big(\frac{\dot{\Omega}}{\Omega}\gamma_{g}^{t}-\frac{\Omega'}{\Omega}\gamma_{g}^{r}-\frac{\partial_{\varphi}(\Omega)}{\Omega}\gamma_{g}^{\varphi}\Big)\gamma_{g\theta}\Big)\psi|\gamma_{\eta 0}\phi\rangle d\mu_{N_t}\\
        &=\int_{N_t}\langle\Big(-\partial_{\theta}-\frac{\cos{\theta}}{\sin{\theta}}-3\frac{\partial_{\theta}(\Omega)}{\Omega}+\frac{1}{2}\gamma_{g\theta}\Big(\frac{\dot{\Omega}}{\Omega}\gamma_{g}^{t}-\frac{\Omega'}{\Omega}\gamma_{g}^{r}-\frac{\partial_{\varphi}(\Omega)}{\Omega}\gamma_{g}^{\varphi}\Big)\Big)\psi|\gamma_{\eta 0}\phi\rangle d\mu_{N_t}
    \end{align*}
    \item[{\rm{(iii)}}]\underline{$(\nabla^{s}_{\varphi})^{\ast}$:}
    \begin{align*}
        &(\psi|\nabla^{s}_{\varphi}\phi)_t=\int_{N_t}\langle\psi| \gamma_{\eta 0}\partial_{\varphi}\phi\rangle d\mu_{N_t}+\frac{1}{2}\int_{N_t}\langle\psi|\gamma_{\eta 0}\gamma_{g\varphi}\Big(\frac{\dot{\Omega}}{\Omega}\gamma_{g}^{t}+\frac{\Omega'}{\Omega}\gamma_{g}^{r}+\frac{\partial_{\theta}(\Omega)}{\Omega}\gamma_{g}^{\theta}\Big)\phi\rangle d\mu_{N_t}\\
        &=\int_{N_t}\langle\Big(-\partial_{\varphi}-3\frac{\partial_{\varphi}(\Omega)}{\Omega}-\frac{1}{2}\Big(\frac{\dot{\Omega}}{\Omega}\gamma_{g}^{t}-\frac{\Omega'}{\Omega}\gamma_{g}^{r}-\frac{\partial_{\theta}(\Omega)}{\Omega}\gamma_{g}^{\theta}\Big)\gamma_{g\varphi}\Big)\psi|\gamma_{\eta 0}\phi\rangle d\mu_{N_t}\\
        &=\int_{N_t}\langle\Big(-\partial_{\varphi}-3\frac{\partial_{\varphi}(\Omega)}{\Omega}+\frac{1}{2}\gamma_{g\varphi}\Big(\frac{\dot{\Omega}}{\Omega}\gamma_{g}^{t}-\frac{\Omega'}{\Omega}\gamma_{g}^{r}-\frac{\partial_{\theta}(\Omega)}{\Omega}\gamma_{g}^{\theta}\Big)\Big)\psi|\gamma_{\eta 0}\phi\rangle d\mu_{N_t}
    \end{align*}
    % Thus~$(\nabla^{s}_{\varphi})^{\ast}=-\partial_{\varphi}-3\frac{\partial_{\varphi}(\Omega)}{\Omega}+\frac{1}{2}\gamma_{g\varphi}\Big(\frac{\dot{\Omega}}{\Omega}\gamma_{g}^{t}-\frac{\Omega'}{\Omega}\gamma_{g}^{r}-\frac{\partial_{\theta}(\Omega)}{\Omega}\gamma_{g}^{\theta}\Big)$.
    
    \item[{\rm{(iv)}}]\underline{$H_{g}^{\ast}$:}
    Since~$\tilde{H}_\eta$ is symmetric (it is unitarily equivalent to $H_{\eta}$) the adjoint of~$H_{g}$ follows directly
    \begin{align*}
        (\psi|H_{g}\phi)&=\int_{N_t}\langle\psi|\gamma_{\eta 0}\Big(\tilde{H}_\eta+(\Omega-1)m\gamma_{\eta0}-\frac{3}{2}i\frac{\dot{\Omega}}{\Omega}\Big)\phi\rangle d\mu_{N_t}\\
        &=\int_{N_t}\langle\Big(\tilde{H}_\eta+(\Omega-1)m\gamma_{\eta0}+\frac{3}{2}i\frac{\dot{\Omega}}{\Omega}\Big)\psi|\gamma_{\eta 0}\phi\rangle d\mu_{N_t}
    \end{align*}
% \noindent In other words,~$H_{g}^{\ast}= \tilde{H}_\eta+(\Omega-1) m\gamma_{\eta 0}+\frac{3}{2}i\frac{\dot{\Omega}}{\Omega}$.

\noindent From the previous computations we see that
\begin{align*}
&H_{g}+H_{g}^{\ast}=2(\tilde{H}_\eta+(\Omega-1) m\gamma_{\eta 0})\\
 &\nabla^{s}_{r}-(\nabla^{s}_{r})^{\ast}=2\partial_{r}+3\frac{\Omega'}{\Omega}+\frac{2}{r}+\gamma_{gr}\Big(\frac{\partial_{\theta}(\Omega)}{\Omega}\gamma_{g}^{\theta}+\frac{\partial_{\varphi}(\Omega)}{\Omega}\gamma_{g}^{\varphi}\Big)\\
 &\nabla^{s}_{\theta}-(\nabla^{s}_{\theta})^{\ast}=2\partial_{\theta}+3\frac{\partial_{\theta}(\Omega)}{\Omega}+\frac{\cos{\theta}}{\sin{\theta}}+\gamma_{g\theta}\Big(\frac{\Omega'}{\Omega}\gamma_{g}^{r}+\frac{\partial_{\varphi}(\Omega)}{\Omega}\gamma_{g}^{\varphi}\Big)\\
 &\nabla^{s}_{\varphi}-(\nabla^{s}_{\varphi})^{\ast}=2\partial_{\varphi}+3\frac{\partial_{\varphi}(\Omega)}{\Omega}+\gamma_{g\varphi}\Big(\frac{\Omega'}{\Omega}\gamma_{g}^{r}+\frac{\partial_{\theta}(\Omega)}{\Omega}\gamma_{g}^{\theta}\Big)\;,
\end{align*}

and thus the symmetrized Hamiltonian is
\begin{align*}
   A_t &=\frac{1}{2}\{u^{t},\tilde{H}_\eta+(\Omega-1) m\gamma_{\eta 0}\}+\frac{i}{2}\Big\{u^{r},\partial_{r}+\frac{1}{r}+\frac{1}{2}\gamma_{gr}\Big(3\frac{\Omega'}{\Omega}\gammag^{r}\\
    &\quad +\frac{\partial_{\theta}(\Omega)}{\Omega}\gamma_{g}^{\theta}+\frac{\partial_{\varphi}(\Omega)}{\Omega}\gamma_{g}^{\varphi}\Big)\Big\}\\
    &\quad\: +\frac{i}{2}\Big\{u^{\theta},\partial_{\theta}+\frac{\cos{\theta}}{\sin{\theta}}+\frac{1}{2}\gamma_{g\theta}\Big(\frac{\Omega'}{\Omega}\gamma_{g}^{r}+3\frac{\partial_{\theta}(\Omega)}{\Omega}\gammag^{\theta}+\frac{\partial_{\varphi}(\Omega)}{\Omega}\gamma_{g}^{\varphi}\Big)\Big\}\\
    &\quad\: +\frac{i}{2}\Big\{u^{\varphi},\partial_{\varphi}+\frac{1}{2}\gamma_{g\varphi}\Big(\frac{\Omega'}{\Omega}\gamma_{g}^{r}+\frac{\partial_{\theta}(\Omega)}{\Omega}\gamma_{g}^{\theta}+3\frac{\partial_{\varphi}(\Omega)}{\Omega}\gammag^{\varphi}\Big)\Big\}\\
    &=\frac{1}{2}\Big\{u^{t},\tilde{H}_\eta+(\Omega-1) m\gamma_{\eta 0}\Big\} \\
    &\quad\: +\frac{i}{2}\Big\{u^{\mu},\partial_{\mu}+\frac{\partial_{\mu}(\sqrt{|\textrm{det}(\eta)|})}{2\sqrt{|\textrm{det}(\eta)|}}+\frac{1}{2}\gamma_{g\mu}\gamma_{g}^{\nu}\frac{\partial_{\nu}(\Omega)}{\Omega}+\frac{\partial_{\mu}(\Omega)}{\Omega}\Big\}
    \end{align*}
    We now use that
    \begin{align*}
        &\frac{1}{2}\gamma_{g\mu}\gamma_{g}^{\nu}=\frac{1}{4}g^{\rho\nu}([\gamma_{g\mu},\gamma_{g\rho}]+\{\gamma_{g\mu},\gamma_{g\rho}\})=\frac{1}{4}[\gamma_{g\mu},\gamma_{g}^{\nu}]+\frac{1}{2}\delta^{\nu}_{\mu}=\frac{1}{4}\eta^{\rho\nu}[\gamma_{\eta\mu},\gamma_{\eta\rho}]+\frac{1}{2}\delta^{\nu}_{\mu}\\
        &\implies\frac{1}{2}\gamma_{g\mu}\gamma_{g}^{\nu}\frac{\partial_{\nu}(\Omega)}{\Omega}+\frac{\partial_{\mu}(\Omega)}{\Omega}=\frac{1}{4}\frac{\partial_{\nu}(\Omega)}{\Omega}\eta^{\rho\nu}[\gamma_{\eta\mu},\gamma_{\eta\rho}]+\frac{3}{2}\frac{\partial_{\mu}(\Omega)}{\Omega}\;,
    \end{align*}
    in order to simplify the symmetrized Hamiltonian,
    \begin{align*}
    A_t &=\frac{1}{2}\Big\{u^{t},\tilde{H}_\eta+(\Omega-1) m\gamma_{\eta 0}\Big\} \\
    &\quad\:+\frac{i}{2}\Big\{u^{\mu},\partial_{\mu}+\frac{\partial_{\mu}(\sqrt{|\textrm{det}(g|_{N_t})|})}{2\sqrt{|\textrm{det}(g|_{N_t})|}}+\frac{1}{4}\frac{\partial_{\nu}(\Omega)}{\Omega}\eta^{\rho\nu}[\gamma_{\eta\mu},\gamma_{\eta\rho}]\Big\}
\end{align*}
This gives the result.
\end{proof}

%\bibliographystyle{amsplain}
%\bibliography{../felix}
\providecommand{\bysame}{\leavevmode\hbox to3em{\hrulefill}\thinspace}
\providecommand{\MR}{\relax\ifhmode\unskip\space\fi MR }
\providecommand{\MRhref}[2]{%
  \href{http://www.ams.org/mathscinet-getitem?mr=#1}{#2}
}
\providecommand{\href}[2]{#2}

\end{document}